\documentclass[12pt]{article}
\usepackage{enumerate}
\usepackage{amsfonts}
\usepackage{amsmath}
\usepackage{amssymb}
\usepackage{amsthm}
\usepackage{latexsym}
\usepackage{color}

\def\A{\mathfrak{A}}
\def\H{\mathcal{H}}

\def\P{\mathcal{P}}
\def\X{\mathcal{X}}
\def\S{\mathfrak{S}}

\def\F{\mathfrak{F}}
\def\C{\mathfrak{C}}
\def\D{\mathcal{D}}
\def\T{\mathfrak{T}}
\def\B{\mathfrak{B}}

\def\N{\mathbb{N}}

\newcommand{\supp}{\mathrm{supp}}
\newcommand{\rank}{\mathrm{rank}}

\newcommand{\id}{\mathrm{Id}}
\newcommand{\Tr}{\mathrm{Tr}}

\newcommand{\shs}{\hspace{1pt}}
\newcounter{defin}  \newcounter{lemma}  \newcounter{theorem}
\newcounter{proposition} \newcounter{corol}  \newcounter{remark} \newcounter{example}
\newenvironment{lemma}{\par\refstepcounter{lemma}     \textbf{Lemma \thelemma.} }{\rm\par}

\newenvironment{proposition}{\par\refstepcounter{proposition}     \textbf{Proposition \theproposition.}\ }{\rm\par}
\newenvironment{corollary}{\par\refstepcounter{corol}     \textbf{Corollary \thecorol.} }{\rm\par}

\newenvironment{remark}{\par\refstepcounter{remark}     \textbf{Remark \theremark.}}{\rm\par}
\newenvironment{example}{\par\refstepcounter{example}     \textbf{Example \theexample.}}{\rm\par}

\begin{document}

\title{On average output entropy of a quantum channel}

\author{M.E.~Shirokov\footnote{email:msh@mi.ras.ru}\\
Steklov Mathematical Institute, Moscow, Russia}
\date{}
\maketitle
\begin{abstract}
We describe analytical properties of the average output entropy of a quantum channel
as a function of a pair (channel, input ensemble).
In particular, tight semicontinuity bounds for this function with the rank/energy constraints are obtained
by using the modified semicontinuity bounds for the quantum conditional entropy of quantum-classical states
and a special approximation technique.

Several applications are considered.  New semicontinuity and continuity  bounds for the
output Holevo information of a channel as a function of a pair (channel, input ensemble) are obtained.  The semicontinuity bound for the entanglement of formation with
the rank constraint obtained in \cite{LCB} is improved.

In the preliminary part, some results concerning ensembles of quantum states are presented. In particular, a new useful metric on the set of generalized ensembles is proposed and explored. The concept of passive energy of an ensemble introduced here plays an important role in the article.
\end{abstract}

\tableofcontents

\section{Introduction}

When studying quantum systems and channels between such systems, it often becomes necessary to evaluate changes
in some of their characteristics induced by small variations in states of quantum systems and channel disturbances.
It is precisely this necessity that causes a significant interest in obtaining estimates of the modulus of continuity of the most
important entropic and information characteristics of quantum systems and channels. Such estimates, usually called  \emph{(uniform) continuity bounds}, are used essentially  both in the proofs of many theoretical results of quantum information theory and in the development of applications (a brief overview of this area can be found in \cite[the Introduction]{QC}, see also \cite{B&D,BDJ,Capel} and the references therein).

When we talk about a (uniform) continuity bound for a given characteristic $f(\Phi,\rho)$ of a pair (channel $\Phi$, input state $\rho$), we mean an upper bound on the value
of
\begin{equation}\label{D-V}
   |f(\Phi,\rho)-f(\Psi,\sigma)|
\end{equation}
under the condition that the  pairs $(\Phi,\rho)$ and $(\Psi,\sigma)$ are $\varepsilon$-close to each other  w.r.t. some
appropriate metric.

The concept of a (uniform) continuity bound for a  characteristic $f(\Phi,\mu)$ of a pair (channel $\Phi$, input states ensemble $\mu$) is defined similarly.

When dealing with finite-dimensional quantum systems and channels, we always strive to obtain estimates
for the value in (\ref{D-V}) depending only on the divergence bound $\varepsilon$ without additional restrictions on $(\Phi,\rho)$ and $(\Psi,\sigma)$. This is natural,
since almost all important characteristics of finite-dimensional quantum systems and channels are uniformly continuous functions of their arguments (channels, states, ensembles of states).

However, when obtaining continuity bounds for characteristics of infinite-dimensional quantum systems and channels, we have to introduce some constraints on
the possible values of its arguments. This is due to the fact that, as a rule,  characteristics of infinite-dimensional quantum systems and channels are not continuous
functions (and often not well-defined functions) on the set of all its arguments.

Usually constraints are imposed on the input or output states of a quantum channel or on the input or output ensembles of states of a quantum channel.
For example, when we talk about continuity bound for a characteristic $f(\Phi,\rho)$ of a pair (channel $\Phi$, input state $\rho$), we mean an upper bound on the value
of (\ref{D-V}) under the condition that the pairs $(\Phi,\rho)$ are $(\Psi,\sigma)$ are  $\varepsilon$-close to each other and the states
$\rho$ and $\sigma$ (resp.  $\Phi(\rho)$ and $\Psi(\sigma)$) belong to some subset $\C$ of input (resp., output) states.
Typically, the constraint sets $\C$ of the following two forms are used:
\begin{enumerate}[1)]
  \item $\C$ is the set of all input (resp. output) states $\omega$ with the rank not exceeding $r<+\infty$;
  \item $\C$ is the set of all input (resp. output) states $\omega$ with the bounded value of a given energy-type functional on the space
  of all  input (resp. output) states.
\end{enumerate}
In the first case we speak about the input (resp. output) \emph{rank constraint}, in the second one --  about the input (resp. output) \emph{energy-type constraint}.

The first work devoted to the problem of obtaining continuity bounds with energy-type constraint  was Winter's article \cite{W-CB},
where continuity bounds for the von Neumann entropy and the quantum conditional entropy in infinite-dimensional
quantum systems under the energy-type constraints are constructed. This
work stimulated research in this field. As a result, a series of papers has appeared in which
new general methods for quantitative continuity analysis of characteristics
of infinite-dimensional quantum systems and channels are being developed \cite{LCB,QC,B&D,BDJ,Capel,CHI}.

A new step in this direction was made in \cite{LCB}, where the concept of a \emph{semicontinuity bound} was
introduced and concrete examples of semicontinuity bounds for several basic characteristics of quantum and classical systems with the rank/energy-type constraints are constructed.
By a semicontinuity bound for a given characteristic $f(\Phi,\rho)$ of a pair (channel $\Phi$, input state $\rho$), we mean the upper bound on the value
of
\begin{equation}\label{D-V+}
   f(\Phi,\rho)-f(\Psi,\sigma)
\end{equation}
(there is no module here, unlike (\ref{D-V})) under the condition that the  pairs $(\Phi,\rho)$ and $(\Psi,\sigma)$ are $\varepsilon$-close to each other w.r.t. some
metric and \emph{assuming that the rank/energy-type constraint is imposed only
on one of the pairs $(\Phi,\rho)$ are $(\Psi,\sigma)$}.

In \cite{LCB}, the semicontinuity bounds for the von Neumann entropy, the quantum conditional entropy of quantum-classical states
and the entanglement of formation with the rank/energy-type constraint imposed on the state $\rho$ (in terms of (\ref{D-V+})) are obtained.

It is clear that having a semicontinuity bound for some characteristic one can obtain a
continuity bound for this characteristic. However,  semicontinuity bounds allow us to get more flexible estimates which have higher accuracy in special cases
(see the examples in \cite{LCB}). It also turns out that semicontinuity bounds may serve as an important  instrument in proofs of theoretical results in QIT (f.i., the
 semicontinuity bound for the entanglement of formation was essentially used in the proof of Theorem 7 in \cite{Lami-new}.)\medskip

In this article, we describe analytical properties of the average output entropy of a quantum channel
at an  ensemble of input states, i.e. the quantity
\begin{equation*}
  \overline{S}(\Phi(\mu))\doteq\int_{\S(\H_A)}S(\Phi(\rho))\mu(d\rho)
\end{equation*}
considered  as a function of a pair (channel $\Phi$, input ensemble $\mu$).
This quantity is widely used in quantum information theory, since it is an integral part of important characteristics used in QIT (see Section 3).

The main aim of this article (achieved in Section 4) is to obtain accurate semicontinuity bounds for the function $(\Phi,\mu)\mapsto \overline{S}(\Phi(\mu))$, i.e. accurate upper bounds on the quantity
\begin{equation*}
  \overline{S}(\Phi(\mu))-\overline{S}(\Psi(\nu))
\end{equation*}
valid for $\varepsilon$-close pairs $(\Phi,\mu)$ and $(\Psi,\nu)$ with the rank/energy-type constraints
imposed only on the output ensemble $\Phi(\mu)$.

Our technical tools are  the modified semicontinuity bounds for the quantum conditional entropy of quantum-classical states
and a special approximation technique which allows us to extend the results proved for discrete ensembles of states to generalized (continuous) ensembles widely used
in "continuous variable" QIT.

We also consider several applications of the obtained results. In Section 5.1, new (semi)continuity  bounds for the
output Holevo information of a channel are presented and compared with the previous results of this type. In Section 5.2, it is  shown  how the
semicontinuity bounds for the function $(\Phi,\mu)\mapsto \overline{S}(\Phi(\mu))$ can be used to get efficient upper bounds for this function.
Finally, the semicontinuity bound for the entanglement of formation with
the rank constraint obtained in \cite{LCB} is improved in Section 5.3.

\section{Preliminaries}

\subsection{Basic notation}

Let $\mathcal{H}$ be a separable Hilbert space,
$\mathfrak{B}(\mathcal{H})$ the algebra of all bounded operators on $\mathcal{H}$ with the operator norm $\|\cdot\|$ and $\mathfrak{T}( \mathcal{H})$ the
Banach space of all trace-class
operators on $\mathcal{H}$  with the trace norm $\|\!\cdot\!\|_1$. Let
$\mathfrak{S}(\mathcal{H})$ be  the set of quantum states (positive operators
in $\mathfrak{T}(\mathcal{H})$ with unit trace) \cite{H-SCI,N&Ch,Wilde}.

Write $I_{\mathcal{H}}$ for the unit operator on a Hilbert space
$\mathcal{H}$ and  $\id_{\mathcal{\H}}$ for the identity
transformation of the Banach space $\mathfrak{T}(\mathcal{H})$.

The \emph{fidelity} between quantum states $\rho$ and $\sigma$ in $\S(\H)$ can be defined as (cf.\cite{H-SCI,Wilde})
\begin{equation}\label{fidelity}
  F(\rho,\sigma)=\|\sqrt{\rho}\sqrt{\sigma}\|^2_1.
\end{equation}
By Uhlmann's theorem,
$F(\rho,\sigma)=\sup_{\psi}|\langle\varphi|\psi\rangle|^2$, where $\varphi$ is
a given purification of $\rho$ in $\S(\H\otimes\H_R)$, $\H_R$ is a separable
Hilbert space, and the supremum is taken 
over all purifications $\psi$ of $\sigma$ in $\S(\H\otimes\H_R)$ \cite{Uhl}. Thus,
\begin{equation}\label{F-Tn-eq}
\inf_{\psi}\||\varphi\rangle\langle\varphi|-|\psi\rangle\langle\psi|\|_1=2\sqrt{1-\sup_{\psi}|\langle\varphi|\psi\rangle|^2}=2\sqrt{1-F(\rho,\sigma)}.
\end{equation}
When $\dim\H<+\infty$, 
one can take $\H_R\cong\H$ and
replace $"\sup"$ and $"\inf"$ by $"\max"$ and $"\min"$, respectively.

The following well-known inequalities hold (cf.\cite{F-Tn,H-SCI,Wilde})
\begin{equation}\label{F-Tn-ineq}
1-\sqrt{F(\rho,\sigma)}\leq\textstyle\frac{1}{2}\|\rho-\sigma\|_1\leq\sqrt{1-F(\rho,\sigma)}.
\end{equation}

The \emph{Bures distance} between quantum states   $\rho$ and $\sigma$ in $\T_+(\H)$ is defined as
\begin{equation}\label{B-dist}
  \beta(\rho,\sigma)\doteq\inf_{\varphi,\psi}\|\varphi-\psi\|=\inf_{\varphi}\|\varphi-\psi\|=\sqrt{2-2\sqrt{F(\rho,\sigma)}},
\end{equation}
where the first infimum is over all purifications $\varphi$ and $\psi$ of the states $\rho$ and $\sigma$, while the second one is over all purifications of  $\rho$
with a given  purification of $\sigma$ (it is assumed that the dimension of the reference system is not less than $\dim\H$) \cite{H-SCI,Wilde}.

The following relations  between the Bures distance and the  trace-norm distance hold (cf.\cite{H-SCI,Wilde})
\begin{equation}\label{B-d-s-r}
\frac{1}{2}\|\rho-\sigma\|_1\leq\beta(\rho,\sigma)\leq\sqrt{\|\rho-\sigma\|_1}.
\end{equation}

We will use the Mirsky inequality
\begin{equation}
  \sum_{i=0}^{+\infty}\vert\lambda^{\rho}_{i}-\lambda^{\sigma}_{i}\vert\leq \|\rho-\sigma\|_1\label{Mirsky-ineq+}%
\end{equation}
valid for any positive operators $\rho$ and $\sigma$ in $\T(\H)$, where  $\{\lambda^{\rho}_i\}_{i=0}^{+\infty}$
and $\{\lambda^{\sigma}_i\}_{i=0}^{+\infty}$ are  sequence
of eigenvalues of $\rho$ and $\sigma$ arranged in the non-increasing order (taking the multiplicity into account) \cite{Mirsky,Mirsky-rr}.


The \emph{von Neumann entropy} of a quantum state
$\rho \in \mathfrak{S}(\H)$ is  defined by the formula
$S(\rho)=\operatorname{Tr}\eta(\rho)$, where  $\eta(x)=-x\ln x$ if $x>0$
and $\eta(0)=0$. It is a concave lower semicontinuous function on the set~$\mathfrak{S}(\H)$ taking values in~$[0,+\infty]$ \cite{H-SCI,L-2,W}.

The \emph{quantum relative entropy} for two states $\rho$ and
$\sigma$ in $\mathfrak{S}(\mathcal{H})$ is defined as
\begin{equation*}
D(\rho\,\|\shs\sigma)=\sum_i\langle
\varphi_i\vert\,\rho\ln\rho-\rho\ln\sigma\,\vert\varphi_i\rangle,
\end{equation*}
where $\{\varphi_i\}$ is the orthonormal basis of
eigenvectors of the state $\rho$ and it is assumed that
$D(\rho\,\|\sigma)=+\infty$ if $\,\mathrm{supp}\rho\shs$ is not
contained in $\shs\mathrm{supp}\shs\sigma$ \cite{H-SCI,L-2,W}.\footnote{The support $\mathrm{supp}\rho$ of a state $\rho$ is the closed subspace spanned by the eigenvectors of $\rho$ corresponding to its positive eigenvalues.}

The \emph{quantum conditional entropy} (QCE) of a state $\rho$ of a finite-dimensional bipartite system $AB$ is defined as
\begin{equation}\label{ce-def}
S(A\vert B)_{\rho}=S(\rho)-S(\rho_{B}).
\end{equation}
The function $\rho\mapsto S(A\vert B)_{\rho}$ is  concave and takes values in $[-S(\rho_A),S(\rho_A)]$ \cite{H-SCI,Wilde}.

Definition (\ref{ce-def}) remains valid for a state $\rho$ of an infinite-dimensional bipartite system $AB$
with finite marginal entropies
$S(\rho_A)$ and $S(\rho_B)$  (since the finiteness of $S(\rho_A)$ and $S(\rho_B)$ are equivalent to the finiteness of $S(\rho)$ and $S(\rho_B)$).
For a state $\rho$ with finite $S(\rho_A)$ and arbitrary $S(\rho_B)$ one can define the QCE
by the formula
\begin{equation}\label{ce-ext}
S(A\vert B)_{\rho}=S(\rho_{A})-D(\rho\shs\|\shs\rho_{A}\otimes\rho_{B})
\end{equation}
proposed and analysed by Kuznetsova in \cite{Kuz} (the finiteness of $S(\rho_{A})$ implies the finiteness of $D(\rho\shs\|\shs\rho_{A}\otimes\rho_{B})$). The QCE  extented by the above formula to the convex set $\,\{\rho\in\S(\H_{AB})\,\vert\,S(\rho_A)<+\infty\}\,$ possesses all the basic properties of the QCE valid in finite dimensions \cite{Kuz}.\medskip

Let $H$ be a positive (semi-definite)  operator on a Hilbert space $\mathcal{H}$ (we will always assume that positive operators are self-adjoint). Denote by $\mathcal{D}(H)$ the domain of $H$. For any positive operator $\rho\in\T(\H)$ we will define the quantity $\Tr H\rho$ by the rule
\begin{equation}\label{H-fun}
\Tr H\rho=
\left\{\begin{array}{l}
        \sup_n \Tr P_n H\rho\;\; \textrm{if}\;\;  \supp\rho\subseteq {\rm cl}(\mathcal{D}(H))\\
        +\infty\;\qquad\qquad\;\textrm{otherwise}
        \end{array}\right.
\end{equation}
where $P_n$ is the spectral projector of $H$ corresponding to the interval $[0,n]$ and ${\rm cl}(\mathcal{D}(H))$ is the closure of $\mathcal{D}(H)$. If
$H$ is the Hamiltonian (energy observable) of a quantum system described by the space $\H$ then
$\Tr H\rho$ is the mean energy of a state $\rho$.

For any positive operator $H$ the set $\C_{H,E}=\left\{\rho\in\S(\H)\,\vert\,\Tr H\rho\leq E\right\}$
is convex and closed (since the function $\rho\mapsto\Tr H\rho$ is affine and lower semicontinuous). It is nonempty if $E> E_0$, where $E_0$ is the infimum of the spectrum of $H$.

The von Neumann entropy is continuous on the set $\C_{H,E}$ for any $E> E_0$ if and only if the operator $H$ satisfies  the \emph{Gibbs condition}
\begin{equation}\label{H-cond}
  \Tr\, e^{-\beta H}<+\infty\quad\textrm{for all}\;\,\beta>0
\end{equation}
and the supremum of the entropy on this set is attained at the \emph{Gibbs state}
\begin{equation}\label{Gibbs}
\gamma_H(E)\doteq e^{-\beta(E) H}/\Tr e^{-\beta(E) H},
\end{equation}
where $\beta(E)$ is the parameter determined by the equation $\Tr H e^{-\beta H}=E\Tr e^{-\beta H}$ \cite{W}. Condition (\ref{H-cond}) can be valid only if $H$ is an unbounded operator having  discrete spectrum of finite multiplicity. It means, in Dirac's notation, that
\begin{equation}\label{H-form}
H=\sum_{k=0}^{+\infty} E_k \vert\tau_k\rangle\langle\tau_k\vert,
\end{equation}
where
$\mathcal{T}\doteq\left\{\tau_k\right\}_{k=0}^{+\infty}$ is the orthonormal
system of eigenvectors of $H$ corresponding to the \emph{nondecreasing} unbounded sequence $\left\{E_k\right\}_{k=0}^{+\infty}$ of its eigenvalues
and \emph{it is assumed that the domain $\D(H)$ of $H$ lies within the closure $\H_\mathcal{T}$ of the linear span of $\mathcal{T}$}.

We will use the function
\begin{equation}\label{F-def}
F_{H}(E)\doteq\sup_{\rho\in\C_{H,E}}S(\rho)=S(\gamma_H(E)).
\end{equation}
This is a strictly increasing concave function on $[E_0,+\infty)$ \cite{W-CB}.  The Gibbs condition (\ref{H-cond}) is equivalent to the following asymptotic property (cf.~\cite[Section 2.2]{QC}))
\begin{equation}\label{H-cond-a}
  F_{H}(E)=o\shs(E)\quad\textrm{as}\quad E\rightarrow+\infty.
\end{equation}

For example, if $\,H=\hat{N}\doteq a^\dagger a\,$ is the number operator of  one mode quantum oscillator then
\begin{equation}\label{F-N}
F_H(E)=g(E),
\end{equation}
where\footnote{$\,h_2(p)=-p\ln p-(1-p)\ln(1-p)\,$ is the binary entropy.}
\begin{equation}\label{g-def}
  g(x)=(x+1)h_2\!\left(\frac{x}{x+1}\right)=(x+1)\ln(x+1)-x\ln x,\;\, x>0,\quad g(0)=0
\end{equation}

We will often assume that
\begin{equation}\label{star}
  E_0\doteq\inf\limits_{\|\varphi\|=1}\langle\varphi\vert H\vert\varphi\rangle=0.
\end{equation}
In this case the concavity and nonnegativity of $F_H$ imply that (cf.~\cite[Corollary 12]{W-CB})
\begin{equation}\label{W-L}
  xF_H(E/x)\leq yF_H(E/y)\quad  \forall y>x>0.
\end{equation}

\medskip

A \emph{quantum  channel} $\,\Phi$ from a system $A$ to a system
$B$ is a completely positive trace preserving  linear map from
$\mathfrak{T}(\mathcal{H}_A)$ into $\mathfrak{T}(\mathcal{H}_B)$. For any  quantum channel  $\,\Phi$ the Stinespring theorem (cf.\cite{St}) implies existence of a Hilbert space
$\mathcal{H}_E$ called \emph{environment} and  an isometry
$V_{\Phi}:\mathcal{H}_A\rightarrow\mathcal{H}_{BE}\doteq\mathcal{H}_B\otimes\mathcal{H}_E$ such
that
\begin{equation*}
\Phi(\rho)=\mathrm{Tr}_{E}V_{\Phi}\rho V_{\Phi}^{*},\quad
\rho\in\mathfrak{T}(\mathcal{H}_A).
\end{equation*}
The minimal dimension of $\H_E$ is called the \emph{Choi rank} of $\Phi$ \cite{H-SCI,Wilde}.\smallskip

The \emph{diamond norm} of a Hermitian preserving linear map $\Phi$ from $\T(\H_A)$ to $\T(\H_B)$ is defined as
\begin{equation}\label{d-norm}
\|\Phi\|_{\diamond}=\sup_{\rho\in\S(\H_{AR})}\|\Phi\otimes \id_R(\rho)\|_1,
\end{equation}
where $R$ is a quantum system such that $\H_R\cong\H_A$ \cite{Kit,Wilde} (it is also called the norm of complete boundedness \cite{Paul,Wat}).  This norm induces the
metric on the set of quantum channels from $A$ to $B$ called the diamond-norm metric
which is used essentially in the study of finite-dimensional quantum channels and systems.\smallskip

We will also use the "unstabilized" version of the diamond norm of a Hermitian preserving linear map $\Phi$ from $\T(\H_A)$ to $\T(\H_B)$
defined as
\begin{equation}\label{d-norm+}
\|\Phi\|_{1\to1}=\sup_{\rho\in\S(\H_{A})}\|\Phi(\rho)\|_1.
\end{equation}
It is clear that  $\|\Phi\|_{1\to1}\leq\|\Phi\|_{\diamond}$. So, the use of the
metrics induced by this norm in all the semicontinuity bounds for the AOE obtained in Section 4
makes these semicontinuity bounds \emph{sharper}. Another advantage of the unstabilized norm  (\ref{d-norm+})
consists in the simplicity of its estimation in comparison with the diamond norm.\smallskip

The \emph{energy-constrained  diamond norm} of a Hermitian preserving linear map $\Phi$ from $\T(\H_A)$ to $\T(\H_B)$ induced by the operator $H$ is defined as
\begin{equation}\label{ec-d-norm}
\|\Phi\|^H_{\diamond,E}=\sup_{\rho\in\S(\H_{AR}),\Tr H\rho_A\leq E}\|\Phi\otimes \id_R(\rho)\|_1,
\end{equation}
where $R$ is a quantum system such that $\H_R\cong\H_A$ \cite{SCT,W-EBN} (this norm slightly differs from the eponymous norm used in \cite{Lupo,Pir}).
\smallskip

The \emph{strong convergence} of a sequence $\{\Phi_n\}$ of quantum channels to a quantum channel $\Phi_0$  means that $\lim_{n\rightarrow+\infty}\Phi_n(\rho)=\Phi_0(\rho)$  for all $\rho\in\S(\H_A)$ \cite{H-SCI,AQC}. If the operator $H$ has a discrete spectrum of finite multiplicity then the metric induced by the energy-constrained  diamond norm $\|\cdot\|^H_{\diamond,E}$ generates the strong convergence on the set  of quantum channels from $A$ to $B$ \cite[Proposition 3]{SCT}.\smallskip

We will use the following simple \smallskip

\begin{lemma}\label{vsl} \emph{If $f(x)$ and $g(x)$ are lower semicontinuous nonnegative functions on a metric space $X$ and $\{x_n\}_{n}\subset X$ is a sequence converging to $x_0\in X$ such that there exists
$$
\lim_{n\to+\infty}(f+g)(x_n)=(f+g)(x_0)<+\infty
$$
then}
$$
\lim_{n\to+\infty}f(x_n)=f(x_0)<+\infty\quad \textit{and}\quad \lim_{n\to+\infty}g(x_n)=g(x_0)<+\infty.
$$
\end{lemma}\smallskip

\subsection{Ensembles of quantum states: different metrics and the average passive energy}

\subsubsection{Discrete ensembles: basic metrics and their properties}

A finite or
countable collection $\{\rho_{i}\}$ of states
with a probability distribution $\{p_{i}\}$ is conventionally called
\textit{discrete ensemble} and denoted by $\{p_{i},\rho_{i}\}$. For any ensemble $\mu=\{p_{i},\rho_{i}\}$ the state
$\bar{\rho}(\mu)\doteq\sum_{i}p_{i}\rho_{i}$ is called the \emph{average state} of this  ensemble. \smallskip

Speaking about  distances  between ensembles of quantum states it is necessary
to consider separately the cases of \emph{ordered} ensembles and  the case when  ensemble $\{p_{i},\rho_{i}\}$  is considered as a discrete probability measure $\sum_i p_i\delta(\rho_i)$  on the set $\S(\H)$ (where $\delta(\rho)$ is the Dirac measure concentrating at a state $\rho$) rather than ordered (or disordered) collection of states.
To separate this two cases we write $(p_{i},\rho_{i})$ for an ensemble treated as an ordered collection of states
and  $\{p_{i},\rho_{i}\}$ for an ensemble treated as a discrete probability measure $\sum_i p_i\delta(\rho_i)$.

Note that any ordered ensemble $(p_{i},\rho_{i})$ of states in $\S(\H)$ can be identified with the q-c state  $\sum_{i} p_i\rho_i\otimes |i\rangle\langle i|$ in $\S(\H\otimes\H_{R})$
determined by some fixed basis $\{|i\rangle\}$ in a separable Hilbert space $\H_R$.

The  quantity
\begin{equation}\label{D-0-metric}
D_0((p_{i},\rho_{i}),(q_{i},\sigma_{i}))\doteq\frac{1}{2}\sum_i\|\shs p_i\rho_i-q_i\sigma_i\|_1
\end{equation}
is a true metric on the set of all ordered ensembles of quantum states. It coincides (up to factor $1/2$) with the trace norm distance between the corresponding q-c states $\sum_{i} p_i\rho_i\otimes |i\rangle\langle i|$
and $\sum_{i} q_i\sigma_i\otimes |i\rangle\langle i|$. \medskip

From the quantum information point of view it is reasonable to consider an ensemble of states in $\S(\H)$ as a discrete probability measure on the set $\S(\H)$ \cite{O&C,CHI}. We will  use two measures of divergence between ensembles $\mu=\{p_{i},\rho_{i}\}$ and $\nu=\{q_{i},\sigma_{i}\}$
treated as discrete probability measures.

The first one is obtained from the metric $D_0$ by factorization: if we want to identify discrete ensembles corresponding to the same probability measure then it is natural to use the measure of divergence between ensembles $\mu=\{p_i,\rho_i\}$ and $\nu=\{q_i,\sigma_i\}$ defined as
\begin{equation}\label{f-metric}
D_*(\mu,\nu)\doteq \inf_{\mu'\in \mathcal{E}(\mu),\nu'\in \mathcal{E}(\nu)}D_0(\mu',\nu'),
\end{equation}
where $\mathcal{E}(\mu)$ and $\mathcal{E}(\nu)$ are the sets
of all countable ordered  ensembles corresponding to the measures $\sum_i p_i\delta(\rho_i)$ and $\sum_i q_i\delta(\sigma_i)$, respectively. It is easy to describe
all the operations transforming one ensemble from $\mathcal{E}(\mu)$ to any other ensemble from $\mathcal{E}(\mu)$ \cite{CHI}.


It is mentioned  in \cite{CHI} that the divergence measure $D_*$ coincides{\footnote{A strict proof of the coincidence of $D_*$ and $D_{\mathrm{ehs}}$ is presented in Appendix A-2.}} with the EHS\nobreakdash-\hspace{0pt}distance $D_{\mathrm{ehs}}$ between ensembles of quantum states proposed by Oreshkov and Calsamiglia in \cite{O&C} which can be defined for  ensembles $\mu=\{p_i,\rho_i\}$ and $\nu=\{q_j,\sigma_j\}$ as
\begin{equation}\label{ehs-metric}
D_{\mathrm{ehs}}(\mu,\nu)\doteq \frac{1}{2}\inf\sum_{i,j}\|\shs P_{ij}\rho_i-Q_{ij}\sigma_j\|_1,
\end{equation}
where the infimum is over all 2-variate probability distributions $\{P_{ij}\}$ and $\{Q_{ij}\}$
such that $\sum_jP_{ij}=p_i$ for all $i$ and $\sum_iQ_{ij}=q_j$ for all $j$.

It is also mentioned  in \cite{CHI} that the metric $D_*=D_{\mathrm{ehs}}$ generates the weak convergence topology on the set of all discrete ensembles (considered as probability measures). This means that a  sequence $\{\{p^n_i,\rho^n_i\}\}_n$  converges to an ensemble $\{p^0_i,\rho^0_i\}$ with respect to the metric $\,D_*=D_{\mathrm{ehs}}\,$ if and only if
$\;\lim_{n\rightarrow+\infty}\sum_i p^n_if(\rho^n_i)=\sum_i p^0_if(\rho^0_i)\;$
for any continuous bounded function $f$ on $\,\S(\H)$ \cite{Bil+}. Hence, for an arbitrary continuous bounded function $f$ on $\S(\H)$ the function
\begin{equation}\label{f-hat}
\hat{f}:\{p_i,\rho_i\}\mapsto \sum_i p_if(\rho_i)
\end{equation}
is  continuous on the set of all discrete ensembles w.r.t. to the metric $\,D_*=D_{\mathrm{ehs}}$.
In fact, the following useful result is valid.\smallskip

\begin{lemma}\label{UC-l} \cite{O&C} \emph{For an arbitrary uniformly continuous  function $f$ on $\S(\H)$
the function $\hat{f}$ defined in (\ref{f-hat})
is uniformly continuous on the set of all discrete ensembles w.r.t. to the metric $\,D_*=D_{\mathrm{ehs}}$.}
\end{lemma}\smallskip

The metric $\,D_*=D_{\mathrm{ehs}}$ is more adequate for the quantum informational tasks than the metric $D_0$, but it is difficult to compute in general.\footnote{For finite ensembles it can be calculated by a linear programming procedure \cite{O&C}.} It is clear that
\begin{equation}\label{d-ineq}
D_*(\mu,\nu)\leq D_0(\mu,\nu)
\end{equation}
for any ensembles $\mu$ and $\nu$ (which are treated as probability measures in $D_*(\mu,\nu)$ and as ordered ensembles in $D_0(\mu,\nu)$).
But in some cases the values of $\,D_0(\mu,\nu)\,$ and $\,D_*(\mu,\nu)\,$  are close to each other or even coincide. This holds, for example,  if the ensemble $\nu$ is obtained by small perturbations  of the states or probabilities of $\mu$.

The second useful metric is the Kantorovich distance
\begin{equation}\label{K-D-d}
D_K(\mu,\nu)=\frac{1}{2}\inf_{\{P_{ij}\}}\sum_{i,j} P_{ij}\|\rho_i-\sigma_j\|_1
\end{equation}
between ensembles $\mu=\{p_i,\rho_i\}$ and $\nu=\{q_i,\sigma_i\}$, where the infimum is over all 2-variate probability distributions $\{P_{ij}\}$ with the marginals $\{p_i\}$ and $\{q_j\}$, i.e. such that $\sum_jP_{ij}=p_i$ for all $i$ and $\sum_iP_{ij}=q_j$ for all $j$ \cite{O&C,Bog,B&K} (see details in Appendix A-1). By noting that  $D_*=D_{\mathrm{ehs}}$, it is easy to show (cf.\cite{O&C}) that
\begin{equation}\label{d-ineq+}
  D_*(\mu,\nu)\leq D_K(\mu,\nu)
\end{equation}
for any discrete ensembles $\mu$ and $\nu$.

The Kantorovich distance generates  the weak convergence topology on the set of all discrete ensembles (considered as probability measures). So, it is topologically
equivalent to the metric $\,D_*=D_{\mathrm{ehs}}$. Moreover, it follows from (\ref{d-ineq+}) that the claim of Lemma \ref{UC-l} holds with $D_*$ replaced by $D_K$.

Naturally, the question arises about the relationship between
$D_K(\mu,\nu)$ and $D_0(\mu,\nu)$ for any ensembles $\mu$ and $\nu$ (which are treated as probability measures in $D_K(\mu,\nu)$ and as ordered ensembles in $D_0(\mu,\nu)$).
It is not hard to construct examples showing that both inequalities $D_K(\mu,\nu)<D_0(\mu,\nu)$ and $D_K(\mu,\nu)>D_0(\mu,\nu)$ are possible.\smallskip

\begin{example}\label{D-0-K-e} To prove that $D_K(\mu,\nu)$ may be essentially less than $D_0(\mu,\nu)$
one can take any different discrete ordered ensembles  $\mu$  and $\nu$ corresponding to the same probability measure.
Then $D_K(\mu,\nu)=0$ (as $D_K$ is a true metric on the set of all ensembles treated as probability measures), while
$D_0(\mu,\nu)$ may be close to $1$.

To prove that $D_0(\mu,\nu)$ may be less than $D_K(\mu,\nu)$ consider
the ensemble  $\mu$ consisting of the  states $\rho_1\otimes|0\rangle\langle0|$ and $\rho_2\otimes|1\rangle\langle1|$ with equal probabilities $1/2$ and
the singleton ensemble  $\mu$ consisting of the state $\sigma\otimes|0\rangle\langle0|$, where $\rho_1=|0\rangle\langle0|$,
$\sigma=\frac{1}{2}|0\rangle\langle0|+\frac{1}{2}|1\rangle\langle1|$ and $\rho_2$ is arbitrary.  Then
$$
2D_0(\mu,\nu)=\textstyle\|\frac{1}{2}\rho_1-\sigma\|_1+\|\frac{1}{2}\rho_2\|_1=\frac{1}{2}+\frac{1}{2}=1
$$
while it is easy to show (cf.~\cite[Section IV]{O&C}) that
$$
2D_K(\mu,\nu)=\textstyle\frac{1}{2}\|\rho_1-\sigma\|_1+\frac{1}{2}\|\rho_2\otimes|2\rangle\langle2|-\sigma\otimes|1\rangle\langle1|\|_1 =\frac{1}{2}+1=\frac{3}{2}
$$
as the states $\rho_2\otimes|1\rangle\langle1|$ and $\sigma\otimes|0\rangle\langle0|$ are orthogonal.
\end{example}\smallskip

We will use the following simple\smallskip
\begin{lemma}\label{D-0-K} \emph{Let $\mu$ and $\nu$ be discrete ordered ensembles of states in $\S(\H)$ with the same probability distribution. Then
$$
D_K(\mu,\nu)\leq D_0(\mu,\nu),
$$
where $\mu$ and $\nu$ in $D_K(\mu,\nu)$ are treated as probability measures on $\S(\H)$.}
\end{lemma}

\begin{proof} Let $\mu=(p_i,\rho_i)$ and $\nu=(p_i,\sigma_i)$. Since $P_{ij}=p_i\delta_{ij}$ is
a joint probability distribution such that $\sum_jP_{ij}=p_i$ for all $i$ and $\sum_iP_{ij}=p_j$ for all $j$, we conclude
from the definition (\ref{K-D-d}) that
$$
D_K(\mu,\nu)\leq\frac{1}{2}\sum_{i,j} P_{ij}\|\rho_i-\sigma_j\|_1=\frac{1}{2}\sum_{i} p_{i}\|\rho_i-\sigma_i\|_1=D_0(\mu,\nu).
$$
\end{proof}

The claim of Lemma  \ref{D-0-K} can be deduce from the universal  easy computable  upper bound on $D_K$: the inequality
\begin{equation}\label{d-ineq++}
  D_K(\mu,\nu)\leq\frac{1}{2}\sum_i[\shs\min\{p_i,q_i\}\|\shs\rho_i-\sigma_i\|_1+|\shs p_i-q_i|\shs]
\end{equation}
holds for any discrete ensembles $\mu=\{p_i,\rho_i\}$ and $\nu=\{q_i,\sigma_i\}$. Note that the r.h.s. of (\ref{d-ineq++})
is a metric between ordered ensembles $\mu=(p_i,\rho_i)$ and $\nu=(q_i,\sigma_i)$ which is not less than the metric $D_0$ (this is shown in \cite[Appendix B]{O&C}).

To prove (\ref{d-ineq++}) consider
the ensembles $\hat{\mu}=\{p_i, \rho_i\otimes |i\rangle\langle i|\}$ and $\hat{\nu}=\{q_i, \sigma_i\otimes |i\rangle\langle i|\}$ of states in $\S(\H\otimes\H_R)$
determined by some fixed basis $\{|i\rangle\}$ in a separable Hilbert space $\H_R$. It is proved in \cite[Appendix B]{O&C} that
$D_K(\hat{\mu},\hat{\nu})$ coincides with the r.h.s. of (\ref{d-ineq++}).
As the ensembles $\mu=\{p_i,\rho_i\}$ and $\nu=\{q_i,\sigma_i\}$ are  obtained from $\hat{\mu}$ and $\hat{\nu}$ by
applying  the channel $\Tr_R(\cdot)$, the monotonicity of the metric $D_K$ under action of a channel (see \cite[Section IV]{O&C}) implies (\ref{d-ineq++}).

More detailed information about properties of the metrics $\,D_*=D_{\mathrm{ehs}}$ and $D_K$ can be found in \cite{O&C}.

In the next subsection we will see that both metric $D_*$ and $D_K$ have natural extensions to the set of all generalized (continuous) ensembles and that both these extensions generate the weak convergence topology on this set.

\subsubsection{Generalized (continuous) ensembles: the Kantorovich distance and the extended factor metric $D_*$}

In analysis of infinite-dimensional quantum systems and channels the notion of \textit{generalized (continuous) ensemble} defined as a
Borel probability measure on the set of quantum states naturally appears \cite{H-SCI,H-Sh-2}.\footnote{The relations between this definition of a generalized ensemble
and the concept of a generalized ensemble often used in QIT are described in Section 6 below.} We write $\mathcal{P}(\mathcal{H})$ for the set of all Borel probability measures on $\mathfrak{S}(\mathcal{H})$ equipped with the topology of weak convergence
\cite{Bil+,Bog,Par}.\footnote{The weak convergence of a sequence $\{\mu_n\}$ in $\P(\H)$ to a measure $\mu_0\in\P(\H)$ means that
$\,\lim_{n\rightarrow+\infty}\int f(\rho)\mu_n(d\rho)=\int f(\rho)\mu_0(d\rho)\,$
for any continuous bounded function $f$ on $\,\S(\H)$.}
 The set $\mathcal{P}(\mathcal{H})$ is a complete
separable metric space containing the dense subset $\mathcal{P}_0(\mathcal{H})$ of discrete measures (corresponding to discrete ensembles) \cite{Par}. The average state of a generalized
ensemble $\mu \in \mathcal{P}(\mathcal{H})$ is the barycenter of the measure
$\mu $ defined by the Bochner integral
\begin{equation*}
\bar{\rho}(\mu )=\int_{\mathfrak{S}(\mathcal{H})}\rho \mu (d\rho ).
\end{equation*}

For an ensemble $\mu \in \mathcal{P}(\mathcal{H}_{A})$ its image $\Phi(\mu) $
under a quantum channel $\Phi :A\rightarrow B\,$ is defined as the
ensemble in $\mathcal{P}(\mathcal{H}_{B})$ corresponding to the measure $\mu
\circ \Phi ^{-1}$ on $\mathfrak{S}(\mathcal{H}_{B})$, i.e. $\,\Phi (\mu )[%
\mathfrak{S}_{B}]=\mu[\Phi ^{-1}(\mathfrak{S}_{B})]\,$ for any Borel subset $%
\mathfrak{S}_{B}\subseteq \mathfrak{S}(\mathcal{H}_{B})$, where $\Phi ^{-1}(%
\mathfrak{S}_{B})$ is the pre-image of $\mathfrak{S}_{B}$ under the map $%
\Phi $. If $\mu =\{p _{i},\rho _{i}\}$ then it follows from this definition that $\Phi (\mu)=\{p _{i},\Phi(\rho_{i})\}$.\smallskip

The metrics $D_*$ and $D_K$ on the set of all discrete ensembles described in Section 2.2.1 can be extended to
generalized ensembles.
\smallskip

The Kantorovich distance (\ref{K-D-d}) between discrete ensembles is extended to generalized ensembles $\mu$ and $\nu$ by the explicit expression
\begin{equation}\label{K-D-c}
D_K(\mu,\nu)=\frac{1}{2}\inf_{\Lambda\in\Pi(\mu,\nu)}\int_{\S(\H)\times\S(\H)}\|\rho-\sigma\|_1\Lambda(d\rho,d\sigma),
\end{equation}
where $\Pi(\mu,\nu)$ is the set of all Borel probability measures on $\S(\H)\times\S(\H)$ with the marginals $\mu$ and $\nu$ \cite{Bog,B&K} (see details in Appendix A-1).
\smallskip

\begin{example}\label{D-star-0} Let $\mu$ be an arbitrary generalized ensemble in $\P(\H)$  and $\nu$
be the singleton ensemble consisting of a state $\sigma$. Then
\begin{equation*}
D_K(\mu,\nu)=\frac{1}{2}\int_{\S(\H)}\|\rho-\sigma\|_1\mu(d\rho).
\end{equation*}
Indeed, in this case the measure $\mu\times\delta(\sigma)$ is the only element of the set $\Pi(\mu,\nu)$.
\end{example}
\smallskip

Since $\frac{1}{2}\|\rho-\sigma\|_1\leq 1$ for all $\rho$ and $\sigma$, the Kantorovich distance (\ref{K-D-c}) generates the weak convergence topology on $\mathcal{P}(\mathcal{H})$  \cite{Bog}. This shows, in particular, that
\begin{equation}\label{K-D-c+}
D_K(\mu,\nu)=\lim_{n\to+\infty}D_K(\mu_n,\nu_n)
\end{equation}
for any sequences $\{\mu_n\}$ and $\{\nu_n\}$ of ensembles weakly converging, respectively, to $\mu$  and $\nu$. Since
the set of discrete measures in $\P(\H)$ is a dense subset of $\P(\H)$  we may
define the Kantorovich distance between generalized ensembles by the formula (\ref{K-D-c+}) in which
$\{\mu_n\}$ and $\{\nu_n\}$ are any sequences of discrete ensembles weakly converging, respectively, to $\mu$  and $\nu$ (so that
$D_K$ in the r.h.s. of (\ref{K-D-c+}) is the Kantorovich distance between discrete ensembles defined in (\ref{K-D-d})).

The above approximation approach to definition of the Kantorovich distance between generalized ensembles
motivates an idea to extend  the metric $D_*=D_{\mathrm{ehs}}$ defined for discrete ensembles  by the equivalent expressions (\ref{f-metric}) and (\ref{ehs-metric}) to
generalized ensembles $\mu$ and $\nu$ by the formula
\begin{equation}\label{f-metric+}
D_*(\mu,\nu)=\lim_{n\to+\infty}D_*(\mu_n,\nu_n),
\end{equation}
where $\{\mu_n\}$ and $\{\nu_n\} $ are any sequences of discrete ensembles weakly converging, respectively, to $\mu$  and $\nu$.
Of course, it is necessary to verify that the definition (\ref{f-metric+}) is correct, i.e. to prove the following lemma.\smallskip

\begin{lemma}\label{D-lemma} \emph{Let $\mu$  and $\nu$ be arbitrary generalized ensembles in $\P(\H)$.} \emph{A finite limit
in (\ref{f-metric+}) exists for any sequences $\{\mu_n\}$ and $\{\nu_n\}$ of discrete ensembles weakly converging to $\mu$  and $\nu$
and this limit does not depend on the choice of these sequences.}
\end{lemma}

\begin{proof} Since $D_*(\mu_n,\nu_n)\leq1$ for any $n$, to prove the lemma it suffices to  come to a contradiction assuming that
there exist
\begin{equation}\label{f-metric+a}
\lim_{n\to+\infty}D_*(\mu_n,\nu_n)=a\quad \textrm{and} \quad\lim_{n\to+\infty}D_*(\mu'_n,\nu'_n)=b<a
\end{equation}
where $\{\mu_n\}$ and $\{\mu'_n\} $ are any sequences of discrete ensembles weakly converging to $\mu$  and
$\{\nu_n\}$ and $\{\nu'_n\} $ are any sequences of discrete ensembles weakly converging to $\nu$.

By the triangle inequality for the metric $D_*$ we have
\begin{equation}\label{ab-in}
\begin{array}{c}
D_*(\mu_n,\nu_n)\leq D_*(\mu_n,\mu'_n)+D_*(\mu'_n,\nu'_n)+D_*(\nu_n,\nu'_n)\\\\\leq D_K(\mu_n,\mu'_n)+D_*(\mu'_n,\nu'_n)+D_K(\nu_n,\nu'_n)\quad \forall n,
\end{array}
\end{equation}
where the inequality (\ref{d-ineq+}) was used.  Let $\theta$ be either $\mu$ or $\nu$. Since the triangle inequality for the metric $D_K$ shows that  $D_K(\theta_n,\theta'_n)\leq D_K(\theta_n,\theta)+D_K(\theta'_n,\theta)$,
the weak convergence of the sequences $\{\theta_n\}$ and $\{\theta'_n\}$ to the ensemble $\theta$  implies that
$D_K(\theta_n,\theta'_n)=o(1)$ as $n\to+\infty$  (because $D_K$ generates the weak convergence topology).

Thus, the inequality (\ref{ab-in}) shows that $a\leq b$  contrary to the assumption (\ref{f-metric+a}).
\end{proof}

In is clear the function $(\mu,\nu)\mapsto D_*(\mu,\nu)$ on $\P(\H)\times\P(\H)$ satisfies the triangle inequality (as it satisfies this inequality
for discrete ensembles). If $\mu=\nu$ then by using inequality (\ref{d-ineq+}) and by noting that $D_K(\mu_n,\nu_n)=o(1)$ as $n\to+\infty$
for any sequences $\{\mu_n\}$ and $\{\nu_n\}$ of discrete ensembles weakly converging to $\mu=\nu$ (see the proof of Lemma \ref{D-lemma}) it is easy to show that
$D_*(\mu,\nu)=0$. Thus, to prove that $D_*$ is a true metric on  $\P(\H)$ it suffices to show that $\mu=\nu$ if $D_*(\mu,\nu)=0$.

Assume that $D_*(\mu,\nu)=0$ for some different generalized ensembles  $\mu$ and $\nu$ (i.e. different probability measures $\mu$ and $\nu$ on $\S(\H)$). Then $D_*(\mu_n,\nu_n)=o(1)$ as $n\to+\infty$ for
any sequences $\{\mu_n\}$ and $\{\nu_n\}$ of discrete ensembles weakly converging to $\mu$ and $\nu$ correspondingly. Since $\mu$ and $\nu$ are
different, there exists  an uniformly continuous bounded function $f$ on $\S(\H)$  such that
$$
\int_{\S(\H)}f(\rho)\mu(d\rho)>\int_{\S(\H)}f(\rho)\nu(d\rho)
$$
and hence there is $\,\varepsilon>0\,$ such that
$$
\int_{\S(\H)}f(\rho)\mu_n(d\rho)-\int_{\S(\H)}f(\rho)\nu_n(d\rho)>\varepsilon
$$
for all $n$ large enough. At the same time, for any $\delta$ there is $n_{\delta}$ such that $D_*(\mu_n,\nu_n)\leq\delta$ for all
$n\geq n_{\delta}$. So, we obtain a contradiction to the claim of Lemma \ref{UC-l}.\smallskip

Thus, we may formulate the following\smallskip

\begin{proposition}\label{D-star} A) \emph{The function $(\mu,\nu)\mapsto D_*(\mu,\nu)$ defined in (\ref{f-metric+}) is a true metric
on the set $\P(\H)$ of all  generalised ensembles   of states in $\S(\H)$
which is bounded from above by the  Kantorovich metric $D_K$ defined in (\ref{K-D-c}).}\footnote{It seems that the metric $D_*(\mu,\nu)$ is a partial case (or analog) of one of the general type metrics used in the modern measure theory. I would be grateful for any comments and references.} \smallskip

B) \emph{The metric $D_*$ generates the weak
convergence topology on $\P(\H)$.}\smallskip

C) \emph{For an arbitrary uniformly continuous function $f$ on $\S(\H)$
the function
\begin{equation*}
\hat{f}(\mu)=\int_{\S(\H)}f(\rho)\mu(d\rho)
\end{equation*}
is uniformly continuous on the set $\P(\H)$  w.r.t. to the metric $\,D_*$.}
\end{proposition}

\begin{proof} A) The first part of this claim is proved before the proposition. The inequality $D_*\leq D_K$
follows from the same inequality for discrete ensembles (inequality (\ref{d-ineq+})) and the representations (\ref{f-metric+}) and (\ref{K-D-c+}).
\smallskip

B) Since $D_*\leq D_K$ and $D_K$ generates the weak
convergence topology on $\P(\H)$, this claim follows from claim C proved below due to the Portmanteau theorem \cite{Bil+}.   \smallskip

C) This claim  follows from the representations (\ref{f-metric+}) and Lemma \ref{UC-l}.
\end{proof}

\begin{lemma}\label{a-s-l} \emph{The inequalities
\begin{equation}\label{a-s-l+}
\|\bar{\rho}(\mu)-\bar{\rho}(\nu)\|_1\leq2D_X(\mu,\nu),\quad  D_X=D_*,D_K,
\end{equation}
hold for arbitrary generalized ensembles $\mu$ and $\nu$ in $\P(\H)$.}
\end{lemma}

\begin{proof} By the inequality $D_*\leq D_K$ it suffices to consider the case $D_X=D_*$.

If  $\mu$ and $\nu$ are discrete ensembles in $\P(\H)$ then the required inequality
is proved in \cite[Section V-B]{O&C} (because $D_*=D_{\mathrm{ehs}}$).

If  $\mu$ and $\nu$ are generalized ensembles in $\P(\H)$ then by using the construction from
the proof of Lemma 1 in \cite[the Appendix]{H-Sh-2} on can obtain sequences $\{\mu_n\}$ and $\{\nu_n\}$
of discrete ensembles in $\P(\H)$ weakly converging, respectively, to $\mu$ and $\nu$ such that
$\bar{\rho}(\mu_n)=\bar{\rho}(\mu)$ and $\bar{\rho}(\nu_n)=\bar{\rho}(\nu)$ for all $n$. By the
above remark, the inequality (\ref{a-s-l+}) holds with $D_X(\mu,\nu)$
replaced by $D_*(\mu_n,\nu_n)$ for all $n$.

Since $D_*(\mu_n,\nu_n)$ tends to $D_*(\mu,\nu)$ as $n\to+\infty$  (by the definition of $D_*$),
the inequality (\ref{a-s-l+}) holds with $D_X(\mu,\nu)=D_*(\mu,\nu)$.
\end{proof}

Although the metric $D_*$ for generalised ensembles is defined in a non-explicit way, sometimes it may be more useful
than the Kantorovich metric $D_K$. In particular, in some cases the metric $D_*$ (unlike $D_K$) can be simply estimated.\footnote{Another advantage of the metric $D_*$ (in comparison with $D_K$) is shown in Proposition \ref{nu-l} below.}\smallskip

\begin{example}\label{D-star-e} Let $\rho(x)$ and $\sigma(x)$ be uniformly continuous $\S(\H)$-valued functions
on a complete separable metric space $X$. Let $\pi$ be a Borel nonnegative measure $\pi$ on $X$.

Consider generalized ensembles $\mu$ and $\nu$ in $\P(\H)$ determined, respectively, by the images under the maps $x\mapsto\rho(x)$ and
$x\mapsto\sigma(x)$ of some Borel probability measures $\mu_{X}$ and $\nu_{X}$ on $X$ absolutely continuous w.r.t. $\pi$.
In terms of Section 6 below, these ensembles can be denoted by $\{\mu_{X},\rho(x)\}_{x\in X}$ and $\{\nu_{X},\sigma(x)\}_{x\in X}$, respectively.

Then
\begin{equation}\label{D-star-e+}
D_*(\mu,\nu)\leq \frac{1}{2}\int_X\|p(x)\rho(x)-q(x)\sigma(x)\|_1\pi(dx),
\end{equation}
where $p(x)$ and $q(x)$ are the probability densities of $\mu_{X}$ and $\nu_{X}$  w.r.t. $\pi$.

Note that the r.h.s. of (\ref{D-star-e+}) can be treated as a "continuous" version of the $D_0$-metric defined in (\ref{D-0-metric}).

To prove (\ref{D-star-e+}) assume for each natural $n$ that $X=\bigcup_iX_i^n$ is a countable decomposition of $X$ into disjoint Borel subsets with the diameter not exceeding $1/n$.
Consider the sequences of discrete ensembles $\mu_n=\{p_i^n,\rho_i^n\}$ and $\nu_n=\{q_i^n,\sigma_i^n\}$, where
$p_i^n=\int_{X_i^n}p(x)\pi(dx)$, $q_i^n=\int_{X_i^n}q(x)\pi(dx)$ and $\rho_i^n=(1/p_i^n)\int_{X_i^n}p(x)\rho(x)\pi(dx)$, $\sigma_i^n=(1/q_i^n)\int_{X_i^n}q(x)\sigma(x)\pi(dx)$ (the Bochner integrals) for all $i$. If $p_i^n=0$ (resp. $q_i^n=0$) we assume that
$\rho_i^n$ (resp. $\sigma_i^n$) is an arbitrary state.

By the basic property of the Bochner integral we have
\begin{equation}\label{D-star-e++}
2D_*(\mu_n,\nu_n)\leq \sum_i \|
p_i^n\rho_i^n-q_i^n\sigma_i^n\|_1\leq\int_X\|p(x)\rho(x)-q(x)\sigma(x)\|_1\pi(dx)\quad \forall n.
\end{equation}
So, to prove (\ref{D-star-e+}) it suffices to show that the sequences $\{\mu_n\}$ and
$\{\nu_n\}$ weakly converge, respectively, to the  ensembles $\mu$ and $\nu$.

By symmetry, we have to consider only the sequence $\{\mu_n\}$.  By the Portmanteau theorem \cite{Bil+} to prove the weak convergence of $\{\mu_n\}$ to $\mu$
one should show that $\sum_ip_i^nf(\rho_i^n)$ tends to $\int_{\S(\H)}f(\rho)\mu(d\rho)$ for
any uniformly continuous function $f$ on $\S(\H)$.

For each $n$ and $i$ the set $\rho(X_i^n)$ lies within  some ball $\B_i^n\subset\S(\H)$ with  the diameter $\omega_\rho(1/n)$,
where  $\omega_{\rho}$ is the modulus of continuity of the function $\rho(x)$. So, since the set $\B_i^n$ is convex,
the state  $\rho_i^n$ belongs to $\B_i^n$. It follows that
$|f(\rho(x))-f(\rho_i^n)|\leq\omega_f(\omega_\rho(1/n))$ for all $x\in X_i^n$, where $\omega_f$ is the modulus of continuity of the function $f$. Hence, we have
$$
\begin{array}{c}
\displaystyle\left|\sum_ip_i^nf(\rho_i^n)-\int_{\S(\H)}f(\rho)\mu(d\rho)\right|=\left|\sum_ip_i^nf(\rho_i^n)-\int_{X}f(\rho(x))p(x)\pi(dx)\right|\\\\
\displaystyle\leq \sum_i \int_{X_i^n}|f(\rho_i^n)-f(\rho(x))|\,p(x)\pi(dx)\leq \sum_i\omega_f(\omega_\rho(1/n))\mu_{X}(X_i^n)=\omega_f(\omega_\rho(1/n)).
\end{array}
$$
Since $\omega_f(\omega_\rho(1/n))$ tends to zero as $n\to+\infty$, the required relation is proved.

Note that the simple upper bound (\ref{D-star-e+}) is not valid for the Kantorovich distance $D_K(\mu,\nu)$.
This can be deduced from the fact the Kantorovich distance for discrete ensembles is not majorized by the $D_0$-metric (Example \ref{D-0-K-e}).
\end{example}
\medskip

An important technical task is  the following: for a given ensemble $\mu$ and a state $\sigma$ close to the average state
$\bar{\rho}(\mu)$ of $\mu$ to find an ensemble $\nu$ with the average state $\sigma$ which
is as close to the ensemble $\mu$ as possible (w.r.t. some adequate metric). The following lemma shows that this can be done
for both metrics $D_*$ and $D_K$, but the use of $D_*$ gives a particular advantage (in comparison with $D_K$).\smallskip

\begin{proposition}\label{nu-l} \emph{Let $\mu$ be a generalized ensemble in $\P(\H)$ and $\sigma$ be a state in $\S(\H)$.}\smallskip

A) \emph{If $\,\delta\geq \sqrt{1-F(\bar{\rho}(\mu),\sigma)}$ then there is an ensemble $\nu\in\P(\H)$ such that $\bar{\rho}(\nu)=\sigma$ and $\,D_*(\mu,\nu)\leq \delta$. The following claims are valid (separately or simultaneously):}\footnote{$F(\cdot,\cdot)$ is the fidelity defined in (\ref{fidelity}).}\smallskip
\begin{itemize}
  \item \emph{if $\,\delta>\sqrt{1-F(\bar{\rho}(\mu),\sigma)}$ and the ensemble $\mu$ is discrete then the ensemble $\nu$ can be chosen discrete;}
  \item \emph{if the ensemble $\mu$  consists of pure states then the ensemble $\nu$ can be chosen consisting of pure states.}
\end{itemize}

B) \emph{If $\,\delta\geq\sqrt{\varepsilon(2-\varepsilon)}$, where $\varepsilon=\frac{1}{2}\|\bar{\rho}(\mu)-\sigma\|_1$,  then  there is an ensemble $\nu\in\P(\H)$ such that $\bar{\rho}(\nu)=\sigma$ and $\,D_K(\mu,\nu)\leq \delta$.} \emph{If the ensemble $\mu$ is discrete then the ensemble $\nu$ can be chosen discrete with the same probability distribution.}

\end{proposition}\smallskip

\begin{proof} In the proofs of both claims A and B it suffices to consider the case $\delta<1$.

A) Assume first that  $\mu=\{p_i,\rho_i\}$ is a discrete ensemble\footnote{The proof of claims of Proposition \ref{nu-l} for discrete ensemble uses the standard arguments \cite{Nielsen,W-CB}.} in $\P(\H)$. Then the Schrodinger-Gisin-Hughston-Jozsa-Wootters theorem (cf.\cite{Schr,Gisin,HJW}) implies that  $p_i\rho_i=\Tr_R [I_{\H}\otimes M_i]\shs\hat{\rho}\,$ for all $i$, where $\hat{\rho}$ is a pure state in $\S(\H\otimes\H_R)$
such that $\rho =\Tr_R\shs\hat{\rho}$ and $\{M_i\}$ is the corresponding POVM
on appropriate Hilbert space $\H_R$.

Let $n$ be a natural number such that $1/n<1-\delta$. As $F(\bar{\rho}(\mu),\sigma)\geq 1-\delta^2$, it follows from (\ref{F-Tn-eq}) that there is a pure state  $\hat{\sigma}_n$ in $\S(\H\otimes\H_R)$ such that $\sigma=\Tr_R\shs\hat{\sigma}_n$ and $\,\frac{1}{2}\|\shs\hat{\rho}-\hat{\sigma}_n\|_1\leq\delta+1/n$.

Let $\nu_n=\{q^n_i,\sigma^n_i\}$, where $q^n_i=\Tr[I_{\H}\otimes M_i]\shs\hat{\sigma}_n\,$ and $\sigma^n_i=(1/q^n_i)\Tr_R [I_{\H}\otimes M_i]\shs\hat{\sigma}_n\,$ 
(if $q^n_i=0$ then we assume that $\sigma^n_i$ is any given state). Then $\bar{\rho}(\nu_n)=\sigma$.
Consider the channel $\,\Lambda(\varrho)=\sum_i[\Tr M_i\varrho]|i\rangle\langle i|\,$ from $R$ to an infinite-dimensional system $E$, where  $\{|i\rangle\}$ is a basis
in $\H_E$. Then
$$
\sum_i p_i\rho_i\otimes|i\rangle\langle i|=I_{\H}\otimes\Lambda(\hat{\rho})\quad \textrm{and}\quad \sum_i q^n_i\sigma^n_i\otimes|i\rangle\langle i|=I_{\H}\otimes\Lambda(\hat{\sigma}_n)
$$
and hence inequality (\ref{d-ineq}) implies that
$$
\begin{array}{c}
2D_*(\mu,\nu_n)\leq2D_0((p_i,\rho_i),(q^n_i,\sigma^n_i))\doteq\sum_i\|p_i\rho_i-q^n_i\sigma^n_i\|_1=\left\|I_{\H}\otimes\Lambda(\hat{\rho})-I_{\H}\otimes\Lambda(\hat{\sigma}_n)\right\|_1\\\\
\leq\|\shs\hat{\rho}-\hat{\sigma}_n\|_1\leq2(\delta+1/n)
\end{array}
$$
by the monotonicity of the trace norm under action of the channel $I_{\H}\otimes\Lambda$. \smallskip

Since $\,\bar{\rho}(\nu_n)=\sigma\,$ for all $n$, Proposition 2 in \cite{H-Sh-2} implies
the relative compactness of the sequence $\{\nu_n\}$ in the weak convergence topology. Thus, there is a subsequence $\{\nu_{n_k}\}$ weakly converging to some
ensemble $\nu\in\P(\H)$ such that $\bar{\rho}(\nu)=\sigma$. Then
$$
D_*(\mu,\nu)=\lim_{n\to+\infty}D_*(\mu,\nu_n)\leq \lim_{n\to+\infty}(\delta+1/n)=\delta.
$$

If the ensemble $\mu$ consists of pure states then the above POVM $\{M_i\}$ can be chosen consisting of 1-rank operators. It follows that
the above sequence $\{\nu_n\}$ consists  of pure states ensembles (since $\hat{\sigma}_n$ is a pure state for each $n$).
Thus, the above limit ensemble $\nu$ consists of pure states, because the set of all pure states ensembles is closed in the weak convergence topology.\footnote{In other words, the set of all probability measures supported by pure states is closed in the weak convergence topology. This follows from the closedness of the set of pure states and the basic results concerning the weak convergence of probability measures \cite{Bil+,Bog}.} \smallskip

If $\,\delta>\sqrt{1-F(\bar{\rho}(\mu),\sigma)}$ then we may take any discrete ensemble $\nu_n$ for $n$ large enough in the role of $\nu$ (instead of the limit point of the sequence $\{\nu_n\}$ that may be a non-discrete ensemble).

Assume now that $\mu$ is a generalized ensemble in $\P(\H)$. Then there is a sequence $\{\mu_n\}$ of discrete ensembles weakly converging to $\mu$.
By the first part of the proof for each $n$ there exists an ensemble $\nu_n\in\P(\H)$ such that $\bar{\rho}(\nu_n)=\sigma$ and $D_*(\mu_n,\nu_n)\leq \delta_n$,
where $\delta_n \doteq \sqrt{1-F(\bar{\rho}(\mu_n),\sigma)}$. Since $\,\bar{\rho}(\nu_n)=\sigma\,$ for all $n$, by applying Proposition 2 in \cite{H-Sh-2} again one can  show the existence of a subsequence $\{\nu_{n_k}\}$ weakly converging to some
ensemble $\nu\in\P(\H)$ such that $\bar{\rho}(\nu)=\sigma$. Then
$$
D_*(\mu,\nu)=\lim_{n\to+\infty}D_*(\mu_n,\nu_n)\leq \lim_{n\to+\infty}\delta_n=\delta,
$$
as the sequence $\{\bar{\rho}(\mu_n)\}$ tends to $\bar{\rho}(\mu)$ by the continuity of the barycenter map.

If $\mu$ is an ensemble of pure states (i.e. $\mu$ is a measure supported by the set of pure states) then the above sequence $\{\mu_n\}$ can be chosen consisting of pure states  ensembles.
By the first  part of the proof the above sequence $\{\nu_n\}$ can be chosen consisting of  pure states ensembles as well. Thus, the
above limit ensemble $\nu$ consists of pure states, because the set of all pure states ensembles is closed in the weak convergence topology.

B) This claim is proved by applying the arguments from the proof of part A  with the use of Winter's "quantum coupling lemma" presented in Proposition 5 in \cite{W-CB}
instead of the relation (\ref{F-Tn-eq}). By this result the condition $\,\delta\geq\sqrt{\varepsilon(2-\varepsilon)}$
allows us to choose (within the proof of part A) a state $\hat{\sigma}$ in $\S(\H\otimes\H_R)$ such that $\sigma=\Tr_R\shs\hat{\sigma}$, $\,\frac{1}{2}\|\shs\hat{\rho}-\hat{\sigma}\|_1\leq\delta$ and $\hat{\sigma}_R=\hat{\rho}_R$.  The last property  implies
that the  ordered ensemble $\nu$ defined via the state $\hat{\sigma}$ by the same way as the ensemble $\nu_n$ via the state $\hat{\sigma}_n$ 
(in the proof of part A) has the same probability distribution as the ordered ensemble $\mu$, i.e. $q_i=p_i$ for all $i$. So, in follows from Lemma \ref{D-0-K} that $D_K(\mu,\nu)\leq D_0(\mu,\nu)\leq\frac{1}{2}\|\shs\hat{\rho}-\hat{\sigma}\|_1\leq \delta$.

If $\mu$ is a generalized ensemble then we repeat the corresponding arguments from the proof of part A with the use of $D_K$ instead of $D_*$  assuming that $\,\delta_n=\sqrt{\varepsilon_n(2-\varepsilon_n)}$.
\end{proof}

The advantage of the metric $D_*$ mentioned before the proposition consists in the fact that claim A
of Proposition \ref{nu-l} is \emph{stronger} than claim B by the following reasons:
\begin{itemize}
  \item the condition
$\,\delta\geq \sqrt{1-F(\bar{\rho}(\mu),\sigma)}$ is weaker than  the condition $\,\delta\geq\sqrt{\varepsilon(2-\varepsilon)}$
(this follows from the first inequality in (\ref{F-Tn-ineq}));
  \item claim B does not assert that the ensemble $\nu$ consists of pure states provided that the ensemble $\mu$ consists of pure states (because the  state $\hat{\sigma}$
   produced by Winter's quantum coupling lemma in the proof of claim B is not pure in general).
\end{itemize}
We don't know how to prove claim A of Proposition \ref{nu-l} with $D_*$ replaced by $D_K$ (because $D_*$ may be strictly less than $D_K$).\smallskip

\begin{remark}\label{st-r}
Proposition \ref{nu-l} gives a quantitative description of the openness on the barycenter map
$$
\P(\H)\ni\mu\mapsto \bar{\rho}(\mu)\in\S(\H),
$$
in terms of the metrics $D_*$ and $D_K$ generating the topology on $\P(\H)$. This property is one of the equivalent characterisations of the \emph{stability} of the convex set $\S(\H)$ \cite{EM,stability}.
Claim A of Proposition \ref{nu-l} also quantifies the openness on the map
$$
\widehat{\P}(\H)\ni\mu\mapsto \bar{\rho}(\mu)\in\S(\H),
$$
where $\widehat{\P}(\H)$ is the closed subset of $\P(\H)$ consisting of pure states ensembles. Hence, \emph{$D_*$  is the only metric 
on $\P(\H)$ that allows us to quantify the openness on the last map.}
\end{remark}

\subsubsection{Average passive energy of ensemble of quantum states}

Assume that $H$ is a positive operator on a Hilbert space $\H$ having the form (\ref{H-form}) with the discrete  spectrum $\{E_i\}_{i=0}^{+\infty}$ arranged in the non-decreasing order. We will treat $H$ as the Hamiltonian (energy observable) of a quantum system described by the space $\H$ and write $E_H(\rho)$ for the (mean) energy $\Tr H \rho$ of a state $\rho$ defined in (\ref{H-fun}).
Define the \emph{passive energy} of
any positive operator $\rho$ in $\T(\H)$ as follows
\begin{equation}\label{p-e-def}
E_{H}^{\rm psv}(\rho)\doteq\sum_{i=0}^{+\infty}E_i\lambda_i^{\rho}\leq+\infty,
\end{equation}
where $\{\lambda^{\rho}_i\}_{i=0}^{+\infty}$ is the sequences of eigenvalues of $\rho$ arranged in the non-increasing order. We use the term "passive", since $E_{H}^{\rm psv}(\rho)$ for a state $\rho$
coincides with the energy $E_{H}(\rho_H^{\rm psv})$ of any passive state $\rho_H^{\rm psv}$ corresponding to $\rho$ \cite{Erg-2,p-state}.\footnote{$\rho_H^{\rm psv}=\sum_{i=0}^{+\infty}\lambda^{\rho}_i|\tau_i\rangle\langle\tau_i|$, where  $\{\tau_i\}_{i=0}^{+\infty}$ is the basic from representation (\ref{H-form}) of $H$.}

It is well known that  $E_{H}^{\rm psv}(\rho)\leq E_{H}(\rho)$, the difference $E_{H}(\rho)-E_{H}^{\rm psv}(\rho)$ called \emph{ergotropy} has been studied a lot in recent years \cite{Erg-1,Erg-2,Erg-3}.

Assume that the operator $H$ satisfies the Gibbs condition (\ref{H-cond}). Since $S(\rho_H^{\rm psv})=S(\rho)$ for any state $\rho$ and  $\Tr H\rho_H^{\rm psv}=E_{H}^{\rm psv}(\rho)$, the well known facts
(described in Section 2.1) imply that
\begin{equation}\label{P-B}
S(\rho)\leq F_H(E_{H}^{\rm psv}(\rho))\doteq S(\gamma_H(E_{H}^{\rm psv}(\rho))),
\end{equation}
where $\gamma_H(E)$ is the Gibbs state (\ref{Gibbs}).

By using Mirsky inequality (\ref{Mirsky-ineq+}) it is easy to show that the function $\rho\mapsto E_{H}^{\rm psv}(\rho)$ is lower semicontinuous (and, hence, Borel) on $\S(\H)$. So, for
any generalized ensemble $\mu$ we may define the \emph{average passive energy} of $\mu$ as
\begin{equation}\label{p-a-e-def}
\bar{E}_{H}^{\rm psv}(\mu)\doteq \int_{\S(\H)}E_{H}^{\rm psv}(\rho)\mu(d\rho)\leq+\infty.
\end{equation}
In particular, if $\mu=\{p_k,\rho_k\}$ is any discrete ensemble then
\begin{equation*}
\bar{E}_{H}^{\rm psv}(\mu)\doteq\sum_kp_kE_{H}^{\rm psv}(\rho_k).
\end{equation*}

The lower semicontinuity of the nonnegative function $\rho\mapsto E_{H}^{\rm psv}(\rho)$ on $\S(\H)$ mentioned before
and the Portmanteau theorem \cite{Bil+} imply the following \smallskip

\begin{lemma}\label{PE-LSC}
\emph{The functional $\,\mu\mapsto\bar{E}_{H}^{\rm psv}(\mu)\,$ is lower semicontinuous on $\P(\H)$.}\footnote{$\P(\H)$ is the set of all generalized ensembles
of states in $\S(\H)$ equipped with the weak convergence topology (see Section 2.2.2).}
\end{lemma}\smallskip

By using upper bound (\ref{P-B}) and definition (\ref{p-a-e-def}) we obtain
\begin{equation}\label{P-B+}
\int_{\S(\H)}S(\rho)\mu(d\rho)\leq F_H(\bar{E}_{H}^{\rm psv}(\mu)).
\end{equation}

It is clear that $\bar{E}_{H}^{\rm psv}(\mu)\leq\bar{E}_{H}(\mu)\doteq E_{H}(\bar{\rho}(\mu))$. For example, if $\mu$ is an ensemble of pure states
then $\bar{E}_{H}^{\rm psv}(\mu)=E_0\doteq \inf_{\|\varphi\|=1}\langle\varphi|H|\varphi\rangle$, while  $\bar{E}_{H}(\mu)$ may be equal to $+\infty$.
\smallskip

\begin{example}\label{APE-ex} Assume that $\H$ is a Hilbert space describing  one-mode quantum oscillator and $H=\hat{N}\doteq a^{\dagger}a$ is the number operator on $\H$.
Let $\mu=\{\pi_{\mu},\rho_\zeta\}_{\zeta\in\mathbb{C}}$ be the  ensemble determined
by the probability measure
$\pi_{\mu}$ on $\mathbb{C}$ with the density $p_{\mu}(\zeta)=\frac{1}{\pi N}\exp(-|\zeta|^2/N)$ and the family of states $\rho_\zeta=D(\zeta)\gamma_{\hat{N}}(N_0)D^*(\zeta)$,
where $D(\zeta)$ is the displacement (unitary) operator and $\gamma_{\hat{N}}(N_0)$ is
the Gibbs state corresponding to the
"energy" (number of quanta) $N_0$ \cite[Ch.12]{H-SCI}.\footnote{Here we use the concept of ensemble
described at the begin of Section 6 below. Its relation with the the above-described concept of ensemble is explained Section 6.}

Since any state $\rho_\zeta$ is unitarily equivalent to the state $\gamma_{\hat{N}}(N_0)$ we have
$E_{H}^{\rm psv}(\rho_\zeta)=N_0$. Hence, $\bar{E}_{H}^{\rm psv}(\mu)=N_0$. At the same time,  it is easy to show that $\bar{\rho}(\mu)=\gamma_{\hat{N}}(N+N_0)$.
So, the average energy $E_{H}(\bar{\rho}(\mu))$ of $\mu$ is equal to $N+N_0>N_0$. $\Box$
\end{example}\smallskip

The notion of the average passive energy of ensemble allows us to essentially improve the semicontinuity bound for the quantum conditional entropy of quantum-classical states
with the energy-type constraint obtained in \cite{LCB} (see Lemma \ref{qce-qc-1} in Section 2.3).

\subsection{Semicontinuity bounds for the quantum conditional entropy of quantum-classical states}

In this subsection we consider modified versions of the (semi)continuity bounds for the quantum conditional entropy (QCE) of quantum-classical states
obtained in \cite{Wilde-CB,LCB}.

States $\rho$ and $\sigma$ of a bipartite quantum system $AB$ having the form
\begin{equation}\label{qc-states}
\rho=\sum_{k} p_k\, \rho_k\otimes \vert k\rangle\langle k\vert \quad \textrm{and}\quad \sigma=\sum_{k} q_k\, \sigma_k\otimes \vert k\rangle\langle k\vert,
\end{equation}
where $(p_k,\rho_k)$ and $(q_k,\sigma_k)$ are ordered ensembles of states in $\S(\H_A)$ and $\{\vert k\rangle\}$ is a fixed orthonormal basis in $\H_B$,
are  called \emph{quantum-classical} (briefly, \emph{q-c states}).

By using definition (\ref{ce-ext}) one can show (see the proof of Corollary 3 in \cite{Wilde-CB}) that
\begin{equation}\label{ce-rep}
S(A\vert B)_{\rho}=\sum_kp_kS(\rho_k)\leq+\infty\quad \textrm{and} \quad S(A\vert B)_{\sigma}=\sum_kq_kS(\sigma_k)\leq+\infty.
\end{equation}
The expressions in (\ref{ce-rep}) allow us to define the QCE on the set $\S_{\mathrm{qc}}$ of all q-c states
(including the q-c states $\rho$ with $\,S(\rho_A)=+\infty\,$ for which definition (\ref{ce-ext}) does not work).\smallskip

In 2019  Wilde  obtained the optimal continuity bound
for the QCE on the set of quantum-classical states presented in \cite{Wilde-CB}. By using this continuity bound one can  obtain the following
result essentially used in this article. \smallskip

\begin{lemma}\label{Wilde-SCB}  \emph{Let $A$ and $B$ be arbitrary quantum systems. Let
$\rho$ and $\sigma$ be q-c states in $\S(\H_{AB})$ defined in (\ref{qc-states}). Assume that
$\,\rank \rho_k\leq r<+\infty$ for all $\,k$ such that $\,p_k\neq0$. }\smallskip

\emph{If $\,\varepsilon=\frac{1}{2}\|\rho-\sigma\|_1\leq 1-1/r\,$ then}\footnote{In (\ref{Wilde-SCB+}) and in the below semicontinuty bounds it is assumed that the l.h.s. may be equal to $-\infty$.}
\begin{equation}\label{Wilde-SCB+}
S(A|B)_{\rho}-S(A|B)_{\sigma}\leq \varepsilon \ln(r-1)+h_2(\varepsilon).
\end{equation}

\emph{If $\,\varepsilon=\frac{1}{2}\|\rho-\sigma\|_1\geq 1-1/r\,$ then (\ref{Wilde-SCB+}) holds with the r.h.s.
replaced by $\,\ln r$.}
\end{lemma}\medskip

\textbf{Note:}  no rank constraint is imposed on the state $\sigma$ in Lemma \ref{Wilde-SCB}.

\begin{proof} Since for any $k$ such that $p_k=0$ we may choose the state $\rho_k$ arbitrarily,
we may assume that $\rho$ and $\sigma$ are q-c states such that
$\,\sup_k\rank \rho_k=r\,$  and
$$
\|\rho-\sigma\|_1=\sum_k\|p_k\rho_k-q_k\sigma_k\|_1=2\varepsilon.
$$

For each $k$ let $\{\lambda^{\rho_k}_i\}_{i=0}^{+\infty}$ and $\{\lambda^{\sigma_k}_i\}_{i=0}^{+\infty}$ be sequences of eigenvalues of $\rho_k$ and $\sigma_k$
arranged in the non-increasing order. Let $\{\varphi_i\}_{i=0}^{+\infty}$ be a fixed orthonormal basis in $\H_A$. Consider the q-c states
\begin{equation*}
\tilde{\rho}=\sum_{k} p_k\, \tilde{\rho}_k\otimes \vert k\rangle\langle k\vert\quad \textrm{and} \quad\tilde{\sigma}=\sum_{k} q_k\tilde{\sigma}_k\otimes \vert k\rangle\langle k\vert,
\end{equation*}
where $\,\tilde{\rho}_k=\sum_{i=0}^{+\infty} \lambda^{\rho_k}_i \vert \varphi_i\rangle\langle \varphi_i\vert\,$ and
$\,\tilde{\sigma}_k=\sum_{i=0}^{+\infty} \lambda^{\sigma_k}_i \vert \varphi_i\rangle\langle \varphi_i\vert\,$ for each $k$.\smallskip

By Mirsky inequality (\ref{Mirsky-ineq+}) we have
\begin{equation}\label{ce-one}
\|\tilde{\rho}-\tilde{\sigma}\|_1=\sum_k\|p_k\tilde{\rho}_k-q_k\tilde{\sigma}_k\|_1\leq\sum_k\|p_k\rho_k-q_k\sigma_k\|_1=2\varepsilon.
\end{equation}
The representations in (\ref{ce-rep}) imply that
\begin{equation}\label{ce-two}
S(A\vert B)_{\rho}=\sum_kp_kS(\rho_k)=\sum_kp_kS(\tilde{\rho}_k)=S(A\vert B)_{\tilde{\rho}}
\end{equation}
and
\begin{equation}\label{ce-two+}
S(A\vert B)_{\sigma}=\sum_kq_kS(\sigma_k)=\sum_kq_kS(\tilde{\sigma}_k)=S(A\vert B)_{\tilde{\sigma}}.
\end{equation}

Consider the quantum channel
$$
\S(\H_A)\ni\varrho\mapsto\Phi(\varrho)=\sum_{j=0}^{+\infty}W_j\varrho W_j^*\in\S(\H_A),
$$
where $\,W_j=\sum_{t=0}^{r-1}\vert \varphi_t\rangle\langle\varphi_{t+rj}\vert\,$ is a partial isometry such that
$W_jW^*_j$ and $W^*_jW_j$ are the projectors on the linear spans of the sets $\,\{\varphi_0,...,\varphi_{r-1}\}\,$ and $\,\{\varphi_{rj},...,\varphi_{rj+r-1}\}\,$
correspondingly.

Let $\tilde{\tilde{\sigma}}_k=\Phi(\tilde{\sigma}_k)$ for all $k$. Then all the states $\tilde{\rho}_k$ and $\tilde{\tilde{\sigma}}_k$ are supported by the $r$-dimensional subspace
generated by the vectors $\varphi_0,...,\varphi_{r-1}$. Since $\tilde{\rho}_k=\Phi(\tilde{\rho}_k)$ by the condition $\rank \rho_k\leq r$ for all $k$, we have
$$
\sum_k\|p_k\tilde{\rho}_k-q_k\tilde{\tilde{\sigma}}_k\|_1=\sum_k\|p_k\Phi(\tilde{\rho}_k)-q_k\Phi(\tilde{\sigma}_k)\|_1\leq\sum_k\|p_k\tilde{\rho}_k-q_k\tilde{\sigma}_k\|_1\leq2\varepsilon
$$
due to monotonicity of the trace norm under action of a channel, where that last inequality follows from (\ref{ce-one}).

By the above inequality it follows from (\ref{ce-two}) and (\ref{ce-two+}) that the claim of the lemma can be derived from Proposition 1 in \cite{Wilde-CB} by showing that $\,S(A|B)_{\tilde{\tilde{\sigma}}}\leq S(A|B)_{\tilde{\sigma}}$, where
$\tilde{\tilde{\sigma}}=\sum_{k} q_k\tilde{\tilde{\sigma}}_k\otimes \vert k\rangle\langle k\vert$. By the representation in (\ref{ce-rep}) this can be done by proving that
\begin{equation}\label{l-in}
 S(\tilde{\tilde{\sigma}}_k)\leq S(\tilde{\sigma}_k)\quad \forall k.
\end{equation}

For each $k$ we have
$$
\tilde{\tilde{\sigma}}_k=\sum_{j=0}^{+\infty}\lambda^{\sigma_k}_{jr}\vert \varphi_0\rangle\langle\varphi_0\vert +...+\sum_{j=0}^{+\infty}\lambda^{\sigma_k}_{jr+t}\vert \varphi_t\rangle\langle\varphi_t\vert +...
+\sum_{j=0}^{+\infty}\lambda^{\sigma_k}_{jr+r-1}\vert \varphi_{r-1}\rangle\langle\varphi_{r-1}\vert,
$$
and hence the state $\tilde{\sigma}_k$ is majorized by the state $\tilde{\tilde{\sigma}}_k$ in the sense of \cite[Section 13.5]{S-T-1}.
This implies  (\ref{l-in}).
\end{proof}

To obtain a necessary modification of the semicontinuity bound
presented in Proposition 3B in \cite{LCB} we use the
notion of the average passive energy of quantum states ensemble introduced in Section 2.2.3. \smallskip

\begin{lemma}\label{qce-qc-1} \emph{Let $A$ and $B$ be arbitrary  quantum systems. Let $H$ be a positive operator on $\H_A$ having the form (\ref{H-form}) and satisfying conditions (\ref{H-cond}) and (\ref{star}).}\smallskip

\emph{If $\,\rho=\sum_{k} p_k\, \rho_k\otimes \vert k\rangle\langle k\vert$ is a q-c state in $\S(\H_{AB})$  s.t. $\bar{E}_{H}^{\rm psv}(\{p_k,\rho_k\})\leq E<+\infty$  then
\begin{equation}\label{qce-qc-2}
   S(A\vert B)_{\rho}-S(A\vert B)_{\sigma}\leq \varepsilon F_H((E-E^{\rm psv}_{H,\shs\varepsilon}(\rho))/\varepsilon)+g(\varepsilon)\leq \varepsilon F_H(E/\varepsilon)+g(\varepsilon)
\end{equation}
for any q-c state $\,\sigma=\sum_{k} q_k\, \sigma_k\otimes \vert k\rangle\langle k\vert$ in $\S(\H_{AB})$ such that $\,\frac{1}{2}\|\rho-\sigma\|_1\leq \varepsilon$, where
\begin{equation*}
E^{\rm psv}_{H,\shs\varepsilon}(\rho)\doteq \sum_{k}E^{\rm psv}_H([p_k\rho_k-\varepsilon I_{A}]_+),
\end{equation*}
$\,E^{\rm psv}_H([p_k\rho_k-\varepsilon  I_{A}]_+)$ is the passive energy of the positive part of the Hermitian operator $\,p_k\rho_k-\varepsilon  I_{A}$ w.r.t. $H$ (defined in (\ref{p-e-def})). The l.h.s. of (\ref{qce-qc-2}) may be equal to $-\infty$.}\smallskip
\end{lemma}\smallskip

The function $\varepsilon\mapsto \varepsilon F_H(E/\varepsilon)+g(\varepsilon)$ is nondecreasing and tends to zero as $\varepsilon\to0$. This follows from inequality (\ref{W-L}) and the equivalence of (\ref{H-cond}) and (\ref{H-cond-a}).

\begin{proof} Let $\tilde{\rho}=\sum_{k}p_k\tilde{\rho}_k\otimes|k\rangle\langle k|$ and
$\tilde{\sigma}=\sum_{i}q_k\tilde{\sigma}_k\otimes|k\rangle\langle k|$ be the q-c states constructed in the proof of Lemma \ref{Wilde-SCB} with the basis $\{\varphi_i\}$ replaced by the basic $\{\tau_i\}$ from representation (\ref{H-form}) of $H$.
By the construction  $\Tr H\tilde{\rho}_A=\bar{E}_{H}^{\rm psv}(\{p_k,\rho_k\})\leq E$ and $\Tr H[p_k\tilde{\rho}_k-\varepsilon I_{A}]_+=E^{\rm psv}_H([p_k\rho_k-\varepsilon I_{A}]_+)$ for all $k$. Thus,
by applying  the semicontinuity bound from Proposition 3B in \cite{LCB} to the q-c states $\tilde{\rho}$ and
$\tilde{\sigma}$ with the use of (\ref{ce-one}), (\ref{ce-two}) and (\ref{ce-two+}) we obtain the claim of the lemma.
\end{proof}

\section{Average Output Entropy of a channel: analytical properties and approximation}

The Average Output Entropy (AOE) of a quantum channel $\Phi:A\to B$ at a (generalised) ensemble $\mu$ of states in $\S(\H_A)$ is defined as
\begin{equation*}
  \overline{S}(\Phi(\mu))\doteq\int_{\S(\H_A)}S(\Phi(\rho))\mu(d\rho).
\end{equation*}

This quantity is widely used in quantum information theory, since several important  characteristics are expressed via AOE. For example, the output Holevo
information
\begin{equation}\label{H-Q}
\chi(\Phi(\mu))=\int_{\mathfrak{S}(\mathcal{H}_A)} D(\mathrm{\Phi}(\rho)\shs \|\shs \mathrm{\Phi}(\bar{\rho}(\mu)))\mu (d\rho)
\end{equation}
of a channel $\Phi:A\to B$ at a (generalised) input ensemble $\mu$ can be expressed as
\begin{equation}\label{H-Q+}
\chi(\Phi(\mu))=S(\Phi(\bar{\rho}(\mu)))-\overline{S}(\Phi(\mu))
\end{equation}
provided that $S(\Phi(\bar{\rho}(\mu)))<+\infty$, where $\bar{\rho}(\mu)$ is the average state of $\mu$ \cite{H-SCI,H-Sh-2}.
Another examples are the discrete and continuous versions of the Entanglement of Formation (EoF) of a state $\rho$ of a bipartite
system $AB$ which can be defined as:
\begin{equation}\label{EoF}
E_{F,d}(\rho)=\inf_{\mu\in\widehat{\P}_0(\H_{AB}),\bar{\rho}(\mu)=\rho}\overline{S}(\Phi(\mu)),\quad E_{F,c}(\rho)=\inf_{\mu\in\widehat{\P}(\H_{AB}),\bar{\rho}(\mu)=\rho}\overline{S}(\Phi(\mu)),
\end{equation}
where $\Phi:AB\to A$ is the partial trace channel $\varrho\mapsto \varrho_A$, $\widehat{\P}_0(\H_{AB})$ is the set of all discrete ensembles of pure states in $\S(\H_{AB})$ and $\widehat{\P}(\H_{AB})$ is the set of all generalised ensembles of pure states in $\S(\H_{AB})$ \cite{Bennett,Lami-new,EM}.\medskip

The continuous version of the EoF is a partial case of the\emph{ convex closure of the output entropy}
of a quantum channel $\Phi$ from $A$ to $B$ at a state $\rho$ in $\S(\H_A)$ defined as
$$
 \widehat{S}_{\Phi}(\rho)=\inf_{\mu\in\widehat{\P}(\H_A),\bar{\rho}(\mu)=\rho}\overline{S}(\Phi(\mu)),
$$
which plays important role in analysis of the classical capacity of a channel \cite{H-SCI,AQC}.\smallskip

Below we describe  analytical properties of the function $(\Phi,\mu)\mapsto\overline{S}(\Phi(\mu))$
on the set $\F(A,B)\times \P(\H_A)$, where $\F(A,B)$ is the set of all channels from $A$ to $B$ equipped with the strong convergence topology
and $\P(\H_A)$ is the a set of all generalised ensembles of states in $\S(\H_A)$ equipped with the weak convergence topology (see Sections 2.1,2.2).
\smallskip

\begin{proposition}\label{AOE-a-p} \emph{The nonnegative function $(\Phi,\mu)\mapsto\overline{S}(\Phi(\mu))$ is lower semicontinuous
on $\,\F(A,B)\times \P(\H_A)$. If $\,\{\Phi_n\}$ is a sequence in $\,\F(A,B)$ strongly converging to $\Phi_0$
and $\,\{\mu_n\}$ is a sequence in $\P(\H_A)$ weakly converging to $\mu_0$ such that
$$
\lim_{n\to+\infty}S(\Phi_n(\bar{\rho}(\mu_n)))=S(\Phi_0(\bar{\rho}(\mu_0)))<+\infty
$$
then
$$
\lim_{n\to+\infty}\overline{S}(\Phi_n(\mu_n))=\overline{S}(\Phi_0(\mu_0))<+\infty.
$$}
\end{proposition}

\begin{proof} The lower semicontinuity of the function $(\Phi,\mu)\mapsto\overline{S}(\Phi(\mu))$
on $\F(A,B)\times \P(\H_A)$ is proved in \cite[Lemma 1]{AQC}. The last claim of the proposition is proved by using
the first one and Lemma  \ref{vsl} in Section 2.1. Indeed, it follows from (\ref{H-Q+}) that
$$
S(\Phi(\bar{\rho}(\mu)))=\chi(\Phi(\mu))+\overline{S}(\Phi(\mu)),\quad (\Phi,\mu)\in\F(A,B)\times \P(\H_A),
$$
provided that $S(\Phi(\bar{\rho}(\mu)))<+\infty$, and the function $(\Phi,\mu)\mapsto\chi(\Phi(\mu))$ is lower semicontinuous
on $\,\F(A,B)\times \P(\H_A)$ \cite[Proposition 15]{CMC}.
\end{proof}

The following approximation lemma plays a basic role in this article.\smallskip

\begin{lemma}\label{b-lem} \emph{Let $\Phi$ be a quantum channel from $A$ to $B$ and $\mu_0$ be an arbitrary  generalized ensemble
in $\P(\H_A)$ such that $\overline{S}(\Phi(\mu_0))<+\infty$. Let $H_A$ and $H_B$ be  positive operators on the spaces $\H_A$ and $\H_B$ correspondingly.  Then there exists a sequence $\{\mu_n\}$ of discrete ensembles in $\P(\H_A)$ weakly converging to the ensemble
$\mu_0$ such that $\;\supp\bar{\rho}(\mu_n)\subseteq \supp\bar{\rho}(\mu_0)\,$ for all $n$ and the following limit relations hold
\begin{equation}\label{b-lem+}
\lim_{n\to+\infty}\overline{S}(\Phi(\mu_n))=\overline{S}(\Phi(\mu_0)),
\end{equation}
\begin{equation}\label{b-lem++}
\lim_{n\to+\infty}\Tr H_A\bar{\rho}(\mu_n)=\Tr H_A\bar{\rho}(\mu_0)\leq+\infty,
\end{equation}
\begin{equation}\label{b-lem+++}
\lim_{n\to+\infty}\bar{E}_{H_B}^{\rm psv}(\Phi(\mu_n))=\bar{E}_{H_B}^{\rm psv}(\Phi(\mu_0))\leq+\infty,
\end{equation}
where $\bar{E}_{H_B}^{\rm psv}$ denotes the average passive energy with respect to $H_B$ defined in (\ref{p-a-e-def}).}

\emph{If $\,\rank\Phi(\rho)\leq r_B<+\infty$ for $\mu_0$-almost all $\rho$ then the above sequence $\{\mu_n\}$ can be chosen
in such a way that $\,\rank\Phi(\rho)\leq r_B$ for $\mu_n$-almost all $\rho$ for all $n$.}
\end{lemma}

\begin{proof} Let $\S_0=\{\rho\in\S(\H_A)\,|\,\supp\rho\subseteq \supp\bar{\rho}(\mu_0)\}$ be a closed convex subset of $\S(\H_A)$.
To prove the main claim and the last claim of the lemma simultaneously assume that
$\S_*=\{\rho\in\S_0\,|\,\rank\Phi(\rho)\leq r_B\}$ if $r_B<+\infty$ and that $\S_*=\S_0$ in the general case.
Since the set $\S_*$ is closed, it is clear that it contains the support of the measure $\mu_0$ \cite{Par}. Since any probability
measure on the complete separable metric space $\S(\H_A)$ is tight~\cite{Bil+,Par}, for any $n\in\N$ there exists
a compact subset $\C_{n}$ of $\S_*$ such that
$\mu_0(\C_{n})>1-1/n$. By the compactness of the set $\C_{n}$ there is a
decomposition $\C_{n}=\bigcup_{i=1}^{m(n)}\A^{n}_i$, where
$\{\A^{n}_i\}_{i=1}^{m(n)}$ is a finite collection of disjoint
Borel subsets with the diameter less than $1/n$. Without loss of
generality we may assume that $\mu_0(\A^{n}_i)>0$ for all $i$ and
$n$. By construction the closure $\overline{\A^{n}_i}$ of $\A^{n}_i$ is a  compact set contained in some closed ball $\B^{n}_i$ of diameter $1/n$ for all
$i$ and $n$.

Consider the nonnegative lower semicontinuous (and hence Borel) function
$$
F(\rho)=S(\Phi(\rho))+a\Tr H_A\rho+b\bar{E}_{H_B}^{\rm psv}(\Phi(\rho))
$$
on the set $\S(\H_A)$, where $a=1$  if $\,\Tr H_A\bar{\rho}(\mu_0)<+\infty\,$   and $a=0$  otherwise,  $b=1$  if $\,\bar{E}_{H_B}^{\rm psv}(\Phi(\mu_0))<+\infty\,$  and  $\,b=0\,$ otherwise. The lower semicontinuity of the first and second summands are well known, the lower semicontinuity of the third one follows from the lower semicontinuity of the passive energy as a function of a state  mentioned in Section 2.2.3.
By the assumption
$$
\int_{\S(\H_A)}F(\rho)\,\mu_0(d\rho)<+\infty
$$
and hence the function $F(\rho)$ is  $\mu_{0}$-almost
everywhere finite. Since the function $F(\rho)$ is lower semicontinuous it
achieves its \textit{finite} minimum on the compact set
$\overline{\A^{n}_i}$ of \textit{positive} measure at some
state $\rho^{n}_i\in\overline{\A^{n}_i}$ for each $n$ and $i$. Consider the ensemble (probability measure)
$$
\mu_n=c_n^{-1}\sum_{i=1}^{m(n)}\mu_0(\A^{n}_i)\delta(\rho^{n}_i),
$$
where
$\delta(\rho^{n}_i)$ is the Dirac (single-atom) measure concentrated at the state $\rho^{n}_i$ and $c_n=\mu_0(\C_n)$. Since  $\rho^{n}_i\in\overline{\A^{n}_i}\subset\S_*\subset\S_0$ for all $n$ and $i$ and the set $\S_0$ is
convex, we have  $\supp\bar{\rho}(\mu_n)\subseteq \supp\bar{\rho}(\mu_0)$ for all $n$.

Show first that the sequence $\{\mu_n\}$ weakly converges to the ensemble
$\mu_0$. Take an arbitrary  uniformly continuous function $f$ on the set $\S(\H_A)$. Let $M_f=\sup_{\rho\in\S(\H_A)}f(\rho)$
and $\,\omega_f(\varepsilon)=\sup_{\rho,\sigma\in\S(\H_A),\|\rho-\sigma\|_1\leq\varepsilon}|f(\rho)-f(\sigma)|\,$ be the modulus of continuity of $f$.
Then for each $n$ we have
$$
\begin{array}{c}
\displaystyle\left|\int_{\S(\H_A)}f(\rho)\mu_n(d\rho)-\int_{\S(\H_A)}f(\rho)\mu_0(d\rho)\right|\\\\
\displaystyle\leq
(1-c_n)\left|\int_{\S(\H_A)}f(\rho)\mu_n(d\rho)\right|+
\left|c_n\int_{\S(\H_A)}f(\rho)\mu_n(d\rho)-\int_{\S(\H_A)}f(\rho)\mu_0(d\rho)\right|\\\\\displaystyle\leq
(1-c_n)M_f+\sum_{i=1}^{m(n)}\int_{\A_i^n}|f(\rho_i^n)-f(\rho)|\mu_0(d\rho)+\int_{\S(\H_A)\backslash\C_{n}}|f(\rho)|\mu_0(d\rho)
\\\\\displaystyle\leq (1-c_n)M_f+c_n\omega_f(1/n)+(1-c_n)M_f\leq \omega_f(1/n)+2M_f/n
\end{array}
$$
because $\A^{n}_i\subseteq\B^{n}_i$ for all $i$. Since $\,\omega_f(1/n)+2M_f/n\,$ tends to zero as $\,n\to+\infty$, the sequence $\{\mu_n\}$ weakly converges to the ensemble
$\mu_0$ by the Portmanteau theorem  \cite{Bil+}.

Denote the functional $\mu\to \int_{\S(\H_A)}
F(\rho)\,\mu(d\rho)$ on $\P(\H_A)$ by $\widehat{F}(\mu)$. Since  the function $F(\rho)$ nonnegative and lower semicontinuous
on $\S(\H_A)$, the functional $\widehat{F}(\mu)$
is lower semicontinuous on $\P(\H_A)$ by the Portmanteau theorem  \cite{Bil+}.

By the choice of the states $\rho^{n}_i$ for each  $i$ and $n$ we
have $F(\rho^{n}_i)\le F(\rho)$ for all $\rho$ in
$\A^{n}_i$. Hence,
$$
\begin{array}{c}
\displaystyle\widehat{F}(\mu_n)=c_n^{-1}\sum_{i=1}^{m(n)}\mu_0(\A^{n}_i)\,F(\rho^{n}_i)\le c_n^{-1}\sum_{i=1}^{m(n)}
\int_{\A^{n}_i}F(\rho)\,\mu_0(d\rho)
\\\\
\displaystyle\;\;\, \le
c_n^{-1}\int_{\S(\H_A)}
F(\rho)\,\mu_0(d\rho)=c_n^{-1}\widehat{F}(\mu_0)<+\infty.
\end{array}
$$

Thus, $\limsup_{n\rightarrow+\infty}\widehat{F}(\mu_n)\le\widehat{F}(\mu_0)$ and
hence the lower semicontinuity of $\widehat{F}(\mu)$ implies that
\begin{equation}\label{F-l-r}
\lim_{n\rightarrow+\infty}\widehat{F}(\mu_n)=\widehat{F}(\mu_0)<+\infty.
\end{equation}

If $\Tr H_A\bar{\rho}(\mu_0),\bar{E}_{H_B}^{\rm psv}(\Phi(\mu_0))<+\infty$ then
$\widehat{F}(\mu)$ is the sum of three  lower semicontinuous
functionals $\mu\mapsto\overline{S}(\Phi(\mu))$, $\mu\mapsto\Tr H_A\bar{\rho}(\mu)$ and $\mu\mapsto\bar{E}_{H_B}^{\rm psv}(\Phi(\mu))$ (the lower semicontinuity of the last functional follows from Lemma \ref{PE-LSC}).
So, it follows from Lemma \ref{vsl} in Section 2.1 that the limit relation (\ref{F-l-r})  implies that (\ref{b-lem+}),(\ref{b-lem++}) and (\ref{b-lem+++}) hold.

If $\Tr H_A\bar{\rho}(\mu_0)<+\infty$ and $\bar{E}_{H_B}^{\rm psv}(\Phi(\mu_0))=+\infty$ then
$\widehat{F}(\mu)$ is the sum of two  lower semicontinuous
functionals $\mu\mapsto\overline{S}(\Phi(\mu))$ and $\mu\mapsto\Tr H_A\bar{\rho}(\mu)$.
So, it follows from Lemma \ref{vsl} in Section 2.1 that the limit relation (\ref{F-l-r})  implies that (\ref{b-lem+}) and (\ref{b-lem++}) hold.
The validity of  (\ref{b-lem+++}) in this case is a direct corollary of the lower semicontinuity
of the functional $\mu\mapsto\bar{E}_{H_B}^{\rm psv}(\Phi(\mu))$.

Similar arguments show that  the limit relation (\ref{F-l-r}) implies the validity of (\ref{b-lem+}),(\ref{b-lem++}) and (\ref{b-lem+++}) in the
case when $\Tr H_A\bar{\rho}(\mu_0)=+\infty$ and $\bar{E}_{H_B}^{\rm psv}(\Phi(\mu_0))<+\infty$ and in the case of $\Tr H_A\bar{\rho}(\mu_0)=\Tr \bar{E}_{H_B}^{\rm psv}(\Phi(\mu_0))=+\infty$.
\end{proof}



\section{Semicontinuity bounds for the AOE}

In this section we use the results from Section 2.3 and the approximation Lemma \ref{b-lem} in Section 3 to obtain
semicontinuity bounds for the function $\,(\Phi,\mu)\mapsto\overline{S}(\Phi(\mu))\,$
under the constraints of different forms. \smallskip

\begin{proposition}\label{main-1} \emph{Let
$\,\Phi$ be an  arbitrary channel from $A$ to $B$ and $\mu$ a (generalized) ensemble in $\P(\H_A)$}.\smallskip

\noindent A) \emph{If $\,\rank\Phi(\rho)\leq r_B\in\N\cap[2,+\infty)\,$ for $\mu$-almost all $\rho$ then the inequality\footnote{If $r_B=1$ then $\overline{S}(\Phi(\mu))=0$, so in this case (\ref{main-1+}) holds trivially with the r.h.s. equal to  $0$.}
\begin{equation}\label{main-1+}
\overline{S}(\Phi(\mu))-\overline{S}(\Psi(\nu))\leq \varepsilon\ln (r_B-1)+h_2(\varepsilon)
\end{equation}
holds for any (generalized) ensemble $\nu$ in $\P(\H_A)$ and any channel $\Psi:A\to B$ such that
\begin{equation}\label{e-cond-1}
D_X(\mu,\nu)+\textstyle\frac{1}{2}\|\Phi-\Psi\|_{Y}\leq\varepsilon\leq1-1/r_B,
\end{equation}
where $D_X$ is one of the metrics $D_K$ and $D_*$ defined, respectively, in [(\ref{K-D-d}),(\ref{K-D-c})] and [(\ref{f-metric}),(\ref{f-metric+})],\footnote{These double references indicate the "discrete" and the general definitions of a metric.} $\|\cdot\|_{Y}$ is either the diamond norm $\|\cdot\|_{\diamond}$ defined in (\ref{d-norm}) or its unstabilized version $\|\cdot\|_{1\to1}$ defined in (\ref{d-norm+}).}\smallskip

\emph{If $\,D_X(\mu,\nu)+\textstyle\frac{1}{2}\|\Phi-\Psi\|_{Y}>1-1/r_B\,$ then (\ref{main-1+}) holds with the r.h.s. equal to $\,\ln r_B$.}\medskip

\noindent B) \emph{The following replacements can be done in claim A independently on each other:}
\begin{enumerate}[(1)]
\item \emph{if $\mu=(p_i,\rho_i)$ and $\nu=(q_i,\sigma_i)$ are discrete ordered
ensembles then the metric $D_0$ defined in (\ref{D-0-metric}) can be used in the role of $D_X$  in (\ref{e-cond-1});} \smallskip
\item \emph{if either $\,\Tr H_A\bar{\rho}(\mu)=E_{A}<+\infty$ or $\,\Tr H_A\bar{\rho}(\nu)=E_{A}<+\infty$ then the energy-constrained diamond norm $\|\cdot\|^{H_A}_{\diamond,E_A}$ defined in (\ref{ec-d-norm}) can be used in the role of $\|\cdot\|_{Y}$ in  (\ref{e-cond-1}).}
\end{enumerate}

\noindent C) \emph{Let $\,r_B\in\N\cap[2,+\infty)$ and $\,\varepsilon\in(0,1-1/r_B]$. Then}
\begin{enumerate}[(1)]
\item \emph{there exist discrete  ensembles  $\mu$ and $\nu$ in $\P(\H_A)$ and a channel $\Phi$  from $A$ to $B$ such that
$\,\rank\Phi(\rho)\leq r_B<+\infty\,$ for $\mu$-almost all $\rho$,
\begin{equation*}
\overline{S}(\Phi(\mu))-\overline{S}(\Phi(\nu))=\varepsilon\ln (r_B-1)+h_2(\varepsilon) \quad\textit{and}\quad D_X(\mu,\nu)=\varepsilon,
\end{equation*}
where  $D_X$ is one of the metrics $D_0,D_*,D_K$;} \smallskip
\item \emph{there exist channels $\Phi$ and  $\Psi$  from $A$ to $B$ and a discrete  ensemble  $\mu$ such that
$\,\rank\Phi(\rho)\leq r_B<+\infty\,$ for $\mu$-almost all $\rho$,
\begin{equation*}
\overline{S}(\Phi(\mu))-\overline{S}(\Psi(\mu))=\varepsilon\ln (r_B-1)+h_2(\varepsilon)\quad\textit{and}\quad\textstyle\frac{1}{2}\|\Phi-\Psi\|_{\diamond}\leq\varepsilon.
\end{equation*}}
\end{enumerate}
\end{proposition}\smallskip

\noindent\textbf{Note:} Claim C shows that the semicontinuity bound presented in claim A is optimal.

\begin{proof} Since $\,\|\cdot\|_{1\to1}\leq\|\cdot\|_{\diamond}$, it suffices
to prove all the claims of the proposition assuming that  $\,\|\cdot\|_{1\to1}$ is  used
in the role of $\,\|\cdot\|_{Y}$ (excepting the case in B-2).\smallskip

A) Note first that the condition $\,\rank\Phi(\rho)\leq r_B<+\infty$ for $\mu$-almost all $\rho$ implies that $\overline{S}(\Phi(\mu))\leq\ln r_B<+\infty$.
So, if $\,D_X(\mu,\nu)+\textstyle\frac{1}{2}\|\Phi-\Psi\|_{Y}>1-1/r_B\,$ then the
claim is trivial. Hence, it suffices to consider the case when condition (\ref{e-cond-1}) holds.

Assume that $\mu=\{p_i,\rho_i\}$ and $\nu=\{q_i,\sigma_i\}$ are discrete ensembles in $\P(\H_A)$ and that $\rank\Phi(\rho_i)\leq r_B$ for all $i$ such that $p_i\neq0$.
Take any $\epsilon>0$. Let  $\,(\tilde{p}_i,\tilde{\rho}_i)$ and $(\tilde{q}_i,\tilde{\sigma}_i)$ be ensembles belonging, respectively, to the sets $\mathcal{E}(\{p_i,\rho_i\})$ and $\mathcal{E}(\{q_i,\sigma_i\})$ such that
\begin{equation}\label{e-1}
D_*(\{p_i,\rho_i\}, \{q_i,\sigma_i\})\geq D_0((\tilde{p}_i,\tilde{\rho}_i), (\tilde{q}_i,\tilde{\sigma}_i))-\epsilon
\end{equation}
(see definition (\ref{f-metric}) of $D_*$). Note that $\rank\Phi(\tilde{\rho_i})\leq r_B$ for all $i$ such that $\tilde{p}_i\neq0$.

Consider the q-c states $\hat{\rho}=\sum_{i} \tilde{p}_i\tilde{\rho}_i\otimes |i\rangle\langle i|$
and $\hat{\sigma}=\sum_{i}\tilde{q}_i\tilde{\sigma}_i\otimes |i\rangle\langle i|$
in $\S(\H_{AC})$, where $\{|i\rangle\}$ is a basic in a separable Hilbert space $\H_C$. It follows from (\ref{ce-rep}) that
\begin{equation}\label{one-v}
\overline{S}(\{p_i,\Phi(\rho_i)\})=\overline{S}(\{\tilde{p}_i,\Phi(\tilde{\rho}_i)\})=S(B|C)_{\Phi\otimes \id_C(\hat{\rho})}
\end{equation}
and
\begin{equation}\label{two-v}
\overline{S}(\{q_i,\Psi(\sigma_i)\})=\overline{S}(\{\tilde{q}_i,\Psi(\tilde{\sigma}_i)\})=S(B|C)_{\Psi\otimes \id_C(\hat{\sigma})}.
\end{equation}

By using the definition (\ref{d-norm+}), the monotonicity of the trace norm under action of a channel and (\ref{e-1}) we obtain
\begin{equation}\label{norm+}
\begin{array}{c}
\|\Phi\otimes \id_C(\hat{\rho})-\Psi\otimes \id_C(\hat{\sigma})\|_1\leq\|\Phi\otimes \id_C(\hat{\rho})-\Psi\otimes \id_C(\hat{\rho})\|_1\\\\+\|\Psi\otimes \id_C(\hat{\rho})-\Psi\otimes \id_C(\hat{\sigma})\|_1
\leq \sum_i\tilde{p}_i\|\Phi(\tilde{\rho}_i)-\Psi(\tilde{\rho}_i)\|_{1}+\|\hat{\rho}-\hat{\sigma}\|_1
\\\\
\leq\sum_i\tilde{p}_i\|\Phi-\Psi\|_{1\to1}+\|\hat{\rho}-\hat{\sigma}\|_1\leq \|\Phi-\Psi\|_{1\to1}+2D_*(\mu,\nu)+2\epsilon.
\end{array}
\end{equation}

The norm $\|\Phi-\Psi\|_{1\to1}$ in this estimate can be replaced with $\|\Phi-\Psi\|^{H_A}_{\diamond,E_A}$ if either $\,\Tr H_A\bar{\rho}(\mu)=E_{A}$ or $\,\Tr H_A\bar{\rho}(\nu)=E_{A}$ for some finite $E_A$.\footnote{We cannot replace $\|\Phi-\Psi\|_{1\to1}$ by the "unstabilized" version $\|\Phi-\Psi\|^{H_A}_{1\to1,E_A}$ of  $\|\Phi-\Psi\|^{H_A}_{\diamond,E_A}$ (defined in an obvious way) because we cannot prove the concavity of the function $E\mapsto\|\Phi-\Psi\|^{H_A}_{1\to1,E}$.} Indeed, if
$\,\Tr H_A\bar{\rho}(\mu)=\Tr H_A\hat{\rho}_A=E_{A}$ then it follows from the first inequality in (\ref{norm+}) that
\begin{equation}\label{norm++}
\|\Phi\otimes \id_C(\hat{\rho})-\Psi\otimes \id_C(\hat{\sigma})\|_1\leq \|\Phi-\Psi\|^{H_A}_{\diamond,E_A}+2D_*(\mu,\nu)+2\epsilon
\end{equation}
The similar arguments (with the obvious permutations) show that the above estimate holds also if $\,\Tr H_A\bar{\rho}(\nu)=E_{A}$.

Since $\rank\Phi(\tilde{\rho}_i)\leq r_B$ for all $i$ such that $\tilde{p}_i\neq0$,  by using (\ref{norm+}) and applying  Lemma \ref{Wilde-SCB} in Section 2.3
to estimate the difference between (\ref{one-v}) and (\ref{two-v})
we obtain  (by possibility to take $\epsilon$ arbitrarily small)  that (\ref{main-1+}) holds for any
$\,\varepsilon\geq D_*(\mu,\nu)+\textstyle\frac{1}{2}\|\Phi-\Psi\|_{1\to1}$.

If either $\,\Tr H_A\bar{\rho}(\mu)=E_{A}<+\infty$ or $\,\Tr H_A\bar{\rho}(\nu)=E_{A}<+\infty$ then
by using (\ref{norm++}) instead of (\ref{norm+}) we conclude that (\ref{main-1+}) holds for any
$\,\varepsilon\geq D_*(\mu,\nu)+\textstyle\frac{1}{2}\|\Phi-\Psi\|^{H_A}_{\diamond,E_A}$.\smallskip

Thus, we have proved the inequality (\ref{main-1+}) for discrete ensembles
$\mu$ and $\nu$  provided that condition (\ref{e-cond-1}) holds with  the distance $D_*(\mu,\nu)$ between them (in the role of  $D_X(\mu,\nu)$).
We have proved also that $\|\Phi-\Psi\|_{Y}$ in condition (\ref{e-cond-1}) can be replaced by $\|\Phi-\Psi\|^{H_A}_{\diamond,E_A}$ provided that
either $\,\Tr H_A\bar{\rho}(\mu)=E_{A}<+\infty$ or $\,\Tr H_A\bar{\rho}(\nu)=E_{A}<+\infty$.
Since the r.h.s. of (\ref{main-1+}) is a nondecreasing function of $\varepsilon$, inequalities (\ref{d-ineq}) and (\ref{d-ineq+})  show that
all these claims  are valid for discrete ensembles
$\mu$ and $\nu$ provided that the quantities $D_0(\mu,\nu)$  and $D_K(\mu,\nu)$ are used in condition (\ref{e-cond-1}).
\smallskip

Assume now that $\mu$ and $\nu$ are arbitrary generalized ensembles. We may assume that $\overline{S}(\Phi(\nu))<+\infty$,
since otherwise inequality (\ref{main-1+}) holds trivially. By Lemma \ref{b-lem} in Section 3 there exist sequences
$\{\mu_n\}$ and $\{\nu_n\}$ of discrete ensembles in $\P(\H_A)$ weakly converging to the ensembles
 $\mu$ and $\nu$ correspondingly such that
\begin{equation}\label{nb-lem+1}
\lim_{n\to+\infty}\overline{S}(\Theta(\vartheta_n))=\overline{S}(\Theta(\vartheta))<+\infty,\quad (\Theta,\vartheta)=(\Phi,\mu),(\Psi,\nu),
\end{equation}
\begin{equation}\label{nb-lem++1}
\lim_{n\to+\infty}\Tr H_A\bar{\rho}(\vartheta_n)=\Tr H_A\bar{\rho}(\vartheta)\leq+\infty,\quad \vartheta=\mu,\nu,
\end{equation}
and $\,\rank\Phi(\rho)\leq r_B<+\infty$ for $\mu_n$-almost all $\rho$ for each $n$.

Let $D_X$ be either $D_*$ or $D_K$. The weak convergence of $\{\mu_n\}$ and $\{\nu_n\}$ to  $\mu$ and $\nu$ implies that (see Section 2.2.2)
\begin{equation}\label{D-lim-1}
\lim_{n\to+\infty}D_X(\mu_n,\nu_n)=D_X(\mu,\nu).
\end{equation}

By the above part of the proof we have $\overline{S}(\Phi(\mu_n))-\overline{S}(\Psi(\nu_n))\leq f(\varepsilon_n)$,
where $\varepsilon_n=D_X(\mu_n,\nu_n)+\frac{1}{2}\|\Phi-\Psi\|_{1\to1}\,$ and
\begin{equation*}
f(\varepsilon)=
\left\{\begin{array}{l}
        \varepsilon\ln (r_B-1)+h_2(\varepsilon) \;\;\textrm{if}\;\; \varepsilon\in[0,1-1/r_B)\\
        \ln r_B\;\qquad\qquad\qquad\;\;\,\textrm{if}\;\;\varepsilon>1-1/r_B.
        \end{array}\right.
\end{equation*}
By taking the limit as $n\to+\infty$ and applying the limit relations (\ref{nb-lem+1}) and (\ref{D-lim-1})
we conclude that  (\ref{main-1+}) holds with $\varepsilon=D_X(\mu,\nu)+\frac{1}{2}\|\Phi-\Psi\|_{1\to1}$ (due to the continuity of the function  $f$).
Hence, (\ref{main-1+})  holds for if $\,D_X(\mu,\nu)+\frac{1}{2}\|\Phi-\Psi\|_{1\to1}\leq\varepsilon\leq1-1/r_B$, since $f$ is a nondecreasing function.\smallskip

If $\,\Tr H_A\bar{\rho}(\vartheta)=E_{A}<+\infty$, where $\vartheta$ is  either $\mu$ or $\nu$,
then
the above part of the proof shows that $\overline{S}(\Phi(\mu_n))-\overline{S}(\Psi(\nu_n))\leq f(\varepsilon_n)$, where $\varepsilon_n=D_X(\mu_n,\nu_n)+\frac{1}{2}\|\Phi-\Psi\|^{H_A}_{\diamond,E^n_A}$, $E^n_A=\Tr H_A\bar{\rho}(\vartheta_n)$ (if $E_A^n=+\infty$ then $\,\|\Phi-\Psi\|^{H_A}_{\diamond,E^n_A}=\|\Phi-\Psi\|_{\diamond}$). So, by taking the limit as $n\to+\infty$ and applying the limit relations (\ref{nb-lem+1}), (\ref{nb-lem++1}) and (\ref{D-lim-1})
we conclude that (\ref{main-1+}) holds with $\varepsilon=D_X(\mu,\nu)+\frac{1}{2}\|\Phi-\Psi\|^{H_A}_{\diamond,E_A}$ (due to the continuity of the functions $f$ and $E\mapsto \frac{1}{2}\|\Phi-\Psi\|^{H_A}_{\diamond,E}$). Hence, (\ref{main-1+})  holds if $\,D_X(\mu,\nu)+\frac{1}{2}\|\Phi-\Psi\|^{H_A}_{\diamond,E_A}\leq\varepsilon\leq1-1/r_B$.\medskip

B) Claims B-1 and B-2 are proved within the proof of claim A.
\smallskip

C) To prove claim C-1 assume that $\mu$ and $\nu$ are singleton ensembles
consisting of the states $\,\rho=(1-\varepsilon)|0\rangle\langle 0|+\varepsilon(r_B-1)^{-1}\sum_{i=1}^{r_B-1}|i\rangle\langle i|\,$
and $\sigma=|0\rangle\langle 0|$, where $\{|i\rangle\}_{i\geq0}$ is a basis in $\H_A$. Then by taking $\Phi=\id_A$ we obtain
$$
\overline{S}(\Phi(\mu))-\overline{S}(\Phi(\nu))=S(\rho)-S(\sigma)=\varepsilon\ln (r_B-1)+h_2(\varepsilon)
$$
and $D_X(\mu,\nu)=\frac{1}{2}\|\rho-\sigma\|_1=\varepsilon$, $D_X=D_0,D_*,D_K$.\smallskip

To prove claim C-2 assume that $\,\Phi(\varrho)=(1-\varepsilon)\varrho+\varepsilon(r_B-1)^{-1}\sum_{i=1}^{r_B-1}|i\rangle\langle i|\,$
and $\Psi=\id_A$, where $\{|i\rangle\}_{i\geq0}$ is a basis in $\H_A$. So, if $\mu$ is the singleton ensemble
consisting of the state $\,\sigma=|0\rangle\langle 0|\,$ then
$$
\overline{S}(\Phi(\mu))-\overline{S}(\Psi(\mu))=S(\rho)-S(\sigma)=\varepsilon\ln (r_B-1)+h_2(\varepsilon),
$$
where $\rho$ is  the state introduced in the proof of claim C-1. Since it is easy to see that $\|\Phi\otimes\id_R(\omega)-\Psi\otimes\id_R(\omega)\|_1\leq 2\varepsilon\,$
for any state $\omega\in\S(\H_{AR})$ (where $R$ is any quantum system), we have $\|\Phi-\Psi\|_{\diamond}\leq 2\varepsilon$.
\end{proof}

\begin{corollary}\label{main-1-c} \emph{Let
$\Phi$ be a channel from $A$ to $B$ with finite Choi rank $\,\mathrm{Ch}(\Phi)\geq2$ and $\mu$ a (generalized) ensemble of pure states in $\S(\H_A)$}.\footnote{If $\mathrm{Ch}(\Phi)=1$ then $\,\overline{S}(\Phi(\mu))=0$ for any ensemble $\mu$  of pure states in $\S(\H_A)$.}\emph{The inequality
\begin{equation*}
\overline{S}(\Phi(\mu))-\overline{S}(\Psi(\nu))\leq \varepsilon\ln (\mathrm{Ch}(\Phi)-1)+h_2(\varepsilon)
\end{equation*}
holds for any (generalized) ensemble $\nu$ in $\P(\H_A)$ and any channel $\Psi:A\to B$ such that
\begin{equation}\label{e-cond-1-c}
D_X(\mu,\nu)+\textstyle\frac{1}{2}\|\Phi-\Psi\|_{Y}\leq\varepsilon\leq1-1/\mathrm{Ch}(\Phi),
\end{equation}
where $D_X$ is one of the metrics $D_K$ and $D_*$ defined, respectively, in [(\ref{K-D-d}),(\ref{K-D-c})] and [(\ref{f-metric}),(\ref{f-metric+})], $\|\cdot\|_{Y}$ is either the diamond norm $\|\cdot\|_{\diamond}$ defined in (\ref{d-norm}) or its unstabilized version $\|\cdot\|_{1\to1}$ defined in (\ref{d-norm+}).}\smallskip\pagebreak

\emph{The following replacements can be done independently on each other:}
\begin{enumerate}[(1)]
\item \emph{if $\mu=(p_i,\rho_i)$ and $\nu=(q_i,\sigma_i)$ are discrete ordered
ensembles then the metric $D_0$ defined in (\ref{D-0-metric}) can be used in the role of $D_X$  in (\ref{e-cond-1-c});} \smallskip
\item \emph{if either $\,\Tr H_A\bar{\rho}(\mu)=E_{A}<+\infty$ or $\,\Tr H_A\bar{\rho}(\nu)=E_{A}<+\infty$ then the energy-constrained diamond norm $\|\cdot\|^{H_A}_{\diamond,E_A}$ defined in (\ref{ec-d-norm}) can be used in the role of $\|\cdot\|_{Y}$ in  (\ref{e-cond-1-c}).}
\end{enumerate}
\end{corollary}\smallskip

To derive the claims of Corollary \ref{main-1-c} from Proposition \ref{main-1}
it suffices to note that  $\,\rank\Phi(\rho)\leq\mathrm{Ch}(\Phi)$ for any pure state $\rho\in\S(\H_A)$.\medskip


The following proposition is an energy  constrained version of Proposition \ref{main-1}. We use the notion of the \emph{average output passive energy} of a channel $\Phi:A\to B$ at a (generalized)
ensemble $\mu$ in $\P(\H_A)$ w.r.t. Hamiltonian $H_B$ of $B$ defined as
\begin{equation}\label{OPE-def}
 \bar{E}_{H_B}^{\rm psv}(\Phi(\mu))=\int_{\S(\H_A)}E_{H_B}^{\rm psv}(\Phi(\rho))\mu(d\rho)=\int_{\S(\H_B)}E_{H_B}^{\rm psv}(\sigma)[\Phi(\mu)](d\sigma),
\end{equation}
i.e. as the average passive energy of the ensemble $\Phi(\mu)$ (see definition (\ref{p-a-e-def}) in Section 2.2.3).
It is clear that
\begin{equation}\label{OPE-def+}
 \bar{E}_{H_B}^{\rm psv}(\Phi(\mu))\leq\int_{\S(\H_A)}\Tr H_B\Phi(\rho)\mu(\rho)=\Tr H_B\Phi(\bar{\rho}(\mu))
\end{equation}
and that the difference between $\Tr H_B\Phi(\bar{\rho}(\mu))$ and $\bar{E}_{H_B}^{\rm psv}(\Phi(\mu))$ (that can be called average output ergotropy)
may be large: it is easy to find a channel $\Phi$, an input ensemble $\mu$ and output Hamiltonian $H_B$ s.t. $\bar{E}_{H_B}^{\rm psv}(\Phi(\mu))=0$ and $\Tr H_B\Phi(\bar{\rho}(\mu))=+\infty$.\smallskip

In Examples \ref{C-states}, \ref{main-3-e+} and \ref{C-states+} below, the quantity $\bar{E}_{H}^{\rm psv}(\Phi(\mu))$ is  estimated for real quantum channels $\Phi$ and ensembles $\mu$.
\smallskip

\begin{proposition}\label{main-2} \emph{Let
$\,\Phi$ be an  arbitrary channel from $A$ to $B$ and $\mu$ a (generalized) ensemble in $\P(\H_A)$. Let $H_B$ be a positive operator on the space $\H_B$ satisfying conditions (\ref{H-cond}) and (\ref{star}). Let $F_{H_B}$ and $g$ be the functions defined, respectively, in (\ref{F-def}) and (\ref{g-def}).}\smallskip

\noindent A) \emph{If $\,\bar{E}_{H_B}^{\rm psv}(\Phi(\mu))=E_B<+\infty\,$ then
\begin{equation}\label{main-2+}
\overline{S}(\Phi(\mu))-\overline{S}(\Psi(\nu))\leq \varepsilon F_{H_B}(E_B/\varepsilon)+g(\varepsilon)
\end{equation}
for any (generalized) ensemble $\nu$ in $\P(\H_A)$ and any channel $\Psi:A\to B$ such that
\begin{equation}\label{e-cond-2}
D_X(\mu,\nu)+\textstyle\frac{1}{2}\|\Phi-\Psi\|_{Y}\leq\varepsilon,
\end{equation}
where $D_X$ is one of the metrics $D_K$ and $D_*$ defined, respectively, in [(\ref{K-D-d}),(\ref{K-D-c})] and [(\ref{f-metric}),(\ref{f-metric+})],\footnote{These double references indicate the "discrete" and the general definitions of a metric.} $\|\cdot\|_{Y}$ is either the diamond norm $\|\cdot\|_{\diamond}$ defined in (\ref{d-norm}) or its unstabilized version $\|\cdot\|_{1\to1}$ defined in (\ref{d-norm+}).}\smallskip\pagebreak

\noindent B) \emph{The following replacements can be done independently on each other:}
\begin{enumerate}[(1)]
\item \emph{if $\mu=(p_i,\rho_i)$ and $\nu=(q_i,\sigma_i)$ are discrete ordered
ensembles then the metric $D_0$ defined in (\ref{D-0-metric}) can be used in the role of $D_X$ in  (\ref{e-cond-2});} \smallskip
\item  \emph{$E_B=\bar{E}_{H_B}^{\rm psv}(\Phi(\mu))$ in (\ref{main-2+}) can be replaced with $\,E_B=\Tr H_B\Phi(\bar{\rho}(\mu))$;} \smallskip
\item \emph{if either $\,\Tr H_A\bar{\rho}(\mu)=E_{A}<+\infty\,$ or $\,\Tr H_A\bar{\rho}(\nu)=E_{A}<+\infty$, where $H_A$ is a positive operator on $\H_A$, then the energy-constrained diamond norm $\|\cdot\|^{H_A}_{\diamond,E_A}$ defined in (\ref{ec-d-norm}) can be used in the role of $\|\cdot\|_{Y}$ in  (\ref{e-cond-2}).}
\end{enumerate}
\smallskip

\emph{If $\mu=(p_i,\rho_i)$ and $\nu=(q_i,\sigma_i)$ are discrete ordered
ensembles and the metric  $D_0$ is used in (\ref{e-cond-2}) in the role of  $D_X$ then $E_B=\bar{E}_{H_B}^{\rm psv}(\Phi(\mu))$ in (\ref{main-2+}) can be replaced by $\,E_B-E_B(\varepsilon)$, where
\begin{equation*}
E_B(\varepsilon)\doteq \sum_{i}\bar{E}_{H_B}^{\rm psv}[p_i\Phi(\rho_i)-\varepsilon I_{B}]_+,
\end{equation*}
$\,[p_i\Phi(\rho_i)-\varepsilon  I_{B}]_+$ is the positive part of the Hermitian operator $\,p_i\Phi(\rho_i)-\varepsilon  I_{B}$.}\medskip

\noindent C) \emph{Let $E>0$ and $\varepsilon\in(0,1]$. Then}
\begin{enumerate}[(1)]
\item \emph{there exist discrete ensembles  $\mu$ and $\nu$ in $\P(\H_A)$ and a channel $\Phi$  from $A$ to $B$ such that
$$
\overline{S}(\Phi(\mu))-\overline{S}(\Phi(\nu))> \varepsilon F_{H_B}(E/\varepsilon),\quad \bar{E}_{H_B}^{\rm psv}(\Phi(\mu))=E\quad \textit{and} \quad D_X(\mu,\nu)=\varepsilon,
$$
where  $D_X$ is one of the metrics $D_0,D_*,D_K$;}

\item \emph{there exist channels $\,\Phi$ and $\,\Psi$  from $A$ to $B$ and a discrete  ensemble  $\mu$ in $\P(\H_A)$ such that}
$$
\overline{S}(\Phi(\mu))-\overline{S}(\Psi(\mu))> \varepsilon F_H(E/\varepsilon),\quad \bar{E}_{H_B}^{\rm psv}(\Phi(\mu))=E\quad \textit{and} \quad \textstyle\frac{1}{2}\|\Phi-\Psi\|_{\diamond}\leq\varepsilon.
$$
\end{enumerate}
\end{proposition}

\noindent\textbf{Note A:} The r.h.s. of (\ref{main-2+}) is a nondecreasing
function of $\varepsilon$ which tends to zero as $\varepsilon\to0$. This follows from inequality (\ref{W-L}) and the equivalence of (\ref{H-cond}) and (\ref{H-cond-a}).
\smallskip

\noindent\textbf{Note B:} Claim C shows that the semicontinuity bound presented in claim A is close-to-optimal  for large $E$.

\begin{proof} Since $\,\|\cdot\|_{1\to1}\leq\|\cdot\|_{\diamond}$, it suffices
to prove all the claims of the proposition assuming that  $\,\|\cdot\|_{1\to1}$ is used
in the role of $\,\|\cdot\|_{Y}$ (excepting the case in B-3).\smallskip

A) The condition $\,\bar{E}_{H_B}^{\rm psv}(\Phi(\mu))=E_B<+\infty$  implies that $\overline{S}(\Phi(\mu))\leq F_{H_B}(E_B)<+\infty$ by inequality (\ref{P-B+}).\smallskip

Assume that $\mu=\{p_i,\rho_i\}$ and $\nu=\{q_i,\sigma_i\}$ are discrete ensembles in $\P(\H_A)$ and that
$\,\bar{E}_{H_B}^{\rm psv}(\Phi(\mu))=E_B<+\infty$.  Take any
$\epsilon>0$. Let  $\,(\tilde{p}_i,\tilde{\rho}_i)$ and $(\tilde{q}_i,\tilde{\sigma}_i)$ be ordered ensembles belonging, respectively, to the sets $\mathcal{E}(\{p_i,\rho_i\})$ and $\mathcal{E}(\{q_i,\sigma_i\})$ such that
\begin{equation*}
D_*(\{p_i,\rho_i\}, \{q_i,\sigma_i\})\geq D_0((\tilde{p}_i,\tilde{\rho}_i), (\tilde{q}_i,\tilde{\sigma}_i))-\epsilon
\end{equation*}
(see definition (\ref{f-metric}) of $D_*$). Then we repeat the arguments from the proof of  Proposition \ref{main-1} up to (\ref{norm++}).

Since  $\,\bar{E}_{H_B}^{\rm psv}(\{\tilde{p}_i,\Phi(\tilde{\rho}_i)\})=E_B<+\infty$, by using (\ref{norm+}) and applying  Lemma \ref{qce-qc-1}  in Section 2.3
to estimate the difference between (\ref{one-v}) and (\ref{two-v}) we obtain  (because of the possibility to take $\epsilon$ arbitrarily small and due to the continuity of $F_{H_B}$) that (\ref{main-2+}) holds for any $\,\varepsilon\geq D_*(\mu,\nu)+\textstyle \frac{1}{2}\|\Phi-\Psi\|_{1\to1}$.

If either $\,\Tr H_A\bar{\rho}(\mu)=E_{A}<+\infty$ or $\,\Tr H_A\bar{\rho}(\nu)=E_{A}<+\infty$ then
by using (\ref{norm++}) instead of (\ref{norm+}) we conclude that (\ref{main-2+}) holds for any $\,\varepsilon\geq D_*(\mu,\nu)+\textstyle\frac{1}{2}\|\Phi-\Psi\|^{H_A}_{\diamond,E_A}$.\smallskip

Thus, we have proved the inequality (\ref{main-2+}) for discrete ensembles
$\mu$ and $\nu$  provided that condition (\ref{e-cond-2}) holds with  the distance $D_*(\mu,\nu)$ between them (in the role of $D_X(\mu,\nu)$).
We have proved also that $\|\Phi-\Psi\|_{Y}$ in condition (\ref{e-cond-2}) can be replaced by $\|\Phi-\Psi\|^{H_A}_{\diamond,E_A}$ provided that
either $\,\Tr H_A\bar{\rho}(\mu)=E_{A}<+\infty$ or $\,\Tr H_A\bar{\rho}(\nu)=E_{A}<+\infty$.
Since the r.h.s. of (\ref{main-2+}) is a nondecreasing function of $\varepsilon$, inequalities (\ref{d-ineq}) and (\ref{d-ineq+}) show that
all these claims  are valid for discrete ensembles
$\mu$ and $\nu$ provided that the quantities  $D_0(\mu,\nu)$ and $D_K(\mu,\nu)$ are used in condition (\ref{e-cond-2}) in the role of $D_X(\mu,\nu)$.
\smallskip

Assume now that $\mu$ and $\nu$ are arbitrary generalized ensembles. We may assume that $\overline{S}(\Psi(\nu))<+\infty$,
since otherwise inequality (\ref{main-2+}) holds trivially. By Lemma \ref{b-lem} in Section 3 there exist sequences
$\{\mu_n\}$ and $\{\nu_n\}$ of discrete ensembles in $\P(\H_A)$ weakly converging to the ensembles
 $\mu$ and $\nu$ correspondingly such that
\begin{equation}\label{nb-lem+}
\lim_{n\to+\infty}\overline{S}(\Theta(\vartheta_n))=\overline{S}(\Theta(\vartheta))<+\infty,\quad (\Theta,\vartheta)=(\Phi,\mu),(\Psi,\nu),
\end{equation}
\begin{equation}\label{nb-lem++}
\lim_{n\to+\infty}\Tr H_A\bar{\rho}(\vartheta_n)=\Tr H_A\bar{\rho}(\vartheta)\leq+\infty,\quad \vartheta=\mu,\nu,
\end{equation}
and
\begin{equation}\label{nb-lem+++}
\lim_{n\to+\infty}\bar{E}_{H_B}^{\rm psv}(\Phi(\mu_n))=\bar{E}_{H_B}^{\rm psv}(\Phi(\mu))<+\infty.
\end{equation}

Let $D_X$ be either $D_*$ or $D_K$. The weak convergence of $\{\mu_n\}$ and $\{\nu_n\}$ to  $\mu$ and $\nu$ implies that (see Section 2.2.2)
\begin{equation}\label{D-lim}
\lim_{n\to+\infty}D_X(\mu_n,\nu_n)=D_X(\mu,\nu).
\end{equation}

By the above part of the proof we have
\begin{equation}\label{main-2+c}
\overline{S}(\Phi(\mu_n))-\overline{S}(\Psi(\nu_n))\leq \varepsilon_n F_H(E_B^n/\varepsilon_n)+g(\varepsilon_n),
\end{equation}
where $\varepsilon_n=D_X(\mu_n,\nu_n)+\frac{1}{2}\|\Phi-\Psi\|_{1\to1}$ and $E_B^n=\bar{E}_{H_B}^{\rm psv}(\Phi(\mu_n))$.
By taking the limit as $n\to+\infty$ and applying the limit relations (\ref{nb-lem+}), (\ref{nb-lem+++}) and (\ref{D-lim})
we conclude that  (\ref{main-2+}) holds with $\varepsilon=D_X(\mu,\nu)+\frac{1}{2}\|\Phi-\Psi\|_{1\to1}$ (due to the continuity of the functions $F_{H_B}$ and $g$).
Hence, (\ref{main-2+})  holds for all $\,\varepsilon\geq D_X(\mu,\nu)+\frac{1}{2}\|\Phi-\Psi\|_{1\to1}$, since the r.h.s. of (\ref{main-2+}) is a nondecreasing function of $\varepsilon$ (this follows from (\ref{W-L})).

If $\,\Tr H_A\bar{\rho}(\vartheta)=E_{A}<+\infty$, where $\vartheta$ is  either $\mu$ or $\nu$,
then
the above part of the proof shows that (\ref{main-2+c}) holds with $\varepsilon_n=D_X(\mu,\nu)+\frac{1}{2}\|\Phi-\Psi\|^{H_A}_{\diamond,E^n_A}$,
where $E^n_A=\Tr H_A\bar{\rho}(\vartheta_n)$ (if $E_A^n=+\infty$ then $\,\|\Phi-\Psi\|^{H_A}_{\diamond,E^n_A}=\|\Phi-\Psi\|_{\diamond}$). So, by taking the limit as $n\to+\infty$ and applying the limit relations (\ref{nb-lem+}),(\ref{nb-lem++}), (\ref{nb-lem+++}) and (\ref{D-lim})
we conclude that (\ref{main-2+}) holds with $\varepsilon=D_X(\mu,\nu)+\frac{1}{2}\|\Phi-\Psi\|^{H_A}_{\diamond,E_A}$ (due to the continuity of the functions $F_{H_B}$, $g$ and $E\mapsto \|\Phi-\Psi\|^{H_A}_{\diamond,E}$). Hence, (\ref{main-2+})  holds for all $\,\varepsilon\geq D_X(\mu,\nu)+\frac{1}{2}\|\Phi-\Psi\|^{H_A}_{\diamond,E_A}$.
\smallskip

B) Claim B-1 and B-3 are proved within the proof of claim A.
Claim B-2 follows from the inequality (\ref{OPE-def+}), since the r.h.s. of (\ref{main-2+}) is a nondecreasing function of $\varepsilon$.
\smallskip

It remains to prove the possibility to replace $E_B$ in (\ref{main-2+}) by $E_B-E_B(\varepsilon)$ provided that
$\mu=(p_i,\rho_i)$ and $\nu=(q_i,\sigma_i)$ are ordered  discrete ensembles and  condition (\ref{e-cond-2}) holds with $D_X(\mu,\nu)=D_0(\mu,\nu)$.
This can be done by repeating all the above steps for the ensembles  $\,(p_i,\rho_i)$ and $(q_i,\sigma_i)$
instead of $\,(\tilde{p}_i,\tilde{\rho}_i)$ and $(\tilde{q}_i,\tilde{\sigma}_i)$ and using  Lemma \ref{qce-qc-1}  in Section 2.3
at the last step.\smallskip

C) To prove claim C-1 assume that $B=A$ and that $\mu$ and $\nu$ are singleton ensembles consisting, respectively,  of the states $\rho=\varepsilon \gamma_{H_B}(E/\varepsilon)+(1-\varepsilon)|\tau_0\rangle\langle\tau_0|$
and $\sigma=|\tau_0\rangle\langle\tau_0|$, where $\gamma_{H_B}(E/\varepsilon)$ is the Gibbs state (\ref{Gibbs}) corresponding to the energy  $E/\varepsilon$ and
$\tau_0$ is a unit vector from the kernel of $H_B$.  Then by taking $\Phi=\id_A$ we obtain
$$
\overline{S}(\Phi(\mu))-\overline{S}(\Phi(\nu))=S(\rho)-S(\sigma)>\varepsilon F_{H_B}(E/\varepsilon),\quad \bar{E}_{H_B}^{\rm psv}(\Phi(\mu))=\Tr H_B\rho=E
$$
and $D_X(\mu,\nu)=\frac{1}{2}\|\rho-\sigma\|_1\leq\varepsilon$, $D_X=D_0,D_*,D_K$.

To prove claim C-2 assume that $B=A$, $\,\Phi(\varrho)=(1-\varepsilon)\varrho+\varepsilon\gamma_{H_B}(E/\varepsilon)$, $\varrho\in\S(\H_A)$,
and $\Psi=\id_A$, where $\gamma_{H_B}(E/\varepsilon)$ is the Gibbs state (\ref{Gibbs}) corresponding to the energy  $E/\varepsilon$. So, if $\mu$ is the singleton ensemble
consisting of the state $\,\sigma=|\tau_0\rangle\langle\tau_0|$, where $\tau_0$ is a unit vector from the kernel of $H_B$, then
$$
\overline{S}(\Phi(\mu))-\overline{S}(\Psi(\mu))=S(\rho)-S(\sigma)>\varepsilon F_{H_B}(E/\varepsilon),\quad \bar{E}_{H_B}^{\rm psv}(\Phi(\mu))=\Tr H_B\rho=E,
$$
where $\rho$ is  the state introduced in the proof of claim C-1. Since it is easy to see that $\|\Phi\otimes\id_R(\omega)-\Psi\otimes\id_R(\omega)\|_1\leq 2\varepsilon\,$
for any state $\omega\in\S(\H_{AR})$ (where $R$ is any quantum system), we have $\|\Phi-\Psi\|_{\diamond}\leq 2\varepsilon$.
\end{proof}

\begin{example}\label{C-states} Assume that $A$ is an one-mode quantum oscillator and $\Phi_{k,N_c}:A\to A$ is
the amplification Gaussian channel with the amplification coefficient $k>1$ and the power of environment noise $\,N_c\geq0$.
This channel can be defined via the action of the dual channel $\Phi^*_{k,N_c}$ on the Weyl operators as follows
\begin{equation*}
\Phi^*_{k,N_c}(W(z))=W(kz)\exp\left(-\frac{1}{2}\left(N_c+\frac{k^2-1}{2}\right)(x^2+y^2)\right),\quad z=(x,y)\in Z,
\end{equation*}
where $Z$ is 2-D symplectic space \cite[Ch.12]{H-SCI}. Assume that  $H_A=\hat{N}\doteq a^{\dagger}a$ is the number operator on $\H_A$.

Let $\Upsilon(\zeta)=\vert\zeta\rangle\langle\zeta\vert$ be the coherent state of $A$ corresponding to a complex number $\zeta$.
Consider the ensemble $\mu=\Upsilon(\pi)$, where $\pi$ is the probability measure on $\mathbb{C}$ with the density $p_{\mu}(\zeta)=\frac{1}{\pi N}\exp(-|\zeta|^2/N)$ (such ensemble is typically denoted by $\{\pi,\vert \zeta\rangle\langle \zeta\vert\}_{\zeta\in\mathbb{C}}$, see Section 6.2). Then $\bar{\rho}(\mu)=\gamma_{\hat{N}}(N)$ is
the Gibbs state corresponding to the
"energy" (the number of quanta) $N$, while $\Phi_{k,N_c}(\bar{\rho}(\mu))=\gamma_{\hat{N}}(k^2N+k^2-1+N_c)$ \cite[Ch.12]{H-SCI}.\smallskip

Since for any $\zeta\in\mathbb{C}$ the state $\Phi_{k,N_c}(\vert \zeta\rangle\langle \zeta\vert)$
is unitary equivalent to the Gibbs state $\Phi_{k,N_c}(\vert 0\rangle\langle 0\vert)=\gamma_{\hat{N}}(k^2-1+N_c)$, we have
$$
\bar{E}_{\hat{N}}^{\rm psv}(\Phi_{k,N_c}(\mu))=\Tr\hat{N}\gamma_{\hat{N}}(k^2-1+N_c)=k^2-1+N_c.
$$
Note that $\,\bar{E}_{\hat{N}}^{\rm psv}(\Phi_{k,N_c}(\mu))\ll \Tr\hat{N}\Phi_{k,N_c}(\bar{\rho}(\mu))=k^2N+k^2-1+N_c\,$ for large $N$.\smallskip

Thus, as $F_{\hat{N}}(E)=g(E)$, Proposition  \ref{main-2} implies that
\begin{equation}\label{main-2-e+}
\overline{S}(\Phi_{k,N_c}(\mu))-\overline{S}(\Psi(\nu))\leq \varepsilon g((k^2-1+N_c)/\varepsilon)+g(\varepsilon)
\end{equation}
for any  ensemble $\nu$ of states in $\S(\H_{A})$ and any channel $\Psi:A\to A$ such that
\begin{equation}\label{e-cond-2-e}
D_X(\mu,\nu)+\textstyle\frac{1}{2}\|\Phi_{k,N_c}-\Psi\|_{Y}\leq\varepsilon,
\end{equation}
where $D_X$ is one of the metrics $D_K$ and $D_*$ and $\|\cdot\|_{Y}$ is either the norm $\|\cdot\|_{1\to1}$ defined in (\ref{d-norm+}) or  the energy-constrained  diamond norm $\|\cdot\|^{\hat{N}}_{\diamond,N}$ defined in (\ref{ec-d-norm}). The norm $\|\cdot\|_{1\to1}$  is easier to calculate and may be less than
the norm $\|\cdot\|^{\hat{N}}_{\diamond,N}$ for large $N$.
\smallskip

Since $\,\overline{S}(\Phi_{k,N_c}(\mu))=S(\gamma_{\hat{N}}(k^2-1+N_c))=g(k^2-1+N_c)$, it follows from (\ref{main-2-e+}) that
\begin{equation}\label{main-2-e++}
\overline{S}(\Psi(\nu))\geq g(k^2-1+N_c)-\varepsilon g((k^2-1+N_c)/\varepsilon)-g(\varepsilon)
\end{equation}
for any  ensemble $\nu$ of states in $\S(\H_{A})$ and any channel $\Psi:A\to A$ satisfying (\ref{e-cond-2-e}), i.e. we obtain
a lower bound on $\overline{S}(\Psi(\nu))$  valid for all $(\Psi,\nu)$ in the $\varepsilon$-vicinity of $(\Phi_{k,N_c},\mu)$.
Note that the r.h.s. of  (\ref{main-2-e++}) tends to $\,\overline{S}(\Phi_{k,N_c}(\mu))=g(k^2-1+N_c)\,$ as $\,\varepsilon\to0$, i.e.
the local lower bound (\ref{main-2-e++}) is faithful.
\end{example}

\section{Applications}



\subsection{Semicontinuity and continuity bounds for the output Holevo information of a quantum channel}

One of the reasons showing the necessity  to investigate the properties of the AOE of a quantum channel
consists in the fact that the AOE is an integral part of the well known expression
\begin{equation}\label{H-Q++}
\chi(\Phi(\mu))=S(\Phi(\bar{\rho}(\mu))-\overline{S}(\Phi(\mu))
\end{equation}
for the output
Holevo information of a channel $\Phi$ at a (generalised) ensemble $\mu$ of input states, which is valid
provided that $S(\Phi(\bar{\rho}(\mu))<+\infty$, where $\bar{\rho}(\mu)$ is the average state of $\mu$ \cite{H-SCI,H-Sh-2}.\footnote{A general definition of $\chi(\Phi(\mu))$ is given by the expression (\ref{H-Q}).}

By using expression (\ref{H-Q++}) and the semicontinuity bounds for the entropy and for the AOE obtained, respectively, in \cite{LCB} and in Section 4
one can get semicontinuity bounds for the function  $(\Phi,\mu)\to\chi(\Phi(\mu))$  under the rank/energy constraints imposed on the output
of a channel.\footnote{Strictly speaking, the inequalities obtained in Proposition \ref{main-3} cannot be called semicontinuity bounds in the sense described in the Introduction,
because these inequalities hold under the rank/energy constraint imposed on both pairs $(\Phi,\mu)$ and $(\Psi,\nu)$.}\smallskip

\begin{proposition}\label{main-3} \emph{Let $\Phi$ and $\Psi $ be quantum channels from $A$ to $B$ and  $\mu$ and $\nu$  (generalized) ensembles of states in $\S(\H_A)$ such that
\begin{equation}\label{e-cond-3}
D_X(\mu,\nu)+\textstyle\frac{1}{2}\|\Phi-\Psi\|_{Y}\leq\varepsilon\leq1,
\end{equation}
where $D_X$ is one of the metrics $D_K$ and $D_*$ defined, respectively, in [(\ref{K-D-d}),(\ref{K-D-c})] and [(\ref{f-metric}),(\ref{f-metric+})],\footnote{These double references indicate the "discrete" and the general definitions of the metrics.}
$\|\cdot\|_{Y}$ is either the diamond norm $\|\cdot\|_{\diamond}$ defined in (\ref{d-norm}) or its unstabilized version $\|\cdot\|_{1\to1}$ defined in (\ref{d-norm+}).}\smallskip

\emph{Let $H^{\mu}_B$ and $H^{\nu}_B$ be positive operators on $\H_B$ satisfying conditions (\ref{H-cond}) and (\ref{star}).}\smallskip

\emph{Assume that
$$
 \textit{either }\rank\Phi(\bar{\rho}(\mu))\leq r_\mu<+\infty\;\, (\textit{case } A_1)\;\textit{ or }\; \Tr H^{\mu}_B\Phi(\bar{\rho}(\mu))\leq E_\mu<+\infty\;\, (\textit{case }A_2)
$$
and
\begin{equation}\label{e-cond-3A}
\textit{either }\;\; \rank\Psi(\rho)\leq r_\nu<+\infty\;\;\textit{ for }\nu\textit{-almost all }\rho\;\, (\textit{case } B_1)
\end{equation}
\begin{equation}\label{e-cond-3B}
\textit{or }\;\; \bar{E}_{H^{\nu}_B}^{\rm psv}(\Psi(\nu))\leq E_\nu<+\infty\;\, (\textit{case }B_2),
\end{equation}
where $\bar{E}_{H^{\nu}_B}^{\rm psv}(\Psi(\nu))$ is the passive energy
of the ensemble $\Psi(\nu)$ w.r.t. $H^{\nu}_B$ defined in (\ref{p-a-e-def}).}\smallskip

\emph{If the cases $A_i$ and $B_j$ ($i,j=1,2$) are valid then
\begin{equation*}
\chi(\Phi(\mu))-\chi(\Psi(\nu))\leq A_i(\varepsilon)+B_j(\varepsilon),
\end{equation*}
where
\begin{itemize}
  \item $A_1(\varepsilon)=\varepsilon\ln(r_\mu-1)+h_2(\varepsilon)\,$ if $\;\varepsilon<1-1/r_{\mu}$ and $\,A_1(\varepsilon)=\ln r_{\mu}$ otherwise;
  \item $A_2(\varepsilon)=\varepsilon F_{H^{\mu}_B}(E_\mu/\varepsilon)+g(\varepsilon)$;
  \item $B_1(\varepsilon)=\varepsilon\ln(r_\nu-1)+h_2(\varepsilon)\,$ if $\;\varepsilon<1-1/r_{\nu}$ and $\,B_1(\varepsilon)=\ln r_{\nu}$ otherwise;
  \item $B_2(\varepsilon)=\varepsilon F_{H^{\nu}_B}(E_\nu/\varepsilon)+g(\varepsilon)$
\end{itemize}
($\shs F_{H^{x}_B}$, $x=\mu,\nu$, and $g$ are the functions defined, respectively, in (\ref{F-def}) and (\ref{g-def})).}\medskip

\emph{The following replacements can be done  independently on each other:}
\begin{enumerate}[(1)]
\item \emph{if $\mu=(p_i,\rho_i)$ and $\nu=(q_i,\sigma_i)$ are discrete ordered
ensembles then the metric $D_0$ defined in (\ref{D-0-metric}) can be used in the role of $D_X$  in (\ref{e-cond-3});} \smallskip
\item \emph{if either $\,\Tr H_A\bar{\rho}(\mu)=E_{A}<+\infty$ or $\,\Tr H_A\bar{\rho}(\nu)=E_{A}<+\infty$, where  $H_A$ is a positive operator on $\H_A$, then the energy-constrained diamond norm $\|\cdot\|^{H_A}_{\diamond,E_A}$ defined in (\ref{ec-d-norm}) can be used in the role of $\|\cdot\|_{Y}$ in (\ref{e-cond-3}).}
\item \emph{the condition $\,\bar{E}_{H^{\nu}_B}^{\rm psv}(\Psi(\nu))\leq E_\nu\,$ in (\ref{e-cond-3B}) can be replaced by one of the conditions
$\bar{E}_{H^{\nu}_B}^{\rm psv}(\Psi(\bar{\rho}(\nu)))\leq E_\nu
\,$ and $\,\Tr H^{\nu}_B\Psi(\bar{\rho}(\nu))\leq E_\nu$.}
\end{enumerate}
\end{proposition}\smallskip

\noindent The functions $A_2(\varepsilon)$ and $B_2(\varepsilon)$ are nondecreasing and tend to zero as $\varepsilon\to0$. This follows from inequality (\ref{W-L}) and the equivalence of  (\ref{H-cond}) and  (\ref{H-cond-a}).

\begin{proof} The conditions of the proposition
implies the finiteness of $S(\Phi(\bar{\rho}(\mu))$ and $S(\Psi(\bar{\rho}(\nu))$ in all the cases $A_i$ and $B_j$, $i,j=1,2$. So, it follows from representation (\ref{H-Q++})  that
\begin{equation}\label{H-Q+++}
\chi(\Phi(\mu))-\chi(\Psi(\nu))=[S(\Phi(\bar{\rho}(\mu))-S(\Psi(\bar{\rho}(\nu))]+[\overline{S}(\Psi(\nu))-\overline{S}(\Phi(\mu))].
\end{equation}
Since $\,\|\cdot\|_{1\to1}\leq\|\cdot\|_{\diamond}$, it suffices
to prove all the claims of the proposition assuming that  $\,\|\cdot\|_{1\to1}$ is used
in the role of $\,\|\cdot\|_{Y}$ (excepting the case  when either $\,\Tr H_A\bar{\rho}(\mu)=E_{A}$ or $\,\Tr H_A\bar{\rho}(\nu)=E_{A}$).
\smallskip

The definition (\ref{d-norm+}) of the  norm $\,\|\cdot\|_{1\to1}$  and Lemma \ref{a-s-l} in Section 2.2.2 imply that
\begin{equation}\label{t-e}
\begin{array}{c}
\|\Phi(\bar{\rho}(\mu))-\Psi(\bar{\rho}(\nu))\|_1\leq \|\Phi(\bar{\rho}(\mu))-\Psi(\bar{\rho}(\mu))\|_1+\|\Psi(\bar{\rho}(\mu))-\Psi(\bar{\rho}(\nu))\|_1\\\\
\leq \|\Phi(\bar{\rho}(\mu))-\Psi(\bar{\rho}(\mu))\|_1+\|\bar{\rho}(\mu)-\bar{\rho}(\nu)\|_1\leq\|\Phi-\Psi\|_{1\to1}+2D_X(\mu,\nu).
\end{array}
\end{equation}
So, to prove the main claim it suffices to apply Propositions 2 and 1 in  \cite{LCB} in cases $A_1$  and  $A_2$, respectively,
to estimate the first term  in (\ref{H-Q+++}) and to apply Propositions \ref{main-1} and \ref{main-2} in Section 4 in cases $B_1$  and  $B_2$, respectively,
to estimate the second term  in (\ref{H-Q+++}).\smallskip

The last claims of the proposition are almost obvious. It suffices only to note that in the case
$\,\Tr H_A\bar{\rho}(\mu)=E_{A}\,$ the definition of the energy-constrained diamond norm (\ref{ec-d-norm})
allows us to replace  $\|\Phi-\Psi\|_{1\to1}$ in (\ref{t-e}) with $\|\Phi-\Psi\|^{H_A}_{\diamond,E_A}$, while
in the case
$\,\Tr H_A\bar{\rho}(\nu)=E_{A}<+\infty\,$ we may obtain the same estimate by permuting  the roles of the pairs $(\Phi,\mu)$ and $(\Psi,\nu)$.
\end{proof}

\begin{example}\label{main-3-e} Let $\Phi=\id_A$ and $\Psi$ be the erasure channels $\Omega_p$ from a  quantum system $A$ to its
 "extension" $B$  defined as
\begin{equation*}
\Omega_p(\rho)=(1-p)\rho+p\shs[\Tr\rho] |\tau_0\rangle\langle\tau_0|,\quad \rho\in\S(\H_A),
\end{equation*}
where $\tau_0$ is a unit vector in $\H_B$ orthogonal to $\H_A\subset\H_B$, $p\in[0,1]$ \cite[Ch.6]{H-SCI}.
It is easy to see that $\,\frac{1}{2}\|\Omega_{p}-\Omega_{q}\|_{\diamond}=|p-q|\,$ and $\,\chi(\Omega_p(\mu))=(1-p)\chi(\mu)\,$ for any ensemble $\mu$ of input states.

Let $\mu=\{p_i,\rho_i\}$ be a discrete ensemble of pure input states such that $\bar{\rho}(\mu)$ is a state proportional to a projector of a finite rank $r_{\mu}\geq 2$.
For given $\varepsilon>0$ consider the ensemble $\nu=\{p_i,(1-\varepsilon)\rho_i+\varepsilon\sigma\}$, where $\sigma$ is a pure state in $\S(\H_A)$.
Then $\bar{\rho}(\nu)=(1-\varepsilon)\bar{\rho}(\mu)+\varepsilon\sigma$ and hence $\chi(\nu)\leq S(\bar{\rho}(\nu))\leq (1-\varepsilon)\ln r+h_2(\varepsilon)$.

Thus, as $\,\chi(\Psi(\nu))=(1-p)\chi(\nu)\,$, we have
\begin{equation}\label{e-v}
\chi(\Phi(\mu))-\chi(\Psi(\nu))\geq (p+\varepsilon-p\varepsilon)\ln r_{\mu}-(1-p)h_2(\varepsilon).
\end{equation}

It is easy to see that $D_*(\mu,\nu)\leq D_0(\mu,\nu)\leq\varepsilon$. Since $\rank\Psi((1-\varepsilon)\rho_i+\varepsilon\sigma)\leq3$ for any $i$, condition
(\ref{e-cond-3A}) holds for the ensemble $\nu$  with $r_\nu=3$.

Since we may identify the channels $\Phi=\id_A$ and $\Omega_0$, we have $\,\frac{1}{2}\|\Phi-\Psi\|_{\diamond}=p$.
So the claim of Proposition \ref{main-3} for the cases $A_1$ and $B_1$  implies that
\begin{equation}\label{a-v}
\!\chi(\Phi(\mu))-\chi(\Psi(\mu))\leq (p+\varepsilon)\ln[2(r_{\mu}-1)]+2h_2(p+\varepsilon)\;\;\textrm{provided that }\;  p+\varepsilon\leq 1/2.
\end{equation}

Comparing  (\ref{e-v}) and (\ref{a-v}) we see that
in this case Proposition \ref{main-3} gives asymptotically tight upper bound on $\chi(\Phi(\mu))-\chi(\Psi(\nu))$ for large $r_{\mu}$ and small $p$ and $\varepsilon$. \smallskip
\end{example}

\begin{example}\label{main-3-e+}
Assume that $A=B$ is an one-mode quantum oscillator and $H_A=H_B=\hat{N}\doteq a^{\dagger}a\,$ is the number operator on $\H_A$ \cite[Ch.12]{H-SCI}.

Let $\Upsilon(\zeta)=\vert\zeta\rangle\langle\zeta\vert$ be the coherent state of $A$ corresponding to a complex number $\zeta$. Let $\mu=\Upsilon(\pi)$, where $\pi$ is the probability measure on $\mathbb{C}$ with the density $p_{\mu}(\zeta)=\frac{1}{\pi N}\exp(-|\zeta|^2/N)$. The ensemble $\mu$ is typically denoted by $\{\pi,\vert \zeta\rangle\langle \zeta\vert\}_{\zeta\in\mathbb{C}}$ (see Section 6.2). Using this notation  define the ensemble  $\nu=\{\pi,(1-\varepsilon)\vert \zeta\rangle\langle \zeta\vert+\varepsilon\gamma_{\hat{N}}(N)\}_{\zeta\in\mathbb{C}}$, where  $\varepsilon\in(0,1/2]$ and $\gamma_{\hat{N}}(N)$ is
the Gibbs state corresponding to the "energy" (the number of quanta) $N$. Then $\bar{\rho}(\mu)=\bar{\rho}(\nu)=\gamma_{\hat{N}}(N)$. It is easy to see that
\begin{equation}\label{tmp-r}
\chi(\mu)=S(\gamma_{\hat{N}}(N))=g(N)\quad \textrm{and} \quad \chi(\nu)\leq (1-\varepsilon)S(\gamma_{\hat{N}}(N))=(1-\varepsilon)g(N).
\end{equation}
Note also that $\Tr\hat{N}\bar{\rho}(\mu)=N$ and $\bar{E}_{\hat{N}}^{\rm psv}(\nu)\leq 2\varepsilon N$.
To prove the last estimate  note that for each $\zeta\in\mathbb{C}$ the state
$(1-\varepsilon)\vert \zeta\rangle\langle \zeta\vert+\varepsilon\gamma_{\hat{N}}(N)$ is unitary equivalent to the state
$(1-\varepsilon)\vert 0\rangle\langle 0\vert+\varepsilon D^*(\zeta)\gamma_{\hat{N}}(N)D(\zeta)$, where $D(\zeta)$ is the displacement operator.
So,
$$
\begin{array}{c}
\displaystyle\bar{E}_{\hat{N}}^{\rm psv}(\nu)\leq \frac{1}{\pi N}\int_{\mathbb{C}} \Tr\hat{N}\left[(1-\varepsilon)\vert 0\rangle\langle 0\vert+\varepsilon D^*(\zeta)\gamma_{\hat{N}}(N)D(\zeta)\right]\exp(-|\zeta|^2/N)d^2\zeta\\\\\displaystyle=
\frac{\varepsilon}{\pi N}\int_{\mathbb{C}} (N+|\zeta|^2)\exp(-|\zeta|^2/N)d^2\zeta=2\varepsilon N.
\end{array}
$$
Using the upper bound\footnote{This upper bound is applicable here, since the map $\Upsilon(\zeta)$ is uniformly continuous (see Section 6.2).} (\ref{D-star-e+}) in Example \ref{D-star-e} we conclude that $D_*(\mu,\nu)\leq\varepsilon$.
Thus the claim of Proposition \ref{main-3} for the cases $A_2$ and $B_2$  (with $\Phi=\Psi=\id_A$) implies that
\begin{equation}\label{a-v+}
\chi(\mu)-\chi(\nu)\leq  \varepsilon (g(N/\varepsilon)+g(2N))+2g(\varepsilon)
\end{equation}
because  $F_{\hat{N}}=g$. \smallskip

It follows from (\ref{tmp-r}) that the r.h.s. of (\ref{a-v+}) is not less than $\varepsilon g(N)$. So,
since
\begin{equation}\label{s-ineq}
\varepsilon g(N/\varepsilon)-\varepsilon g(N)\leq -\varepsilon\ln \varepsilon+\varepsilon(1+\varepsilon/N)\leq 1/e+1+1/N,
\end{equation}
in this case Proposition \ref{main-3} gives asymptotically tight upper bound on $\,\chi(\mu)-\chi(\mu)\,$ for large $N$
up to factor $2$ in the main term.
\end{example} \smallskip

\begin{corollary}\label{main-3-c} \emph{Let $\Phi$ and $\Psi $ be quantum channels from $A$ to $B$. Let  $\mu$ and $\nu$ be (generalized) ensembles of states in $\S(\H_A)$ such that
\begin{equation}\label{e-cond-3-c}
D_X(\mu,\nu)+\textstyle\frac{1}{2}\|\Phi-\Psi\|_{Y}\leq\varepsilon\leq1,
\end{equation}
where $D_X$ is one of the metrics $D_K$ and $D_*$ defined, respectively, in [(\ref{K-D-d}),(\ref{K-D-c})] and [(\ref{f-metric}),(\ref{f-metric+})],
$\|\cdot\|_{Y}$ is either the diamond norm $\|\cdot\|_{\diamond}$ defined in (\ref{d-norm}) or its unstabilized version $\|\cdot\|_{1\to1}$ defined in (\ref{d-norm+}).}\smallskip

\noindent A) \emph{If
\begin{equation*}
\rank\Phi(\bar{\rho}(\mu))\leq r_\mu<+\infty\quad\textit{and}\quad\rank\Psi(\bar{\rho}(\nu))\leq r_\nu<+\infty
\end{equation*}
then
\begin{equation*}
|\chi(\Phi(\mu))-\chi(\Psi(\nu))|\leq C_\mu(\varepsilon)+C_\nu(\varepsilon),
\end{equation*}
where  $C_x(\varepsilon)=\varepsilon\ln(r_x-1)+h_2(\varepsilon)\,$ if $\;\varepsilon<1-1/r_x$ and $\,C_x(\varepsilon)=\ln r_x$ otherwise.}\smallskip

\noindent B) \emph{If
\begin{equation*}
\Tr H^{\mu}_B\Phi(\bar{\rho}(\mu))\leq E_\mu<+\infty\quad \textit{and}\quad\Tr H^{\nu}_B\Psi(\bar{\rho}(\nu))\leq E_\nu<+\infty,
\end{equation*}
where $H^{\mu}_B$ and $H^{\nu}_B$ are positive operators on  $\H_B$ satisfying conditions (\ref{H-cond}) and (\ref{star}), then
\begin{equation}\label{main-3B-c}
|\chi(\Phi(\mu))-\chi(\Psi(\nu))|\leq \varepsilon F_{H^{\mu}_B}(E_\mu/\varepsilon)+\varepsilon F_{H^{\nu}_B}(E_\nu/\varepsilon)+2g(\varepsilon),
\end{equation}
where $F_{H^{x}_B}$, $x=\mu,\nu$, and $\,g$ are the functions defined, respectively, in (\ref{F-def}) and (\ref{g-def}).}\smallskip

\noindent C) \emph{The following replacements can be done independently on each other:}
\begin{enumerate}[(1)]
\item \emph{if $\mu=(p_i,\rho_i)$ and $\nu=(q_i,\sigma_i)$ are discrete ordered
ensembles then the metric $D_0$ defined in (\ref{D-0-metric}) can be used in the role of $D_X$  in (\ref{e-cond-3-c});} \smallskip
\item \emph{if either $\,\Tr H_A\bar{\rho}(\mu)=E_{A}<+\infty$ or $\,\Tr H_A\bar{\rho}(\nu)=E_{A}<+\infty$, where  $H_A$ is a positive operator on  $\H_A$,  then the energy-constrained diamond norm $\|\cdot\|^{H_A}_{\diamond,E_A}$ defined in (\ref{ec-d-norm}) can be used in the role of $\|\cdot\|_{Y}$ in  (\ref{e-cond-3-c}).}
\end{enumerate}

\end{corollary}\smallskip

\noindent The r.h.s. of (\ref{main-3B-c}) is a nondecreasing function of $\,\varepsilon$, which tends to zero as $\,\varepsilon\to0\,$. This follows from inequality (\ref{W-L}) and the equivalence of  (\ref{H-cond}) and  (\ref{H-cond-a}).

\smallskip

Corollary \ref{main-3-c}A (with $\Phi=\Psi=\id_A$) implies that
\begin{equation}\label{main-3A-c+}
|\chi(\mu)-\chi(\nu)|\leq 2\varepsilon\ln(d-1)+2h_2(\varepsilon)
\end{equation}
for any ensembles $\mu$ and $\nu$ of states in a $d$-dimensional Hilbert space provided that
$D_*(\mu,\nu)\leq\varepsilon\leq1-1/d$. The same inequality is proved by Oreshkov and Calsamiglia
in \cite[Section IV-B]{O&C} under the stronger condition $D_K(\mu,\nu)\leq\varepsilon\leq1-1/d$ (it is stronger as
$D_K(\mu,\nu)$ may be greater than $D_*(\mu,\nu)$). Another advantage of continuity bound (\ref{main-3A-c+})
in comparison with the Oreshkov-Calsamiglia continuity bound consists in the fact that for discrete ensembles $\mu$ and $\nu$
the metric $D_*(\mu,\nu)$ (unlike $D_K(\mu,\nu)$) is bounded from above by the easily-computable quantity $D_0(\mu,\nu)$ defined in (\ref{D-0-metric}).
\smallskip

Corollary \ref{main-3-c}B (with $\Phi=\Psi=\id_A$) implies that
\begin{equation}\label{main-3B-c+}
|\chi(\mu)-\chi(\nu)|\leq 2\varepsilon F_{H}(E/\varepsilon)+2g(\varepsilon)
\end{equation}
for any ensembles $\mu$ and $\nu$ of states in an infinite-dimensional separable Hilbert space $\H$ provided that
$D_*(\mu,\nu)\leq\varepsilon\leq1$ and $\Tr H\bar{\rho}(\mu),\Tr H\bar{\rho}(\nu)\leq E<+\infty$, where
$H$ is a positive operator on $\H$ satisfying conditions (\ref{H-cond}) and (\ref{star}).

It is reasonable to compare continuity bounds (\ref{main-3A-c+}) and (\ref{main-3B-c+})
for the Holevo information with the analogous continuity bounds obtained in \cite{CHI} (which, as far as I know, have been the most accurate estimates of this kind so far).\footnote{We mean the continuity bounds for the Holevo information with the dimension/energy constraints. There are continuity bounds for the Holevo information
with constraints of another type, for example, with the constraint on the number of states in an ensemble \cite{Aud,CHI}.}

Proposition 16 in \cite{CHI} implies that
\begin{equation}\label{CHI-CB-1}
|\chi(\mu)-\chi(\nu)|\leq \varepsilon\ln d+2g(\varepsilon)
\end{equation}
for any ensembles $\mu$ and $\nu$ of states in a $d$-dimensional Hilbert space provided that
$D_*(\mu,\nu)\leq\varepsilon$, where $g$ is the function defined in (\ref{g-def}).
\smallskip

It is mentioned in \cite{CHI} that continuity bound (\ref{CHI-CB-1}) is asymptotically tight for large $d$.
It means that continuity bound (\ref{main-3A-c+}) does not possess this property because of the first factor 2
in its r.h.s. It turns out, however, that continuity bound (\ref{main-3A-c+}) is \emph{more
accurate} than continuity bound (\ref{CHI-CB-1}) for low dimensions $d$. Indeed,
the difference between the right hand sides of  (\ref{CHI-CB-1}) and (\ref{main-3A-c+}) is equal to
$$
\textstyle\varepsilon\!\left(\frac{2}{\varepsilon}\ln(1-\varepsilon^2)+2\ln\frac{1+\varepsilon}{1-\varepsilon}-\ln d-2\ln\!\left(1-\frac{1}{d}\right)\right)=
\varepsilon\!\left(\ln\!\left[(1-\varepsilon^2)^{\frac{1}{\varepsilon}}\,\frac{1+\varepsilon}{1-\varepsilon}\right]^2-\ln\!\left[d\left(1-\frac{1}{d}\right)^2\right]\right)
$$
Thus, the continuity bound (\ref{main-3A-c+}) is more
accurate than continuity bound (\ref{CHI-CB-1}) provided that
$$
(1-\varepsilon^2)^{\frac{2}{\varepsilon}}\left(\frac{1+\varepsilon}{1-\varepsilon}\right)^2\geq d\left(1-\frac{1}{d}\right)^2
$$
The function $u(\varepsilon)=(1-\varepsilon^2)^{\frac{2}{\varepsilon}}\left(\frac{1+\varepsilon}{1-\varepsilon}\right)^2$ increases on $[0,1]$ from
$u(0)=1$ to $u(1)=16$, while the function $v(d)=d\left(1-\frac{1}{d}\right)^2$ increases on $\N\cup[2,17]$ from $v(2)=1/2$ to\break $v(17)=256/17\approx15.06$ and $v(18)=289/18\approx16.06$. Thus, for each
$d=2,3,...,17$  continuity bound (\ref{main-3A-c+}) is more
accurate than continuity bound (\ref{CHI-CB-1}) for all  $\varepsilon\in(\varepsilon_d,1]$, where $\varepsilon_2=0$ and $\varepsilon_d$
is a unique solution in $(0,1)$ of the equation $u(\varepsilon)=v(d)$. It is easy to see that $\varepsilon_3\approx0.11$, $\varepsilon_4\approx0.45$ and $\varepsilon_d>\varepsilon_5\approx0.55$ for $d=6,7,...17$. \emph{Thus, continuity bound (\ref{main-3A-c+}) is better than continuity bound (\ref{CHI-CB-1})
in the qubit case $d=2$ for all $\varepsilon$, while in the cases $d=3,4,5$  the same is true for  significant ranges of $\,\varepsilon$.} Continuity bound (\ref{main-3A-c+}) is less accurate than continuity bound (\ref{CHI-CB-1}) for all $\varepsilon\in(0,1]$ if and only if $\,d\geq18$.\smallskip

Proposition 21 in \cite{CHI} implies that
\begin{equation}\label{CHI-CB-2}
\left|\chi(\mu)-\chi(\nu)\right|\leq\inf_{t\in(0,\frac{1}{2\varepsilon}]}\left[\,\varepsilon(2t+r_{\!\varepsilon}(t))F_{H}(E/(\varepsilon t))
+2g(\varepsilon r_{\!\varepsilon}(t))+2h_2(\varepsilon t)\,\right],
\end{equation}
for any ensembles $\mu$ and $\nu$ of states in an infinite-dimensional separable Hilbert space $\H$ provided that
$D_*(\mu,\nu)\leq\varepsilon\leq1$ and $\Tr H\bar{\rho}(\mu),\Tr H\bar{\rho}(\nu)\leq E<+\infty$, where
$H$ is a positive operator on $\H$ satisfying conditions (\ref{H-cond}),(\ref{star}) and $r_{\!\varepsilon}(t)=(1+t/2)/(1-\varepsilon t)$.\smallskip

It is mentioned in \cite{CHI} that continuity bound (\ref{CHI-CB-2}) is asymptotically tight for large $E$.
In contrast, continuity bound (\ref{main-3B-c+}) does not possess this property because of the first factor 2
in its r.h.s. An obvious advantage of  continuity bound (\ref{main-3B-c+}) in comparison with (\ref{CHI-CB-2}) is its
simplicity. Moreover, the numerical analysis with the number operator $\hat{N}$ of a quantum oscillator in the role of $H$
shows that  continuity bound (\ref{main-3A-c+}) is \emph{more
accurate} than continuity bound (\ref{CHI-CB-2}) in the non-asymptotic regime, especially, for low energy bounds $E$ ($E\leq 10$).

\subsection{Upper bounds on the Average Entropy of an ensemble and the AOE of a channel}

The semicontinuity bounds for the  Average Output Entropy  (AOE) of a quantum channel  presented in Section 4 can be used to obtain
quite accurate upper bounds on the AOE and on the Average Entropy (AE)
\begin{equation*}
  \overline{S}(\mu)\doteq\int_{\S(\H)}S(\rho)\mu(d\rho).
\end{equation*}
of an ensemble $\mu$ of quantum states in $\S(\H)$.

Note first that the following obvious upper bounds on the AE are valid:
if $\mu$ is a (generalized) ensemble in $\P(\H)$ such that either $\rank\rho\leq r<+\infty$ for $\mu$-almost all $\rho$ (case $A$) or
$\bar{E}_{H}^{\rm psv}(\mu)<+\infty$ for some positive operator $H$ on $\H$ satisfying the Gibbs condition (\ref{H-cond}) (case $B$)\footnote{$\bar{E}_{H}^{\rm psv}(\mu)$
is the passive energy of $\mu$ w.r.t. $H$ defined in (\ref{p-a-e-def}).}
then
\begin{equation}\label{OUB}
  \overline{S}(\mu)\leq \ln r\;\;\textrm{ in case }A\quad\textrm{ and }\quad \overline{S}(\mu)\leq F_H(\bar{E}_{H}^{\rm psv}(\mu))\;\;\textrm{ in case }B,
\end{equation}
where $F_H$ is the function defined in (\ref{F-def}). The second upper bound in (\ref{OUB}) was mentioned in Section 2.2.3 (inequality (\ref{P-B+})).\smallskip

Propositions \ref{main-1} and \ref{main-2} in Section 4 imply the following\smallskip

\begin{proposition}\label{AE-UB} \emph{Let $\H$  be a separable Hilbert space and $\mu$ a (generalized) ensemble in $\P(\H)$.}\footnote{$\P(\H)$ is the set of all generalized ensembles of states in $\S(\H)$, see Section 2.2.2.}
\emph{Let
$$
\Delta(\mu)\doteq\inf\left\{D_*(\mu,\nu)\,|\,\nu\in\P(\H),\; \rank\rho=1\textit{ for }\mu\textit{-almost all }\rho\shs\right\},
$$
where $D_*$ is the metric in $\P(\H)$ defined in (\ref{f-metric}) and (\ref{f-metric+}).}\smallskip

\noindent A) \emph{If $\,\rank\rho\leq r\in\N\cap[2,+\infty)\,$\footnote{If this condition holds with $r=1$ then $\overline{S}(\mu)=0$.} for $\mu$-almost all $\rho$ then}
\begin{equation*}
\overline{S}(\mu)\leq \Delta(\mu)\ln (r-1)+h_2(\Delta(\mu))\quad\textit{ provided that } \Delta(\mu)\leq 1-1/r.
\end{equation*}

\noindent B) \emph{If $\,\bar{E}_{H}^{\rm psv}(\mu)<+\infty\,$ for some positive operator $H$ on $\H$ satisfying  conditions (\ref{H-cond}) and (\ref{star}) then
\begin{equation}\label{UB-2}
\overline{S}(\mu)\leq \Delta(\mu)F_H(\bar{E}_{H}^{\rm psv}(\mu)/\Delta(\mu))+g(\Delta(\mu)),
\end{equation}
where $F_H$ and $g$ are the functions defined in (\ref{F-def}) and (\ref{g-def}).}

\noindent C) \emph{Claims $A$ and $B$ remain valid with $\Delta(\mu)$ replaced by any its upper bound $\widetilde{\Delta}(\mu)$, in particular,
by $\,D_*(\mu,\nu_0)$, where $\nu_0$ is a given ensemble of pure states in $\S(\H)$.}
\end{proposition}\smallskip

\noindent The r.h.s. of (\ref{UB-2}) is a nondecreasing function of $\,\Delta(\mu)$, which tends to zero as $\,\varepsilon\to0\,$. This follows from inequality (\ref{W-L}) and the equivalence of  (\ref{H-cond}) and  (\ref{H-cond-a}).

\smallskip

\begin{example}\label{UB-1} Assume that $\mu=\{p_i,\rho_i\}$ is a discrete ensemble of states
in $\S(\H)$ such that the maximal eigenvalue of any state $\rho_i$ is not less than $\,1-p\,$ and
$$
\textrm{either }\;\rank\rho_i\leq r\in\N\cap[2,+\infty)\;\; \forall i\;\, (\textrm{case }A)\;\textrm{ or }\;\bar{E}_{H}^{\rm psv}(\mu)<+\infty\;\, (\textrm{case }B),
$$
where $H$ is a positive operator on $\H$ satisfying  conditions (\ref{H-cond}) and (\ref{star}).

Consider the ensemble $\nu=\{p_i,\hat{\rho}_i\}$, where $\hat{\rho}_i$ is a pure state corresponding to the
maximal eigenvalue of $\rho_i$ for each $i$. Then it is easy to see that
$D_*(\mu,\nu)\leq D_0(\mu,\nu)\leq p$. So, since $\nu$ is an ensemble if pure states,
Proposition \ref{AE-UB} implies that
\begin{equation}\label{c-A}
\overline{S}(\mu)\leq p\ln (r-1)+h_2(p)\quad\textrm{provided that }\;p\leq 1-1/r\quad\textrm{ in case }A.
\end{equation}
and
\begin{equation}\label{c-B}
\overline{S}(\mu)\leq pF_H(\bar{E}_{H}^{\rm psv}(\mu)/p)+g(p)\quad \textrm{ in  case } B.
\end{equation}

Let $\mu$ be an ensemble consisting of states with the spectrum $\{1-p, \frac{p}{r-1},...,\frac{p}{r-1}\}$, $p\leq 1-1/r$. Then
$\overline{S}(\mu)= p\ln (r-1)+h_2(p)$. So, the upper bound (\ref{c-A}) is \emph{optimal}.\smallskip

Let $\mu$ be an ensemble consisting of states with the spectrum
\begin{equation}\label{S-P-d}
\{1-p, p(1-q),p(1-q)q,p(1-q)q^2,...\}
\end{equation}
where
$q=N/(N+1)$, $N>0$. Assume that $H$ is a positive operator on $\H$ having the form (\ref{H-form}) with the spectrum $\{0,0,1,2,3,....\}$. Then it is easy
to see that $\bar{E}_{H}^{\rm psv}(\mu)=pN$. By using (\ref{F-def}),(\ref{F-N}) and (\ref{W-L}) one can show that
$$
g(E)\leq F_H(E)=\max_{x\in[0,1]}\left[(1-x)g(E/(1-x))+h_2(x)\right]=g(E)+o(E)\;\textrm{ as }\; E\to +\infty,
$$
where $o(E)\leq\ln 2$. So, the r.h.s. of (\ref{c-B}) is equal to $\,pg(N)+g(p)+p\shs o(N)\,$ as $N\to +\infty$, while the l.h.s. of (\ref{c-B})
is equal to the Shannon entropy $\,pg(N)+h_2(p)\,$ of the probability distribution in (\ref{S-P-d}). Hence,
the upper bound (\ref{c-B}) is \emph{asymptotically
tight} for large $E$.\medskip
\end{example}

The upper bounds on the AE in (\ref{OUB}) imply the upper  bounds on the
AOE of a quantum channel $\Phi$ from $A$ to $B$: if $\mu$ is a (generalized) ensemble in $\P(\H_A)$ such that either $\rank\Phi(\rho)\leq r<+\infty$ for $\mu$-almost all $\rho$ (case $A$) or
$\bar{E}_{H_B}^{\rm psv}(\Phi(\mu))<+\infty$ for some positive operator $H_B$ on $\H_B$ satisfying  condition (\ref{H-cond}) (case $B$)
then
\begin{equation}\label{OUB+}
  \overline{S}(\Phi(\mu))\leq \ln r\;\;\textrm{ in case }A\quad\textrm{ and }\quad \overline{S}(\Phi(\mu))\leq F_{H_B}(\bar{E}_{H_B}^{\rm psv}(\Phi(\mu)))\;\;\textrm{ in case }B,
\end{equation}
where $F_{H_B}$ is the function defined in (\ref{F-def}).

By using Propositions \ref{main-1} and \ref{main-2} in Section 4 one can strengthen the upper bounds in (\ref{OUB+})
in some special cases. For example, one can prove the following\smallskip

\begin{proposition}\label{AOE-UB}
\emph{Let
$\Phi$ be an  arbitrary channel from $A$ to $B$ and $\mu$ a (generalized) ensemble of pure states in $\S(\H_A)$
such that $\,\bar{E}_{H_B}^{\rm psv}(\Phi(\mu))<+\infty\,$, where $H_B$ is a positive operator on $\H_B$ satisfying conditions (\ref{H-cond}) and (\ref{star})}. \emph{Then
\begin{equation}\label{AOE-UB+}
  \overline{S}(\Phi(\mu))\leq \ln r+\Delta^r(\Phi) F_{H_B}(\bar{E}_{H_B}^{\rm psv}(\Phi(\mu))/\Delta^r(\Phi))+g(\Delta^r(\Phi)) \quad \forall r\in\N,
\end{equation}
where\footnote{$\|\cdot\|_{1\to1}$ is the metric defined in (\ref{d-norm+}).}
$$
\Delta^r(\Phi)\doteq\inf\left\{\textstyle\frac{1}{2}\|\Phi-\Psi\|_{1\to1}\,|\,\Psi\textrm{ is a channel with the Choi rank }\leq r\shs\right\},
$$
i.e.$\Delta^r(\Phi)$ is the $\|\cdot\|_{1\to1}$-distance from $\Phi$ to the set of channels with the Choi rank $\leq r$.}

\emph{The quantity $\Delta^r(\Phi)$ in (\ref{AOE-UB+}) can be replaced by any its upper bound $\widetilde{\Delta}^r(\Phi)$, in particular, by $\,\frac{1}{2}\|\Phi-\Psi_0\|_{1\to1}$, where $\Psi_0$ is any channel with the Choi rank $\leq r$.}
\end{proposition}\smallskip

Note that upper bound (\ref{AOE-UB+}) turns into the first upper bound in (\ref{OUB+}) if the Choi rank of $\Phi$ does not exceed $r$.

\begin{proof} It suffices to use Proposition \ref{main-2} and the first upper bound in (\ref{OUB+}) by noting
that $\rank\Psi(\rho)\leq r$  for $\mu$-almost all $\rho$ provided that the Choi rank of $\Psi$ does not exceed $r$ (due to the assumed purity of the states
of the ensemble $\mu$).
\end{proof}

\subsection{New semicontinuity bound for the entanglement of formation}

The entanglement of
formation (EoF) in a finite-dimensional bipartite system $AB$  is defined as the convex roof extension to the set $\S(\H_{AB})$ of the function $\omega\mapsto S(\omega_{A})$ on the set $\mathrm{ext}\shs\S(\H_{AB})$ of pure states in $\S(\H_{AB})$ , i.e.
\begin{equation*}
  E_{F}(\omega)=\inf_{\sum_kp_k\omega_k=\omega}\sum_k p_kS([\omega_k]_{A})=\inf_{\sum_kp_k\omega_k=\omega}\overline{S}(\{p_k,[\omega_k]_{A}\}),
\end{equation*}
where the infima are over all finite ensembles $\{p_k, \omega_k\}$ of pure states in $\S(\H_{AB})$ with the average state $\omega$ \cite{Bennett}.

In an infinite-dimensional bipartite system $AB$ there are two versions $E_{F,d}$ and $E_{F,c}$ (discrete and continuous) of the EoF defined in (\ref{EoF}).
It follows from the definitions that $E_{F,d}(\omega)\geq E_{F,c}(\omega)$ for any state $\omega\in\S(\H_{AB})$.
In \cite{EM} it is shown that $E_{F,d}(\omega)=E_{F,c}(\omega)$ for any state $\omega$ in $\S(\H_{AB})$ such that $\min\{S(\omega_{A}),S(\omega_{B})\}<+\infty$.  The conjecture of coincidence of $E_{F,d}$ and $E_{F,c}$ on $\S(\H_{AB})$ is an interesting open question (see details in \cite[Section 5]{EM}).

By using the semicontinuity bound for the QCE of quantum-classical states presented Lemma \ref{Wilde-SCB} in Section 2.3 one can  obtain the following
result strengthening the claim of Proposition 4A in \cite{LCB}. \smallskip

\begin{proposition}\label{new-EF-SCB} \emph{Let $AB$ be an infinite-dimensional bipartite quantum system.}\smallskip

\emph{If $\rho$ is a state in $\S(\H_{AB})$ such that $\,r=\min\{\rank\rho_A,\rank \rho_B\}<+\infty$ then
\begin{equation}\label{new-EF-SCB+}
   E^*_F(\rho)-E^*_F(\sigma)\leq\delta\ln(r-1)+h_2(\delta),\quad E^*_F=E_{F,d},E_{F,c},
\end{equation}
for any state $\sigma$ in $\S(\H_{AB})$ such that $\,\frac{1}{2}\|\rho-\sigma\|_1\leq \varepsilon\leq 1-r^{-1}\sqrt{2r-1}$,  where $\delta=\sqrt{\varepsilon(2-\varepsilon)}$ and  the l.h.s. of (\ref{new-EF-SCB+}) may be equal to $-\infty$.}\footnote{The condition $\varepsilon\leq 1-r^{-1}\sqrt{2r-1}$ means that $\delta\leq1-1/r$.}
\end{proposition}\smallskip

The claims of Proposition \ref{new-EF-SCB} are  deduced from Lemma \ref{Wilde-SCB} in Section 2.3 by the same arguments (the Nielsen-Winter technique plus approximation) that are used to deduce the claims of Proposition 4A in \cite{LCB} from  Proposition 3A in \cite{LCB}.

The semicontinuity bound (\ref{new-EF-SCB+}) agrees with Wilde's continuity bound for the entanglement of formation
obtained in \cite{Wilde-CB}.\medskip

\begin{remark}\label{EoF-fidelity}
In Proposition \ref{new-EF-SCB},  the condition $\,\frac{1}{2}\|\rho-\sigma\|_1\leq \varepsilon$, where
$\varepsilon$ is such that $\,\delta=\sqrt{\varepsilon(2-\varepsilon)}$, can be replaced by the
condition  $\,F(\rho,\sigma)\geq 1-\delta^2$, where  $F(\rho,\sigma)$
is the fidelity of $\rho$ and $\sigma$ defined in (\ref{fidelity}).  This follows from the  Nielsen-Winter arguments used
in the proof. The use of fidelity as a measure of closeness of the
states $\rho$ and $\sigma$ makes the semicontinuity bounds in (\ref{new-EF-SCB+}) essentially sharper in some cases. Moreover,
the semicontinuity bounds in (\ref{new-EF-SCB+}) with $\,\delta=\sqrt{1-F(\rho,\sigma)}$ are \emph{tight} for large $r$.

To show this assume that $\varphi$ is a unit vector in $\H_{AB}$ such that $\Tr_B|\varphi\rangle\langle\varphi|$ is
a state proportional to a projector of rank $r$. Assume that $\alpha$ and $\beta$ are unit vectors in $\H_{A}$ and $\H_B$, respectively, such that
$\alpha\,\bot\, \supp \Tr_B|\varphi\rangle\langle\varphi|$ and $\beta\,\bot\,\supp \Tr_A|\varphi\rangle\langle\varphi|$.
Consider the pure states $\rho=|\theta_{\frac{1}{2}-\delta}\rangle\langle\theta_{\frac{1}{2}-\delta}|$ and $\sigma=|\theta_{\frac{1}{2}}\rangle\langle\theta_{\frac{1}{2}}|$,
where $\theta_p=\sqrt{1-p}\,\varphi+\sqrt{p}\,\alpha\otimes\beta$ is a unit vector in $\H_{AB}$ for any $p\in[0,1]$. Then we have
$$
\begin{array}{c}
1-F(\rho,\sigma)\leq[\beta(\rho,\sigma)]^2\leq \|\theta_{\frac{1}{2}}-\theta_{\frac{1}{2}-\delta}\|^2=\left[\sqrt{\frac{1}{2}+\delta}-\sqrt{\frac{1}{2}}\right]^2+\left[\sqrt{\frac{1}{2}-\delta}-\sqrt{\frac{1}{2}}\right]^2\\\\
=2-\sqrt{1+2\delta}-\sqrt{1-2\delta}=\delta^2+o(\delta^2)\quad \textrm{as}\;\; \delta \to0,
\end{array}
$$
where $\beta(\rho,\sigma)$ is the Bures distance defined in (\ref{B-dist}).

Thus, the semicontinuity bounds in (\ref{new-EF-SCB+}) with $\,\delta=\sqrt{1-F(\rho,\sigma)}$ imply that
\begin{equation}\label{EoF++}
E^*_F(\rho)-E^*_F(\sigma)\leq (\delta+o(\delta))\ln r+h_2(\delta+o(\delta))\quad \textrm{as}\;\; \delta \to0,
\end{equation}
where $E^*_F(\vartheta)=E_{F,d}(\vartheta)=E_{F,c}(\vartheta)$, $\vartheta=\rho,\sigma$ (because $\rank\rho_A=r+1$).\smallskip

Since $\rho_A=(\frac{1}{2}+\delta)\Tr_B|\varphi\rangle\langle\varphi|+(\frac{1}{2}-\delta)|\alpha\rangle\langle\alpha|$ and
$\sigma_A=\frac{1}{2}\Tr_B|\varphi\rangle\langle\varphi|+\frac{1}{2}|\alpha\rangle\langle\alpha|$, we have
$$
\!E^*_F(\rho)-E^*_F(\sigma)=S(\rho_A)-S(\sigma_A)=\delta\ln r+h_2(\textstyle\frac{1}{2}+\delta)-h_2(\frac{1}{2})=\delta\ln r-o(\delta)\;\; \textrm{as}\;\, \delta \to0.
$$
By comparing this with (\ref{EoF++}) we see  that the semicontinuity bounds in (\ref{new-EF-SCB+}) with $\,\delta=\sqrt{1-F(\rho,\sigma)}$ are \emph{tight} for large $r$.
\end{remark}\medskip

\begin{corollary}\label{new-EF-SCB-c} \emph{Let $AB$ be an infinite-dimensional bipartite quantum system.}\smallskip

\emph{Let $\rho$ be a state in $\S(\H_{AB})$ such that $\,r=\min\{\rank\rho_A,\rank\rho_B\}<+\infty$  and
$$
\Delta^{\!F}(\rho)\doteq \sqrt{\,1-\sup_{\sigma\in\S_{\rm sep}(\H_{AB})}F(\rho,\sigma)}\leq 1-1/r,
$$
where $\S_{\rm sep}(\H_{AB})$ is the set of separable states in $\S(\H_{AB})$. Then}
\begin{equation}\label{new-EF-SCB+c}
   E_{F,d}(\rho)=E_{F,c}(\rho)\leq\Delta^{\!F}(\rho)\ln(r-1)+h_2(\Delta^{\!F}(\rho)).
\end{equation}

\emph{The quantity $\Delta^{\!F}(\rho)$ in (\ref{new-EF-SCB+c}) can be replaced by any its upper bound $\widetilde{\Delta}^{\!F}(\rho)$ not exceeding $1-1/r$, in particular,
by the Bures distance from $\rho$ to $\S_{\rm sep}(\H_{AB})$ defined as
$$
\Delta^{\beta}(\rho)\doteq \inf_{\sigma\in\S_{\rm sep}(\H_{AB})}\beta(\rho,\sigma)
$$
or by any of the quantities $\beta(\rho,\sigma_0)$ and $\sqrt{1-F(\rho,\sigma_0)}$, where $\sigma_0$ is any  separable state.}
\end{corollary}\smallskip

\begin{proof} To prove (\ref{new-EF-SCB+c}) it suffices to use the semicontinuity bound in (\ref{new-EF-SCB+})
for $E_{F,c}$ with $\,\delta=\sqrt{1-F(\rho,\sigma)}$ (see Remark \ref{EoF-fidelity}), since
$E_{F,c}(\sigma)=0$ for any separable state $\sigma\in\S(\H_{AB})$. Note that the
same property of $E_{F,d}$ is not proved yet (because of the existence of countably-non-decomposable separable states in infinite-dimensional
bipartite systems).

The last claim follows from the fact
that the function $x\mapsto x\ln(r-1)+h_2(x)$ is nondecreasing on $[0,1-1/r]$ and the well known inequality $\,\sqrt{1-F(\rho,\sigma)}\leq\beta(\rho,\sigma)$ \cite{H-SCI,Wilde}.
\end{proof}

\section{How to apply the results of this article dealing with the concept of ensemble commonly used in QIT}

\subsection{General observations}

In this section we describe the relations between the notion of a generalized  ensemble
introduced in Section 2.2.2  and  the notion of a generalized (continuous) ensemble commonly used in quantum information theory,
quantum statistics, etc. We also show how the main results of this article can be reformulated in terms of the latter concept of ensemble.

Let $\X$ be a measurable space. According to the commonly used notation (cf.\cite{H-SCI,H-GE}) a generalized ensemble $\{\pi, \rho_x\}_{x\in\X}$
consists of a probability measure $\pi$  on $\X$ and a "measurable family" $\{\rho_x\}_{x\in\X}$ of
quantum states on $\H$. It means, mathematically, that $\rho_x=\Upsilon(x)$, where $\Upsilon$ is a measurable $\S(\H)$-valued function on $\X$.
So, the image $\Upsilon(\pi)$ of the measure $\pi$ under the map $\Upsilon$ (also denoted by $\pi\circ\Upsilon^{-1}$, since
$\Upsilon(\pi)[\S]=\pi(\Upsilon^{-1}[\S])$ for any Borel subset $\S\subseteq\S(\H)$) is a Borel probability measure on $\S(\H)$, i.e.
$\Upsilon(\pi)$ is a (generalized) ensemble in $\P(\H)$ in terms of Section 2.2.2.

If $\X=\S(\H)$ and $\Upsilon$ is the identity map then  $\,\{\pi, \rho_x\}_{x\in\X}=\pi\,$ for any $\pi$ in $\P(\H)$, but in general
the above correspondence  $\,\{\pi, \rho_x\}_{x\in\X}\to \Upsilon(\pi)\,$ is not bijective and surjective:
the measures $\Upsilon(\pi)$ and $\Upsilon(\pi')$ may coincide for different $\pi$ and $\pi'$ and
not every measure from $\P(\H)$ can be represented as $\Upsilon(\pi)$.

Note also that when we speak about some information characteristic of an ensemble  $\{\pi, \rho_x\}_{x\in\X}$
we really speak about this characteristic of the corresponding ensemble $\Upsilon(\pi)$. For example, the Holevo information
$$
\chi(\{\pi, \rho_x\})\doteq \int_{\X}D(\rho_x\,\|\bar{\rho}_x)\pi(dx),\quad \bar{\rho}_x=\int_{\X}\rho_x\pi(dx),
$$
of an ensemble  $\{\pi, \rho_x\}_{x\in\X}$ is equal (due to the basic property of the Lebesgue integral) to the Holevo information
$$
\chi(\Upsilon(\pi))\doteq \int_{\S(\H)}D(\rho\,\|\bar{\rho}(\mu))\mu(d\rho),\quad \mu=\Upsilon(\pi)
$$
of the corresponding ensemble $\Upsilon(\pi)$.

Thus, the notion of a generalized ensemble considered in  Section 2.2.2 is not no poorer (narrower) than the notion
described at the begin of this subsection.\footnote{In fact, the notion of a generalized ensemble considered in  Section 2.2.2 gives a great convenience in proving theoretical results.} So, we may  formulate the following
\smallskip

\begin{remark}\label{b-rem} All the semicontinuity and continuity bounds for the AOE and the output Holevo information
a channel obtained in this article can be reformulated
by using the notion of ensemble described at the begin of this subsection.
\end{remark}\smallskip

The only problem consist in the fact that, in general, we have to use the distances $D_X(\Upsilon(\pi_1),\Upsilon(\pi_2))$, $D_X=D_K,D_*$,
in all the semicontinuity and continuity bounds in which probability measures $\pi_1$ and $\pi_2$ on $\X$ are involved. It is very natural as long as we do not assume any metric on $\X$ and continuity properties of the map $\Upsilon$. But if $\X$ is a metric space and $\Upsilon$ is a Lipschitz map then
we can define the Kantorovich–Rubinshtein distance and the modified Kantorovich–Rubinshtein distance between Borel probability measures $\pi_1$ and $\pi_2$ on $\X$
in the way described in Appendix A-1 and use the following\smallskip\pagebreak

\begin{lemma}\label{K-L} \emph{If $\,\Upsilon:\X\to\S(\H)$ is a map with the  Lipschitz constant $L$ then\footnote{We assume that $\S(\H)$ is equipped with the trace norm metric $d(\rho,\sigma)\doteq\|\rho-\sigma\|_1$.}
$$
D_K(\Upsilon(\pi_1),\Upsilon(\pi_2))\leq \textstyle\frac{1}{2}\shs LD_{KR}(\pi_1,\pi_2)
$$
for any Borel probability measures $\pi_1$ and $\pi_2$ on $\X$, where $D_{KR}(\pi_1,\pi_2)$ is
the Kantorovich–Rubinshtein distance between $\pi_1$ and $\pi_2$ described in Appendix A-1.}

\emph{If $\pi_1$ and $\pi_2$ are measures in the set $\P_*(\X)$ (defined in  Appendix A-1) then
$$
D_K(\Upsilon(\pi_1),\Upsilon(\pi_2))\leq \textstyle\frac{1}{2}\shs LD^*_{KR}(\pi_1,\pi_2),
$$
where $D^*_{KR}(\pi_1,\pi_2)$ is  the modified Kantorovich–Rubinshtein distance between $\pi_1$ and $\pi_2$ described in Appendix A-1.}
\end{lemma}\smallskip

Both claims of Lemma \ref{K-L} follow directly from the definitions of $D_K(\Upsilon(\pi_1),\Upsilon(\pi_2))$, $D_{KR}(\pi_1,\pi_2)$ and $D^*_{KR}(\pi_1,\pi_2)$
and the remark at the end of   Appendix A-1.\smallskip

By applying  Lemma \ref{K-L} one can reformulate  all the semicontinuity and continuity bounds obtained  in this article with the use of
the distances  $D_{KR}(\pi_\mu,\pi_\nu)$ and $D^*_{KR}(\pi_\mu,\pi_\nu)$ between the "original" measures $\pi_\mu$ and $\pi_\nu$ determining the ensembles $\mu=\{\pi_\mu,\rho_x\}_{x\in\X}$ and $\nu=\{\pi_\nu,\rho_x\}_{x\in\X}$.

\subsection{Example: ensembles of coherent states}

Assume that $A_1...A_n$ is a $n$-mode quantum oscillator and  $\{\vert \bar{\zeta}\rangle\langle \bar{\zeta}\vert\}_{\bar{\zeta}\in\mathbb{C}^n}$ is the family of coherent states \cite{Gla,IQO},\cite[Ch.12]{H-SCI}. Coherent states plays a special role in the "continuous variable" quantum information theory. In particular, it is proved that coherent states forms optimal or close-to-optimal codes for transmitting classical information through many Gaussian channels  \cite{MGH}. Mathematically, this is expressed by the fact that
the supremum of the output Holevo information of such channels is attained at ensembles of coherent states, i.e. ensembles of the form
$\{\pi, \vert \bar{\zeta}\rangle\langle \bar{\zeta}\vert\}_{\bar{\zeta}\in\mathbb{C}^n}$, where $\pi$ is a Borel probability measure on $\mathbb{C}^n$.

Note that
the map $\mathbb{C}^n\ni\bar{\zeta}\mapsto \vert \bar{\zeta}\rangle\langle \bar{\zeta}\vert\in\S(\H_{A_1...A_n})$ has the  Lipschitz constant equal to $2$.
Indeed, by using the well known expression for the scalar product of two "coherent" vectors $\vert \bar{\zeta_1}\rangle$ and  $\vert \bar{\zeta_2}\rangle$ (see, f.i., \cite[Section 12.1.3]{H-SCI}) it is
easy to show that
$$
\frac{1}{2}\|\vert\bar{\zeta_1}\rangle\langle \bar{\zeta_1}\vert-\vert \bar{\zeta_2}\rangle\langle \bar{\zeta_2}\vert\|_1\leq\beta(\vert\bar{\zeta_1}\rangle\langle \bar{\zeta_1}\vert,\vert \bar{\zeta_2}\rangle\langle \bar{\zeta_2}\vert)=\sqrt{2-2\exp(-\textstyle\frac{1}{2}\|\bar{\zeta_1}-\bar{\zeta_2}\|^2)}\leq\|\bar{\zeta_1}-\bar{\zeta_2}\|,
$$
where the first inequality follows from (\ref{B-d-s-r}), the second one -- from the inequality\break $1-1/x\leq\ln x$ and $\|\bar{\zeta_1}-\bar{\zeta_2}\|$ denotes the
Euclidean norm in $\mathbb{C}^n$.

Thus, we can apply the main results of this article to ensembles of  coherent states by using Lemma \ref{K-L} with $L=2$.
For example, Proposition \ref{main-2} implies the following\smallskip\pagebreak

\begin{corollary}\label{main-2-c-s} \emph{Let $\Phi$ be a channel from $A_1...A_n$ to $B$ and $\mu=\{\pi_{\mu},\vert \bar{\zeta}\rangle\langle \bar{\zeta}\vert\}_{\bar{\zeta}\in\mathbb{C}^n}$
an  ensemble of coherent states in $\S(\H_{A_1...A_n})$ determined by a Borel probability measure $\pi_{\mu}$ on $\mathbb{C}^n$. Let $H_B$ be a positive operator on $\H_B$ satisfying  conditions (\ref{H-cond})  and (\ref{star}). Let $F_{H_B}$ and $g$ be the functions defined, respectively, in (\ref{F-def}) and (\ref{g-def}).}\smallskip

\emph{If $\,\bar{E}_{H_B}^{\rm psv}(\Phi(\mu))=E_B<+\infty\,$ then\footnote{$\bar{E}_{H_B}^{\rm psv}(\Phi(\mu))$ is  the average output passive energy of the channel $\Phi$ at the ensemble $\mu$ w.r.t. the "Hamiltonian" $H_B$ defined in (\ref{OPE-def}).}
\begin{equation}\label{main-2-c-s+}
\overline{S}(\Phi(\mu))-\overline{S}(\Psi(\nu))\leq \varepsilon F_{H_B}(E_B/\varepsilon)+g(\varepsilon)
\end{equation}
for any  ensemble $\nu=\{\pi_{\nu},\vert \bar{\zeta}\rangle\langle \bar{\zeta}\vert\}_{\bar{\zeta}\in\mathbb{C}^n}$ of coherent states in $\S(\H_{A_1...A_n})$ determined by a probability measure $\pi_{\nu}$ on $\mathbb{C}^n$ and any channel $\Psi:A\to B$ such that
\begin{equation}\label{e-cond-2-s}
D_{KR}(\pi_\mu,\pi_\nu)+\textstyle\frac{1}{2}\|\Phi-\Psi\|_{Y}\leq\varepsilon,
\end{equation}
where $D_{KR}(\pi_\mu,\pi_\nu)$ is
the Kantorovich–Rubinshtein disctance between the measures $\pi_{\mu}$ and $\pi_{\nu}$, $\|\cdot\|_{Y}$ is either the diamond norm $\|\cdot\|_{\diamond}$ defined in (\ref{d-norm}) or its unstabilized version $\|\cdot\|_{1\to1}$ defined in (\ref{d-norm+}).}
\smallskip

\emph{The following replacements can be done independently on each other:}
\begin{enumerate}[(1)]
\item \emph{if $\pi_\mu$ and $\pi_\nu$ are measures from the set $\P_*(\mathbb{C}^n)$ (defined in  Appendix A-1) then the Kantorovich–Rubinshtein disctance $D_{KR}(\pi_\mu,\pi_\nu)$ in (\ref{e-cond-2-s}) can be replaced with the modified Kantorovich–Rubinshtein disctance $D_{KR}^*(\pi_\mu,\pi_\nu)$;}
\item  \emph{$E_B=\bar{E}_{H_B}^{\rm psv}(\Phi(\mu))$ in (\ref{main-2-c-s+}) can be replaced with $\,E_B=\Tr H_B\Phi(\bar{\rho}(\mu))$;} \smallskip
\item \emph{if either $\,\Tr H_A\bar{\rho}(\mu)=E_{A}<+\infty\,$ or $\,\Tr H_A\bar{\rho}(\nu)=E_{A}<+\infty$, where $H_A$ is a positive operator on $\H_A$, then the energy-constrained diamond norm $\|\cdot\|^{H_A}_{\diamond,E_A}$ defined in (\ref{ec-d-norm}) can be used in the role of $\|\cdot\|_{Y}$ in  (\ref{e-cond-2-s}).}
\end{enumerate}
\end{corollary}

\begin{example}\label{C-states+} Assume that $A$ is an one-mode quantum oscillator and $\Phi$ is a channel
from $A$ to $A$ defined as
$$
\Phi(\rho)=\sum_{n=0}^{+\infty}\langle n|\rho|n\rangle|n\rangle\langle n|,\quad \rho\in\S(\H_A),
$$
where $\{|n\rangle\}$ is the Fock basis in $A$ \cite[Ch.12]{H-SCI}.

Let $H_A=\hat{N}\doteq a^{\dagger}a$ be the number operator on $\H_A$ and
$\mu=\{\pi_{\mu},\vert \zeta\rangle\langle \zeta\vert\}_{\zeta\in\mathbb{C}}$ the  ensemble of coherent states determined by
the measure  $\pi_{\mu}$ on $\mathbb{C}$ with the density $p_{\mu}(\zeta)=\frac{1}{\pi N}\exp(-|\zeta|^2/N)$. Then $\bar{\rho}(\mu)=\gamma_{\hat{N}}(N)$ is
the Gibbs state corresponding to the
"energy" (the number of quanta) $N$ \cite[Ch.12]{H-SCI}. As the state $\gamma_{\hat{N}}(N)$ is diagonizable in the basis $\{|n\rangle\}$, we have $\Phi(\bar{\rho}(\mu))=\bar{\rho}(\mu)=\gamma_{\hat{N}}(N)$.

By using the well known expression for the scalar product of the vectors $\vert \zeta\rangle$ and  $\vert n\rangle$ (see, f.i., \cite[Section 12.1.3]{H-SCI})
we obtain
\begin{equation}\label{exp-1}
\Phi(\vert \zeta\rangle\langle \zeta\vert)=\sum_{n=0}^{+\infty}|\langle n|\zeta\rangle|^2|n\rangle\langle n|=\sum_{n=0}^{+\infty}\frac{|\zeta|^{2n}}{n!}\exp(-|\zeta|^2)|n\rangle\langle n|.
\end{equation}
It follows that
\begin{equation}\label{E-psv}
\bar{E}_{\hat{N}}^{\rm psv}(\Phi(\vert \zeta\rangle\langle \zeta\vert))\leq\Tr\hat{N}\Phi(\vert \zeta\rangle\langle \zeta\vert)=\sum_{n=0}^{+\infty}n\,\frac{|\zeta|^{2n}}{n!}\exp(-|\zeta|^2)=|\zeta|^{2}
\end{equation}
and hence
\begin{equation}\label{E-psv+}
\bar{E}_{\hat{N}}^{\rm psv}(\Phi(\mu))\leq\int_{\mathbb{C}}|\zeta|^2\pi_{\mu}(d^2\zeta)=\frac{1}{\pi N}\int_{\mathbb{C}}|\zeta|^2\exp(-|\zeta|^2/N)d^2\zeta=N.
\end{equation}
The last upper bound on $\bar{E}_{\hat{N}}^{\rm psv}(\Phi(\mu))$ coincides with the obvious upper bound $\Tr\hat{N}\Phi(\bar{\rho}(\mu))$. However, we will need this
integral expression below.

Expression (\ref{exp-1}) also implies that
\begin{equation}\label{exp-2}
 \overline{S}(\Phi(\mu))=\frac{1}{\pi N}\int_{\mathbb{C}}H_P(|\zeta|^2)d^2\zeta,
\end{equation}
where $\,H_P(\lambda)=\lambda(1-\ln\lambda)+e^{-\lambda}\sum_{n=0}^{+\infty}\frac{\lambda^n\ln n!}{n!}\,$ is the Shannon entropy of the
Poisson distribution with parameter $\lambda$.

Thus, since $F_{\hat{N}}(E)=g(E)$ -- the function defined in (\ref{g-def}), Corollary \ref{main-2-c-s} implies that
\begin{equation}\label{main-2-c-s+2}
\overline{S}(\Phi(\mu))-\overline{S}(\Psi(\nu))\leq \varepsilon g(N/\varepsilon)+g(\varepsilon)
\end{equation}
for any channel $\Psi:A\to A$ and any  ensemble $\nu=\{\pi_{\nu},\vert \zeta\rangle\langle \zeta\vert\}_{\bar{\zeta}\in\mathbb{C}}$ of coherent states in $\S(\H_{A})$ determined by a probability measure $\pi_{\nu}$ on $\mathbb{C}$ such that
\begin{equation*}
D_{KR}(\pi_\mu,\pi_\nu)+\textstyle\frac{1}{2}\|\Phi-\Psi\|_{Y}\leq\varepsilon,
\end{equation*}
where $D_{KR}(\pi_\mu,\pi_\nu)$ is
the Kantorovich–Rubinshtein distance between the measures  $\pi_{\mu}$ and $\pi_{\nu}$ on $\mathbb{C}$ and
$\|\cdot\|_{1\to1}$ is either the norm $\|\cdot\|_{1\to1}$ defined in (\ref{d-norm+})
or the energy-constrained  diamond norm
$\|\cdot\|^{\hat{N}}_{\diamond,N}$ defined in (\ref{ec-d-norm}).\smallskip

Below we show how to apply  the semicontinuity bound (\ref{main-2-c-s+2}) to estimate
the variation of the AOE of the channel $\Phi$ under  replacing the continuous ensemble $\mu$ with its discrete approximation
induced by discretization of the corresponding probability measure $\pi_\mu$ on $\mathbb{C}$.\smallskip

Take $\Delta>0$ and consider the probability measure $\pi^{\Delta}_{\mu}=\sum_{k,j\in\mathbb{Z}}p_{kj}\delta(\zeta_{kj})$ on $\mathbb{C}$, where
$p_{kj}=\pi_{\mu}(\Pi_{kj})$,
\begin{equation}\label{Pi-set}
\Pi_{kj}=\{\zeta\in\mathbb{C}\,|\,\Re\zeta\in(\Delta k,\Delta (k+1)],\Im\zeta\in(\Delta j,\Delta(j+1)]\}
\end{equation}
and $\delta(\zeta_{kj})$ is the Dirac (single atom) measure concentrating at the complex number $\,\zeta_{kj}=\Delta((k+1/2)+\mathrm{i}(j+1/2))$. The measure
$\pi^{\Delta}_{\mu}$ is a $\Delta$-discretization of the continuous measure $\pi_{\mu}$. Denote the ensemble of coherent states corresponding to the measure
$\pi^{\Delta}_{\mu}$ by $\mu_{\Delta}$. It can be called $\Delta$-discretization of the ensemble $\mu$.

It follows from the definition of
the Kantorovich–Rubinshtein distance (see formula (\ref{K-R}) in the Appendix A-1) that
\begin{equation}\label{K-R-D}
D_{KR}(\pi_\mu,\pi^{\Delta}_{\mu})\leq\frac{\Delta}{\sqrt{2}}.
\end{equation}
Thus, by taking $\Psi=\Phi$ and $\pi_{\nu}=\pi^{\Delta}_{\mu}$ we obtain from (\ref{main-2-c-s+2}) the inequality
\begin{equation}\label{D-bound}
\overline{S}(\Phi(\mu))-\overline{S}(\Phi(\mu_{\Delta}))\leq \frac{\Delta}{\sqrt{2}}\;g\!\left(\frac{\sqrt{2}N}{\Delta}\right)+g\!\left(\frac{\Delta}{\sqrt{2}}\right),
\end{equation}
which gives an upper bound on the possible loss of $\overline{S}(\Phi(\mu))$ under the $\Delta$-discretization of $\mu$. Note that this
bound tends to zero as $\Delta\to0$.

To obtain  an upper bound on the possible gain of $\overline{S}(\Phi(\mu))$ under the $\Delta$-discretization of $\mu$ one can use
(\ref{main-2-c-s+2}) with $\mu$ and $\nu$ replaced by $\mu_{\Delta}$ and $\mu$, respectively, and $\Psi=\Phi$. We have only to estimate $\bar{E}_{\hat{N}}^{\rm psv}(\Phi(\mu_{\Delta}))$.

Introduce the function $f(\zeta)=\sum_{k,j\in\mathbb{Z}}|\zeta_{kj}|^2\mathbb{I}_{\Pi_{kj}}(\zeta)$ on $\mathbb{C}$, where $\mathbb{I}_{\Pi_{kj}}$ is the indicator
function of the set $\Pi_{kj}$ defined in (\ref{Pi-set}). It follows from (\ref{E-psv}) that
$$
\bar{E}_{\hat{N}}^{\rm psv}(\Phi(\mu_{\Delta}))\leq\int_{\mathbb{C}}|\zeta|^2\pi^{\Delta}_{\mu}(d^2\zeta)=\frac{1}{\pi N}\int_{\mathbb{C}}f(\zeta)\exp(-|\zeta|^2/N)d^2\zeta.
$$
Comparing this with (\ref{E-psv+}) we have
$$
\bar{E}_{\hat{N}}^{\rm psv}(\Phi(\mu_{\Delta}))\leq N+\frac{1}{\pi N}\int_{\mathbb{C}}(f(\zeta)-|\zeta|^2)\exp(-|\zeta|^2/N)d^2\zeta.
$$
To estimate the last integral note that $f(\zeta)-|\zeta|^2\leq \frac{\Delta}{\sqrt{2}}(2|\zeta|+\frac{\Delta}{\sqrt{2}})=\sqrt{2}\Delta|\zeta|+\frac{\Delta^2}{2}$ for any
$\zeta\in\mathbb{C}$ by the construction of $f$. Hence
$$
\bar{E}_{\hat{N}}^{\rm psv}(\Phi(\mu_{\Delta}))\leq N+\frac{\Delta^2}{2}+\frac{\sqrt{2}\Delta}{\pi N}\int_{\mathbb{C}}|\zeta|\exp(-|\zeta|^2/N)d^2\zeta=N+\frac{\Delta^2}{2}+\Delta\sqrt{\frac{\pi N}{2}}.
$$
Thus, by using (\ref{K-R-D}) and the semicontinuity bound (\ref{main-2-c-s+2}) with $\mu$ and $\nu$ replaced by $\mu_{\Delta}$ and $\mu$, respectively, and $\Psi=\Phi$
we obtain the inequalily
\begin{equation*}
\overline{S}(\Phi(\mu_{\Delta}))-\overline{S}(\Phi(\mu))\leq \frac{\Delta}{\sqrt{2}}\;g\!\left(\frac{\sqrt{2}N}{\Delta}+\frac{\Delta}{\sqrt{2}}+\sqrt{\pi N}\right)+g\!\left(\frac{\Delta}{\sqrt{2}}\right)
\end{equation*}
which complements (\ref{D-bound}) and gives an upper bound on the possible gain  of $\overline{S}(\Phi(\mu))$ under the $\Delta$-discretization of $\mu$. Note that this
bounds also tends to zero as $\Delta\to0$.
\end{example}

\section*{Appendix}

\subsection*{A-1 Kantorovich–Rubinshtein distance and the modified Kantorovich–Rubinshtein distance between probability measures}

Assume that $X$ is a metric space with a metric $d$. The Kantorovich–Rubinshtein distance
between Borel probability measures $\mu$  and $\nu$ on $X$ is defined as (cf.\cite{Bog,B&K})
\begin{equation}\label{K-R}
  D_{KR}(\mu,\nu)=\sup\left\{\left.\int_X f(x)(\mu-\nu)(dx)\,\right|\,f\in \mathrm{Lip}_1(X),\; \sup_{x\in X}|f(x)|\leq 1\right\},
\end{equation}
where $\mathrm{Lip}_1(X)$ is the set of all real valued functions on $X$ with the Lipschitz constant not exceeding $1$, i.e.
$$
\mathrm{Lip}_1(X)=\left\{f:X\to\mathbb{R}\,|\, |f(x)-f(y)|\leq d(x,y)\;\,\forall x,y\in X\right\}.
$$
If $X$ is a complete separable metric space then the metric $D_{KR}$ generates the weak
convergence topology on the set $\P(X)$ of all Borel probability measures on $X$, i.e. a sequence
of $\{\mu_n\}\subset\P(X)$ converges  to a measure $\mu_0\in\P(X)$ w.r.t. the metric $D_{KR}$ if and only if
$$
\lim_{n\rightarrow+\infty}\int_X f(x)\mu_n(dx)=\int_X f(x)\mu_0(dx)
$$
for any continuous bounded function $f$ on $\,X$ \cite{Bog,B&K}. The set $\P(X)$ equipped
with the metric $D_{KR}$ is a complete separable metric space in this case \cite{Bil+,Bog}.

Let $\P_*(X)$  be the subset of $\P(X)$ consisting of all  measures for which the
function $x\mapsto d(x,x_0)$ is integrable for any $x_0\in X$ (by the triangle inequality it suffices
to check this condition for some $x_0\in X$). Note that $\P_*(X)=\P(X)$ provided that $X$ has a finite diameter.

The  modified Kantorovich–Rubinshtein distance between measures $\mu$  and $\nu$ in $\P_*(X)$
is defined as (cf.\cite{Bog,B&K})
\begin{equation}\label{K-R-m}
  D^*_{KR}(\mu,\nu)=\sup\left\{\left.\int_X f(x)(\mu-\nu)(dx)\,\right|\,f\in \mathrm{Lip}_1(X)\right\}.
\end{equation}
This definition differs from (\ref{K-R})  by omitting the condition  $\sup_{x\in X}|f(x)|\leq 1$. Hence
$$
D_{KR}(\mu,\nu)\leq D^*_{KR}(\mu,\nu)\quad \forall \mu,\nu\in \P_*(X)
$$
in general. But if the diameter of  $X$ is not greater than $1$ then $\sup_{x\in X}|f(x)|\leq 1$ for any
$f\in \mathrm{Lip}_1(X)$, so in this case $D_{KR}(\mu,\nu)=D^*_{KR}(\mu,\nu)$ for all $\mu,\nu\in \P_*(X)=\P(X)$.

\smallskip

If $\mu$  and $\nu$ are Radon measures\footnote{A Borel probability measure
$\mu$ is called  Radon measure if for every Borel set $B$ and every $\epsilon>0$, there exists
a compact set $K\subset B$ such that $\mu(B\setminus K)<\epsilon$. If $X$ is a complete separable metric space then every Borel probability measure is a Radon measure \cite{Bog}.} in $\P_*(X)$ then the  modified Kantorovich–Rubinshtein distance between $\mu$  and $\nu$ can be also defined by the dual expression
\begin{equation}\label{K-R-m+}
  D^*_{KR}(\mu,\nu)=\inf\left\{\left.\int_{X\times X} d(x-y)\Lambda(dxdy)\,\right|\,\Lambda\in\Pi(\mu,\nu)\right\},
\end{equation}
where $\Pi(\mu,\nu)$ is the set of all Radon probability measures on $X\times X$ with the marginals $\mu$ and $\nu$ \cite{Bog,B&K}.

The quantity in the r.h.s. of (\ref{K-R-m+}) was introduced by Kantorovich and is called the Kantorovich distance \cite{Bog}.

Assume that $X$ is the complete separable metric space $\S(\H)$ with the metric $d(\rho,\sigma)=\frac{1}{2}\|\rho-\sigma\|_1$
and the diameter equal to $1$. Then the above observations show that the modified Kantorovich–Rubinshtein distance is well defined in (\ref{K-R-m}) for all Borel probability measures on $X=\S(\H)$ and can be also defined by the dual expression
(\ref{K-R-m+}) for all such measures. Moreover, it coincides with the Kantorovich–Rubinshtein distance  defined in (\ref{K-R}). By this reason, in the article we use the term  "Kantorovich distance"
speaking about the metrics defined in  (\ref{K-D-d}) and (\ref{K-D-c}).

\subsection*{A-2 The proof of the coincidence of $D_*$ and $D_{\mathrm{ehs}}$}

Below we give a rigorous proof of  the coincidence of the metrics  $D_*$ and $D_{\mathrm{ehs}}$
between discrete ensembles of quantum states, defined, respectively, in (\ref{f-metric}) and  (\ref{ehs-metric}).
This coincidence (mentioned  in \cite{CHI} as an "obvious fact") is essentially used in the article. \smallskip

Let $\mu=\{p_i,\rho_i\}_{i\in I}$ and $\nu=\{q_j,\sigma_j\}_{j\in J}$. If $\{P_{ij}\}_{(i,j)\in I\times J}$ and $\{Q_{ij}\}_{(i,j)\in I\times J}$ are arbitrary 2-variate probability distributions
such that $\sum_{j\in J}P_{ij}=p_i$ for all $i\in I$ and $\sum_{i\in I}Q_{ij}=q_j$ for all $j\in J$ then it is easy to
see that the ordered ensembles $(P_{ij},\rho_i)_{(i,j)\in I\times J}$ and  $(Q_{ij},\sigma_j)_{(i,j)\in I\times J}$
belong, respectively, to the classes $\mathcal{E}(\mu)$ and $\mathcal{E}(\nu)$ introduced after (\ref{f-metric}).
By definitions (\ref{f-metric}) and (\ref{ehs-metric}) this implies $D_*(\mu,\nu)\leq D_{\mathrm{ehs}}(\mu,\nu)$.

Thus, we have to show that $D_*(\mu,\nu)\geq D_{\mathrm{ehs}}(\mu,\nu)$. This can be done assuming that the representations $\mu=\{p_i,\rho_i\}_{i\in I}$ and $\nu=\{q_j,\sigma_j\}_{j\in J}$ of arbitrary discrete ensembles $\mu$ and $\nu$ are \emph{minimal} in the sense that $\rho_i\neq\rho_{i'}$ for all $i\neq i'$ and $\sigma_j\neq\sigma_{j'}$ for all $j\neq j'$. This minimality assumption is acceptable because
the set of all representations of a discrete ensemble (probability measure on $\S(\H)$) contains  a minimal representation and the
quantities $D_*(\mu,\nu)$ and $D_{\mathrm{ehs}}(\mu,\nu)$ depend only on the discrete measures $\mu$ and $\nu$ not depending on their concrete
representations.

Let $\mu=\{p_i,\rho_i\}_{i\in I}$ and $\nu=\{q_j,\sigma_j\}_{j\in J}$ be  minimal representations of $\mu$ and $\nu$. For any $\epsilon>0$
there exists ordered ensembles $(p_k,\rho_k)_{k\in K}\in \mathcal{E}(\mu)$ and $(q_k,\sigma_k)_{k\in K}\in \mathcal{E}(\nu)$ such that
\begin{equation}\label{A-2-1}
D_0((p_k,\rho_k),(q_k,\sigma_k))\leq D_*(\mu,\nu)+\epsilon.
\end{equation}
For each $(i,j)\in I\times J$ let $A_{ij}$ be  the set of all $k$ such that $\rho_k=\rho_i$
and $\sigma_k=\sigma_j$. Let $P_{ij}=\sum_{k\in A_{ij}}p_k$ and
$Q_{ij}=\sum_{k\in A_{ij}}q_k$. By the minimality assumption $\{A_{ij}\}_{(i,j)\in I\times J}$ is a collection of mutually
disjoint sets such that $\bigcup_{(i,j)\in I\times J}A_{ij}=K$. Hence
\begin{equation}\label{A-2-2}
\begin{array}{c}
\displaystyle\sum_{(i,j)\in I\times J}\|\shs P_{ij}\rho_i-Q_{ij}\sigma_j\|_1=\sum_{(i,j)\in I\times J}\left\|\shs \sum_{k\in A_{ij}}(p_k\rho_k-q_k\sigma_k)\right\|_1\\\\\displaystyle\leq\sum_{(i,j)\in I\times J}\sum_{k\in A_{ij}}\|\shs p_k\rho_k-q_k\sigma_k\|_1
\\\\\displaystyle=\sum_{k\in K}\|\shs p_k\rho_k-q_k\sigma_k\|_1=2D_0((p_k,\rho_k),(q_k,\sigma_k)).
\end{array}
 \end{equation}
Since $\sum_{j\in J}P_{ij}=p_i$ for all $i\in I$ and $\sum_{i\in I}Q_{ij}=q_j$
for all $j\in J$ by the construction, it follows from (\ref{A-2-1}) and (\ref{A-2-2}) along with (\ref{ehs-metric}) that $D_{\mathrm{ehs}}(\mu,\nu)\leq D_*(\mu,\nu)$.\smallskip

\bigskip

I am grateful to A.S.Holevo, E.R.Loubenets and G.G.Amosov for valuable discussion.

\end{document}